%% file: main.tex
\documentclass[12pt]{article}

\usepackage{bm, commath}
\usepackage{natbib}
\usepackage{caption}
\usepackage{graphicx}
\usepackage{subfig}
\usepackage{amsmath, amsfonts, amsthm}
\usepackage{float}
\usepackage{booktabs,siunitx}
\usepackage{url}
\usepackage{multirow}
\usepackage{amssymb}
\usepackage{blkarray}
\usepackage{bbm}
\usepackage{multicol}
\usepackage{authblk}
\usepackage{natbib}
\usepackage[shortlabels]{enumitem}
\usepackage{longtable}
\usepackage{adjustbox}
\usepackage{threeparttable}
\usepackage{hhline}
\usepackage{colortbl}
\usepackage{lscape}

\usepackage{listings}
\usepackage{bbm}
\usepackage{textcomp}
\usepackage{authblk}
\usepackage[ruled]{algorithm2e}
\SetKwInput{KwParam}{Parameter}
\SetAlgoCaptionLayout{centerline}

\usepackage{sectsty}
\usepackage[doublespacing]{setspace}

\usepackage[utf8]{inputenc}
\usepackage{float}
\usepackage{tikz}
\usetikzlibrary{positioning,chains,fit,shapes,calc}

\usepackage{mwe}
\usepackage{graphicx}

\def\bs{\boldsymbol}
\def\I{\mathbbm{1}}
\usepackage[textsize=scriptsize, textwidth = 2cm, shadow]{todonotes}

\newtheorem{theorem}{Theorem}
\newtheorem{definition}{Definition}

\newtheorem{remark}{Remark}

\newtheorem{assumption}{Assumption}
\newtheorem*{assumption*}{Assumption}

\newtheorem{proposition}{Proposition}
\newtheorem{condition}{Condition}
\newtheorem{lemma}{Lemma}

\newtheorem{innercustomgeneric}{\customgenericname}
\providecommand{\customgenericname}{}
\newcommand{\newcustomtheorem}[2]{%
  \newenvironment{#1}[1]
  {%
   \renewcommand\customgenericname{#2}%
   \renewcommand\theinnercustomgeneric{##1}%
   \innercustomgeneric
  }
  {\endinnercustomgeneric}
}

\newcustomtheorem{customthm}{Theorem}
\newcustomtheorem{customassumption}{Assumption}

\usepackage{mathtools}

\usepackage{xcolor}

\makeatletter
\renewcommand{\algocf@captiontext}[2]{#1\algocf@typo. \AlCapFnt{}#2} 
\renewcommand{\AlTitleFnt}[1]{#1\unskip}
\def\@algocf@capt@plain{top}
\renewcommand{\algocf@makecaption}[2]{%
  \addtolength{\hsize}{\algomargin}%
  \sbox\@tempboxa{\algocf@captiontext{#1}{#2}}%
  \ifdim\wd\@tempboxa >\hsize
  \hskip .5\algomargin%
  \parbox[t]{\hsize}{\algocf@captiontext{#1}{#2}}
  \else%
  \global\@minipagefalse%
  \hbox to\hsize{\box\@tempboxa}
  \fi%
  \addtolength{\hsize}{-\algomargin}%
}
\makeatother

\begin{document}

\def\spacingset#1{\renewcommand{\baselinestretch}%
{#1}\small\normalsize} \spacingset{1}

\sectionfont{\bfseries\large\sffamily}%

\subsectionfont{\bfseries\sffamily\normalsize}%

\title{{ Design-based nested instrumental variable analysis}
}

\author[1]{Zhe Chen}
\author[2]{Xinran Li}
\author[1]{Michael O. Harhay}
\author[3]{Bo Zhang\thanks{Assistant Professor of Biostatistics, Vaccine and Infectious Disease Division, Fred Hutchinson Cancer Center. Email: {\tt bzhang3@fredhutch.org}. }}

\affil[1]{Department of Biostatistics, Epidemiology and Informatics, University of Pennsylvania, Pennsylvania, USA}
\affil[2]{Department of Statistics, University of Chicago, Illinois, USA}
\affil[3]{Vaccine and Infectious Disease Division, Fred Hutchinson Cancer Center, Washington, USA}

\date{}
\maketitle

\begin{abstract}
\input{abstract}
\end{abstract}
\vspace{0.3 cm}

\newpage
\spacingset{1.6}

\setlength\abovedisplayskip{2pt}
\setlength\belowdisplayskip{1pt}%

\input{1_Introduction}
\input{2_Design}
\input{3_Notation}
\input{4_Inference}
\input{5_Simulation}
\input{6_Case_study}

\input{7_Discussion}
\onehalfspacing
\bibliographystyle{apalike}
\bibliography{paper-ref}
\input{supp}

\end{document}

%% file: abstract.tex
Two binary instrumental variables (IVs) are \emph{nested} if individuals who comply under one binary IV also comply under the other. This situation often arises when the two IVs represent different intensities of encouragement or discouragement to take the treatment---one stronger than the other. In a nested IV structure, treatment effects can be identified for two latent subgroups: always-compliers and switchers. Always-compliers are individuals who comply even under the weaker IV, while switchers are those who do not comply under the weaker IV but do under the stronger IV. We introduce a novel pair-of-pairs nested IV (PoP-NIV) design, where each matched stratum consists of four units organized in two pairs. We develop design-based inference for the always-complier sample average treatment effect (ACO-SATE) and switcher sample average treatment effect (SW-SATE). In a nested IV analysis, IV assignment is randomized within each IV pair; however, whether a study unit receives the weaker or stronger IV may not be randomized. To address this complication, we then propose a novel \emph{partly biased randomization scheme} and study design-based inference under this new scheme. Using extensive simulation studies, we demonstrate the validity of the proposed method even in challenging scenarios with small sample sizes and a low proportion of switchers. Applying the nested IV framework, we estimated that 52.2\% (95\% CI: 50.4\%–53.9\%) of participants enrolled at the Henry Ford Health System in the Prostate, Lung, Colorectal, and Ovarian (PLCO) Cancer Screening Trial were always-compliers, while 26.7\% (95\% CI: 24.5\%–28.9\%) were switchers. Among always-compliers, flexible sigmoidoscopy was associated with a trend toward a decreased colorectal cancer rate ($-0.62\%$; 95\% CI: $-1.90\%$ to $0.65\%$). No effect was detected among switchers ($1.46\%$; 95\% CI: $-1.97\%$ to $4.91\%$). The analysis offers a richer interpretation of why no increase in the intention-to-treat effect was observed after $1997$, even though the compliance rate rose.


%% file: 1_Introduction.tex
\section{Introduction}
\label{sec: intro}

\subsection{Design-based instrumental variable analysis}\label{subsec: design based IV analysis}

In many practical scenarios, the treatment and outcome of interest may be confounded by unmeasured variables. For example, in studies assessing the impact of educational attainment on future earnings, unobserved factors such as ambition and perseverance likely influence both educational attainment and earnings. These latent traits are difficult, if not impossible, to measure, posing significant challenges for valid causal inference. An instrumental variable (IV) provides a strategy to address this issue by leveraging a source of variation in the treatment that is independent of unmeasured confounders and affects the outcome only through its effect on the treatment. While direct randomization of the treatment may be infeasible in many studies, random variation in encouragement or eligibility for treatment can sometimes be exploited towards conducting valid causal inference. A classic example is provided by \citet{angrist1991does}, who used individuals’ quarter of birth as an IV to estimate the causal effect of schooling on earnings. Quarter of birth was shown to influence educational attainment---albeit modestly---and is assumed to affect earnings only through its impact on education. Crucially, unlike years of schooling, quarter of birth is essentially random.

A valid binary instrumental variable---to be formally defined later---can non-parametrically identify the treatment effect within a latent subgroup known as compliers \citep{angrist1996identification}. Furthermore, under additional assumptions that limit treatment effect heterogeneity, a valid IV may enable partial or even point identification of the treatment effect in the overall population. For a comprehensive discussion of such identification assumptions, see \citet{swanson2018partial}. Practical strategies for identifying valid IVs are reviewed in \citet{angrist2001instrumental}.

In a thought-provoking piece, \citet{imbens2005robust} introduced a randomization inference-based approach to instrumental variable analysis. Unlike traditional methods grounded in structural equation modeling and moment conditions, \citeauthor{imbens2005robust}'s approach \citeyearpar{imbens2005robust} enables exact hypothesis testing for a class of sharp null hypotheses. Notably, it preserves valid inference even in challenging settings where the IV is only weakly correlated with the treatment---circumstances under which conventional methods such as two-stage least squares (2SLS) are known to produce biased results \citep{bound1995problems, imbens2005robust}.

This randomization inference framework for IV analysis has since been extended in multiple directions. These include factorial designs with noncompliance \citep{blackwell2023noncompliance}; binary IVs that achieve validity only after adjustment for baseline covariates via matching \citep{small2008war}; continuous IVs and weak null hypotheses such as the sample average treatment effect (SATE) among compliers \citep{baiocchi2010building}; cluster-randomized encouragement designs with individual-level noncompliance \citep{zhang2022bridging}; biased randomization schemes \citep{small2008war, fogarty2021biased}; and approaches to bounding the SATE of the entire study cohort \citep{chen2024manipulating}.

\subsection{Nested IV: people respond to stronger incentives}
\label{subsec: intro nested IV}
An important characteristic of an instrumental variable is its compliance rate, defined as the difference in the proportion of participants who receive treatment when encouraged versus when discouraged. Compliance rates can vary across studies, and an individual participant’s compliance behavior may differ depending on which encouragement is applied. Recently, \citet{wang2024nested} introduced the concept of nested IVs: a binary IV is nested within another binary IV if two IVs differ in their induced compliance rates, and a complier under the weaker IV remains a complier under the stronger IV. Notably, the nested IV assumption does not impose additional restrictions on latent subgroup membership.

The classical monotonicity assumption \citep{angrist1996identification} is a statement about cross-world potential outcomes---specifically, the potential treatment received under two levels of a single binary instrumental variable. These outcomes cannot be observed simultaneously for the same individual, hence the term ``cross-world." Similarly, the nested IV assumption concerns the structure of four potential treatment-received outcomes defined by two binary IVs. Like other cross-world assumptions, it cannot be directly verified from data; however, in many applications, researchers have reasonable grounds, often informed by economic theory or behavioral science, to believe such assumptions hold. For example, in a multicenter clinical trial, large differences in the compliance rates across sites or time periods may support the assumption that a complier under a low-compliance-inducing IV would remain a complier under a high-compliance-inducing IV. These differences in the compliance rates are often attributable to identifiable factors; in the Prostate, Lung, Colorectal and Ovarian (PLCO) trial, for instance, the compliance rate at the Henry Ford Health System increased sharply from $50\%$ to $90\%$ after the site switched from a dual consent to a single consent process \citep{marcus2014non}. Nested IV structures also arise in policy research. For example, \citet{hull2018isolateing} and \citet{guo2025rationed} used twinning (bearing two children in one delivery) at the second birth as an IV for realized fertility before and after the implementation of a birth control policy. Conditional on the policy status (absence or presence of birth control), a complier is defined as a mother who would have two children if the second birth is a singleton and three children if the second birth is a twin. An always-taker, by contrast, is a mother who would always have three children regardless of twinning at the second birth. The nested IV assumption states that a complier in the absence of the birth control policy would remain a complier under the policy. Indeed, \citet{guo2025rationed} made this assumption in their empirical studies, although they did not formally frame the problem under a nested IV structure.

While conventional instrumental variable analysis identifies the treatment effect among compliers, an analysis using a pair of nested IVs can further distinguish treatment effects among two subgroups: always-compliers, those who comply under both IVs, and switchers, those who shift from non-compliers under a weaker IV to compliers under a stronger IV \citep{wang2024nested}. These subgroups are particularly relevant for policymaking because always-compliers represent individuals who are consistently responsive to encouragement, whereas switchers are those whose treatment adherence could be influenced through stronger incentives or interventions. For instance, in a clinical trial setting, switchers may represent participants whose treatment uptake could plausibly be improved through enhanced or more intensive encouragement strategies. Estimating treatment effects within these subgroups can inform implementation science, a key priority at the National Institutes of Health. For example, if switchers demonstrate a substantial treatment effect, this would support more proactive efforts to enhance adherence and ensure that evidence-based interventions reach those most likely to benefit. Moreover, the nested IV framework allows rigorously testing whether the treatment effect among the always-compliers equals that among the switchers, thus further yielyding insights for effect heterogeneity.


\subsection{A design-based framework for nested IV analysis}
\label{subsec: intro our contribution}
In this paper, we develop a design-based framework for nested IV analysis that has three main components. First, we introduce a pair-of-pairs design that mimics a finely stratified experiment under a hierarchical, two-stage randomization scheme. For instance, in a multicenter clinical trial, participants are first randomized to a clinical site and then, within each site, randomized to receive either treatment or control. We construct such designs by efficiently solving a network flow problem on a tripartite network \citep{karmakar2019using,zhang2023matchingone}. 

Second, we develop design-based inference for (1) the sample average treatment effect (SATE) among always-compliers (ACO-SATE); (2) the SATE among switchers (SW-SATE); and (3) whether ACO-SATE equals SW-SATE, all under the proposed pair-of-pairs design. 

Third, in our case study, the assignment of participants to one IV pair versus the other, such as attending a clinical site in Seattle versus Portland, need not be randomized, even though the IV assignment within a site is randomized; see \citet{stuart2011use} for a related discussion on trial selection probability. To address this setting, we introduce a novel \emph{partly biased randomization scheme} and study design-based inference under this scheme.

Unlike the approach of \citet{wang2024nested}, inference developed in this article relies solely on the IV assignment mechanism and does not require modeling any component of the potential outcomes. 
Through simulation studies and a case study, we show that this design-based method remains robust even in challenging scenarios, such as small sample sizes, low proportions of switchers, or binary endpoints with few observed events.

The rest of the article is organized as follows. Section \ref{sec: design} describes the proposed pair-of-pairs design and a computationally efficient method to achieve the design. Section \ref{sec: notation and estimands} describes the potential outcomes and estimands of interest. Design-based inference under randomization and a novel partly biased randomization scheme is studied in Section \ref{sec: inference}. We present simulation studies in Section \ref{sec: simulation} and apply the proposed design and analysis tools to the PLCO trial in Section \ref{sec: case study}. Section \ref{sec: discussion} concludes with a discussion. Technical proofs can be found in the Supplemental Materials, and code accompanying the article can be found via the following link: \url{https://github.com/Zhe-Chen-1999/Design-Based-Nested-IV}.

%% file: 2_Design.tex
\section{Study design}
\label{sec: design}

\subsection{Pair-of-pairs nested IV design}
\label{subsec: study design pop-IV}
Our proposed pair-of-pairs nested IV (PoP-NIV) design consists of $I$ matched strata, indexed by $i = 1, \dots, I$, each consisting of two matched pairs, indexed by $j = 1, 2$. Each matched pair $j$ in stratum $i$ consists of two units indexed by $k$. Two units in one pair in the matched stratum are assigned one binary IV, denoted as $\{0_a, 1_a\}$, and those in the other pair are assigned the other binary IV, denoted as $\{0_b, 1_b\}$. We use $ijk$, $i = 1,\dots,I$, $j = 1,2$, $k = 1,2$, to index the $k$-th participant in the matched pair $j$ in the $i$-th stratum. 

A matching algorithm $\textsf{M}$ embeds the observational data into a hypothetical, two-stage randomized encouragement experiment as follows. 
Within each matched stratum $i$, one pair is randomized to the IV dose pair $(0_a, 1_a)$ and the other to $(0_b, 1_b)$. Then, within the pair randomized to $(0_a, 1_a)$, two units are further randomized to $0_a$ or $1_a$. Analogously, within the pair randomized to $(0_b, 1_b)$, two units are further randomized to $0_b$ or $1_b$. Let $\boldsymbol{Z}_i = (Z_{i11}, Z_{i12}, Z_{i21}, Z_{i22})$ denote the IV assignment vector of stratum $i$ and 
\begin{equation*}
    \begin{split}
        \Omega = \{(1_a, 0_a, 1_b, 0_b), (0_a, 1_a, 1_b, 0_b), (1_a, 0_a, 0_b, 1_b), (0_a, 1_a, 0_b, 1_b), \\
        (1_b, 0_b, 1_a, 0_a), ( 1_b, 0_b, 0_a, 1_a), (0_b, 1_b, 1_a, 0_a), (0_b, 1_b, 0_a, 1_a)\}
    \end{split}
\end{equation*}
denote the set of $2^3 = 8$ possible configurations of the IV assignment vector $\boldsymbol{Z}_i$. Finally, let $\boldsymbol{x}_{ijk}$ denote participant $ijk$'s observed covariates and $\boldsymbol{X}$ the covariate matrix for all $4I$ participants. 

In design-based IV analysis, the probability of IV assignment provides the sole basis of statistical inference. Randomization assumption is often a starting point for analysis \citep{rosenbaum2002observational,imbens2005robust,small2008war}, followed by a biased randomization scheme that allows the IV to deviate from the randomization scheme by at most some controlled amount \citep{baiocchi2010building, fogarty2021biased}. The randomization assumption in a finite-sample, pair-of-pairs nested IV design can be formally stated as follows:

\begin{assumption}\label{ass: Randomization Assumption}
    \textbf{(Randomization Assumption in a PoP-NIV Design).} IV assignments across matched strata are assumed to be independent of each other, with $\Pr(\boldsymbol{Z_i} = \boldsymbol{a} \mid \boldsymbol{X}, \textsf{M}) = 1/8$, for all $\boldsymbol{a} \in \Omega$. 
\end{assumption}

In an integrated analysis of RCTs from multiple centers or studies, the randomization assumption may be violated, for instance, if there exist factors that simultaneously affect outcomes of interest and participants' selection into different study centers (e.g., choosing to enroll at a clinical site in Seattle not Portland) or different studies (e.g., choosing to participate in one study not the other). We will discuss and study how to relax the randomization assumption in Section \ref{subsec: biased randomization scheme}.

\subsection{PoP-NIV design in a study of colorectal cancer screening}
\label{subsec: PLCO match}
We illustrate the proposed PoP-NIV design using the PLCO Cancer Screening Trial. Supplemental Material C summarizes data of study participants who enrolled at the Henry Ford Health System before and after $1997$ where the trial consent policy underwent a major shift, and the compliance rate was observed to increase massively from $50\%$ to $80\%$ \citep{marcus2014non,wang2024nested}. Table \ref{table: cov balance after matching} summarizes the PoP-NIV design applied to the PLCO trial data. The design consists of $I = 3071$ matched strata, where each stratum consists of two pairs: in the first pair, two participants enrolled in the study prior to $1997$ and were each randomized to the treatment group (cancer screening) or the control group (standard of care), while in the second pair, two participants enrolled in the study after $1997$ and were each randomized to the treatment group or the control group. The match was constructed using an efficient, network-flow-based algorithm in less than $2$ minutes on a standard laptop, the details of which are to be elaborated on in the next subsection.

\begin{table}[ht]
\centering
\caption{Covariate balance after matching. Mean (SD) are reported for continuous variables; Count (percentage) are reported for categorical variables.}
\label{table: cov balance after matching}
\resizebox{\textwidth}{!}{
\begin{tabular}{lllllll}
\toprule
& \multicolumn{2}{c}{Prior to 1997 ($G= a$)} & \multicolumn{2}{c}{After 1997 ($G= b$)}& &\\
\toprule
 & Control arm & Treatment arm & Control arm & Treatment arm & $P$-value &  \\ 
 \toprule
Sample size&  3071 &  3071 &  3071 &  3071 &  &  \\ 
 Treatment uptake (=Screening) (\%) &     0 ( 0.0)  &  1602 (52.2)  &     0 ( 0.0)  &  2422 (78.9)  & $<$0.001 &  \\ 
   Age (\%) &   &   &   &   &  1.000 &  \\ 
     \quad $\leq 60$ &  1031 (33.6)  &  1031 (33.6)  &  1031 (33.6)  &  1031 (33.6)  &  &  \\ 
 \quad $(60,65]$ &  1092 (35.6)  &  1092 (35.6)  &  1092 (35.6)  &  1092 (35.6)  &  &  \\ 
 \quad $(65, 70]$ &   650 (21.2)  &   650 (21.2)  &   650 (21.2)  &   650 (21.2)  &  &  \\ 
   \quad  $>70$ &   298 ( 9.7)  &   298 ( 9.7)  &   298 ( 9.7)  &   298 ( 9.7)  &  &  \\ 
  Age & 63.31 (4.84) & 63.33 (4.85) & 63.08 (4.98) & 63.09 (4.98) &  0.081 &  \\ 
   Sex (=Male) (\%) &  1185 (38.6)  &  1178 (38.4)  &  1195 (38.9)  &  1190 (38.7)  &  0.975 &  \\ 
  Race (=Minority) (\%) &   520 (16.9)  &   518 (16.9)  &   496 (16.2)  &   510 (16.6)  &  0.841 &  \\ 
   Education (\%) &   &   &   &   &  0.381 &  \\ 
   \quad No high school &   334 (10.9)  &   336 (10.9)  &   297 ( 9.7)  &   295 ( 9.6)  &  &  \\ 
   \quad  High school &  1192 (38.8)  &  1182 (38.5)  &  1173 (38.2)  &  1189 (38.7)  &  &  \\ 
 \quad  College or above &  1545 (50.3)  &  1553 (50.6)  &  1601 (52.1)  &  1587 (51.7)  &  &  \\ 
  Smoking status (\%) &   &   &   &   &  0.948 &  \\ 
  \quad  Non-smoker &  1309 (42.6)  &  1304 (42.5)  &  1328 (43.2)  &  1332 (43.4)  &  &  \\ 
 \quad  Current smoker &   415 (13.5)  &   401 (13.1)  &   408 (13.3)  &   418 (13.6)  &  &  \\ 
   \quad  Former smoker  &  1347 (43.9)  &  1366 (44.5)  &  1335 (43.5)  &  1321 (43.0)  &  &  \\ 
   BMI  ($>25$) (\%) &  2095 (68.2)  &  2102 (68.4)  &  2135 (69.5)  &  2096 (68.3)  &  0.654 &  \\ 
  \toprule
\end{tabular}}
\end{table}
We assessed covariate balance across the four groups using standard multiple-group comparison tests: $\chi^2$ tests for categorical variables and one-way analysis of variance (ANOVA) for continuous variables. We highlight two key observations. First, the compliance rate---defined as the difference between the proportion of treatment-group participants receiving their first scheduled flexible sigmoidoscopy screening and the corresponding proportion in the control group---was markedly lower among participants enrolled before 1997 compared with those enrolled afterward (52.2\% vs. 78.9\%). Second, for all covariates listed in Table \ref{table: cov balance after matching}, no statistically significant differences were detected at the conventional $0.05$ level, indicating that the design adequately balanced observed covariates, emulating the target randomized controlled trial. Nonetheless, as in any matched design of observational data, balance on unmeasured covariates cannot be guaranteed. In particular, unobserved factors related to clinical outcomes and to participants' decision to enroll before versus after $1997$ may still be imbalanced, in contrast to what would be expected in a genuine randomized controlled trial.

\subsection{A computationally efficient algorithm constructing matched strata}
\label{subsec: match and computation}
The design in Table \ref{table: cov balance after matching} was constructed using a modified network-flow-based tripartite matching algorithm \citep{karmakar2019using,zhang2023matchingone}. The matching structure is represented as a four-layer network, in addition to a source code $\xi$ and a sink node $\overline{\xi}$, as illustrated in Figure \ref{fig: network flow}. Four layers of nodes in the network consist of (1) nodes $\{\tau^a_1, \dots, \tau^a_{T_a}\}$ corresponding to $|T_a|$ units assigned to the IV dose $Z=1_a$, (2) nodes $\{\gamma^a_1, \dots, \gamma^a_{C_a}\}$ corresponds to $|C_a|$ units assigned to $Z = 0_a$, (3) nodes $\{\tau^b_1, \dots, \tau^b_{{T}_b}\}$ corresponding to $|T_b|$ units assigned to $Z = 1_b$, and (4) nodes $\{\gamma^b_1, \dots, \gamma^b_{C_b}\}$ corresponding to $|C_b|$ units assigned to $Z = 0_b$. The resulting network then contains $T_a + C_a + T_b + C_b + 2$ vertices $\mathcal{V}$, connected by a total of $|\mathcal{E}| = T_a  + T_aC_a + C_aT_b + T_bC_b+ C_b$ edges of the following form:
\begin{equation*}
\begin{split}
    \mathcal{E} = \bigg\{(\xi, \tau^a_{t}), (\tau^a_{t}, \gamma^a_{c}), (\gamma^a_{c}, \tau^b_{t'}), (\tau^b_{t'}, \gamma^b_{c'}), ( \gamma^b_{c'}, \overline\xi), ~&t = 1, \dots, T_a,~c = 1, \dots, C_a, \\
    &t' = 1, \dots, T_b,  ~c' = 1, \dots, C_b\bigg\}.
    \end{split}
\end{equation*}

\begin{figure}[ht]
\centering
\begin{tikzpicture}[thick, color = black,
  fsnode/.style={circle, fill=black, inner sep = 0pt, minimum size = 5pt},
  ssnode/.style={circle, fill=black, inner sep = 0pt, minimum size = 5pt},
  shorten >= 3pt,shorten <= 3pt
]


\begin{scope}[start chain=going below,node distance=7mm]
\foreach \i in {1,2,3,4}
  \node[fsnode,on chain] (r\i) [label=above left: {\small$\tau^a_\i$} ] {};
\end{scope}

\begin{scope}[xshift=2.5cm,yshift=0cm,start chain=going below,node distance=7mm]
\foreach \i in {1,2,3,4,5}
  \node[ssnode,on chain] (t\i) [label=above : {\small$\gamma^a_\i$}] {};
\end{scope}

\begin{scope}[xshift=5cm,yshift=0cm,start chain=going below,node distance=7mm]
\foreach \i in {1,2,3}
  \node[ssnode,on chain] (tt\i) [label=above : {\small$\tau^b_\i$}] {};
\end{scope}

\begin{scope}[xshift=7.5cm,yshift=0cm,start chain=going below,node distance=7mm]
\foreach \i in {1,2,3,4,5,6}
  \node[ssnode,on chain] (c\i) [label=above: {\small$\gamma^b_\i$}] {};
\end{scope}

\node [circle, fill = black, inner sep = 0pt, minimum size = 5pt, label=left: $\xi$] at (-2, -1.5) (source) {};

\node [circle, fill = black, inner sep = 0pt, minimum size = 5pt, label=right: $\overline\xi$ ] at (9.5, -1.8) (sink) {};

\foreach \i in {1,2,3,4} {
   \draw[color=gray] (source) -- (r\i);
   }

\draw[color=blue, line width=0.5mm] (source) -- (r1);
\draw[color=blue, line width=0.5mm] (source) -- (r3);

\foreach \i in {1,2,3,4} {
   \foreach \j in {1,2,3,4,5} {
   \draw[color=gray] (r\i) -- (t\j);
   }
}

\draw[color=blue, line width=0.5mm] (r1) -- (t2);
\draw[color=blue, line width=0.5mm] (r3) -- (t3);

\foreach \i in {1,2,3,4,5} {
\foreach \j in {1,2,3} {
   \draw[color=gray] (t\i) -- (tt\j);
}  
}

\draw[color=blue, line width=0.5mm] (t2) -- (tt3);
\draw[color=blue, line width=0.5mm] (t3) -- (tt2);

\foreach \i in {1,2,3} {
   \foreach \j in {1,2,3,4,5,6} {
   \draw[color=gray] (tt\i) -- (c\j);
   }
} 

\draw[color=blue, line width=0.5mm] (tt2) -- (c6);
\draw[color=blue, line width=0.5mm] (tt3) -- (c1);

\foreach \i in {1,2,3,4,5,6} {
   \draw[color=gray] (c\i) -- (sink);
}

\draw[color=blue, line width=0.5mm] (c6) -- (sink);
\draw[color=blue, line width=0.5mm] (c1) -- (sink);

\end{tikzpicture}

\caption{Network-flow representation of proposed pair-of-pairs design. Nodes $\{\tau_1^a,\dots,\tau_4^a\}$ represent units assigned $Z = 1_a$; $\{\gamma_1^a,\dots,\gamma_5^a\}$ are assigned $Z = 0_a$, $\{\tau_1^b,\dots,\tau_3^b\}$ are assigned $Z = 1_b$, and $\{\gamma_1^b,\dots,\gamma_6^b\}$ are assigned $Z = 0_b$. The bold blue lines correspond to two matched strata, each consisting of a pair of pairs: $\{(\tau_1^a, \gamma_2^a), (\tau_3^b, \gamma_1^b)\}$ and $\{(\tau_3^a, \gamma_3^a), (\tau_2^b, \gamma_6^b)\}$ }
\label{fig: network flow}
\end{figure}
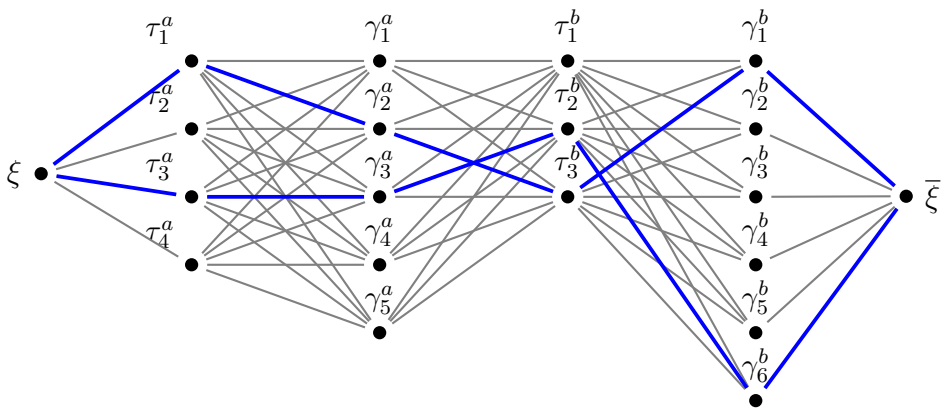

Each edge $e \in \mathcal{E}$ has unit capacity $\text{cap}(e) = 1$. A \textit{feasible integral flow} $f(\cdot)$ assigns a nonnegative integer value to each edge, subject to the following constraints:
\begin{itemize}
    \item Capacity constraints:  $0 \leq f(e) \leq \text{cap}(e)$, for all $e \in \mathcal{E}$;
    \item Supply–demand balance: all flow units emitted by the source $\xi$ are absorbed at the sink $\overline{\xi}$, i.e., $\sum_{t=1}^{T_a} f(\xi, \tau^a_{t}) = \sum_{c'=1}^{C_b} f(\gamma^b_{c'}, \overline\xi) = I \leq \min\{T_a, T_b, C_a, C_b\}$, where $I$ determines the number of matched strata in the final design;
    \item Flow conservation: for all intermediate vertices $b \in \mathcal{V} \setminus \{\xi, \overline\xi\}$, the inflow equals the outflow:
    \[
        \sum_{(a,b) \in \mathcal{E}} f(a,b) = \sum_{(b,c) \in \mathcal{E}} f(b,c).
    \]
\end{itemize}

Under these constraints, a feasible integral flow  $f(\cdot)$  encodes a set of matched strata, each containing two treated-control pairs of the form:
\[
    \mathcal{M}(f) = \left\{\left\{ (\tau^a_{t}, \gamma^a_{c}),(\tau^b_{t'}, \gamma^b_{c'}) \right\}: f(\tau^a_{t}, \gamma^a_{c}) = f(\gamma^a_{c}, \tau^b_{t'}) =f(\tau^b_{t'}, \gamma^b_{c'})= 1 \right\}.
\]
\noindent In the simple network illustrated in Figure~\ref{fig: network flow}, blue edges represent one feasible flow when constructing $I = 2$ matched strata:  $\{(\tau^a_1, \gamma^a_2), (\tau^b_3, \gamma^b_1)\}$ and $\{(\tau^a_3, \gamma^a_3), (\tau^b_2, \gamma^b_6)\}$. 

Each edge is further associated with a cost, e.g., rank-based, robust Mahalanobis distance calculated from baseline covariates \citep{rosenbaum2010design}. Let $\delta_{\tau^g_{t}, \gamma^{g}_{c}}(\boldsymbol{x})$ denote the distance between treated unit  $\tau^g_{t}$ and  control unit $\gamma^{g}_{c}$ on covariates $\boldsymbol{x}$, for $t = 1,\dots,T_g$, $c = 1, \dots,C_{g}$, and $g\in \{a,b\}$, and $\delta_{\gamma^a_{c}, \tau^{b}_{t}}(\boldsymbol{x})$ be defined similarly. The total cost associated with a feasible flow $f(\cdot)$ then equals:
\begin{equation*}
\begin{split}
    \text{cost}(f) := \sum_{f(\tau^a_{t}, \gamma^a_{c}) =1}\delta_{\tau^a_{t}, \gamma^a_{c}}(\boldsymbol{x}) + 
     \sum_{f(\gamma^a_{c}, \tau^b_{t'}) =1}\delta_{\gamma^a_{c},\tau^b_{t'}}(\boldsymbol{x}) +\sum_{f(\tau^b_{t'}, \gamma^b_{c'}) =1}\delta_{\tau^b_{t'}, \gamma^b_{c'}}(\boldsymbol{x}).
    \end{split}
\end{equation*}
A feasible flow is said to be optimal if it minimizes $\text{cost}(f)$. As described in more details in \citet{zhang2023matchingone}, a standard polynomial-time, minimum-cost network flow algorithm can be used to find the feasible flow $f$ with the minimum cost, which then yields the pair-of-pairs design displayed in Table \ref{table: cov balance after matching}.

%% file: 3_Notation.tex
\section{Notation and causal estimands}
\label{sec: notation and estimands}

\subsection{Potential outcomes; identification assumptions; principal stratification}
\label{subsec: PO}

We formalize the causal estimands and identification assumptions under the potential outcomes framework \citep{neyman1923application,rubin1974estimating}. First, we adopt Rubin's Stable Unit Treatment Value Assumption (SUTVA), which states that (i) the potential outcome of a participant does not depend on the treatment assignment of others; and (2) there is only one version of the treatment assignment  \citep{rubin1980randomization,rubin1986statistics}. Under SUTVA, we define the following potential treatment received for each participant indexed by $ijk$.

\begin{definition}[Potential treatment received]\label{def: potential treatment received}
    Let $D_{ijk}({Z} = z)$ denote participant $ijk$'s potential treatment received under IV level ${z} \in  \{0_a, 0_b, 1_a, 1_b\}$ as follows:
\begin{align*}
         &d_{C^aijk}\equiv D_{ijk}( {Z}={0}_{a}),~ \quad\quad~ d_{T^aijk}\equiv D_{ijk}( {Z}={1}_{a} ),\nonumber\\ 
         &d_{C^bijk}\equiv D_{ijk}( {Z}={0}_{b}),~ \quad\quad~ d_{T^bijk}\equiv D_{ijk}( {Z}={1}_{b} ).
\end{align*}
\end{definition}
\noindent In any given dataset, only the actual treatment received, $D_{ijk} = d_{C^aijk}\mathbbm{1}\{{Z}_{ijk}=0_a\} + d_{T^aijk}\mathbbm{1}\{{Z}_{ijk}=1_a\}+d_{C^bijk}\mathbbm{1}\{{Z}_{ijk}=0_b\} + d_{T^bijk}\mathbbm{1}\{{Z}_{ijk}=1_b\}$, is observed for each participant.

Analogously, Definition \ref{def: potential clinical outcomes} defines potential outcomes for each participant:
\begin{definition}[Potential outcomes]
\label{def: potential clinical outcomes}
    Let $R_{ijk}(D_{ijk}(Z_{ijk}), Z_{ijk})$ denote $ijk$'s potential outcome under IV assignment $Z_{ijk}$ and treatment received $D_{ijk}(Z_{ijk})$. We make the standard exclusion restriction assumption so that $R_{ijk}(D_{ijk}(Z_{ijk}) = d, Z_{ijk} = z) = R_{ijk}(D_{ijk}(Z_{ijk}) = d, Z_{ijk} = z'),~\forall~d, z, z'$.  We write $R_{ijk}(Z_{ijk})$ as a shorthand for $R_{ijk}(D_{ijk}(Z_{ijk}), Z_{ijk})$. Each participant $ijk$ is associated with the following $4$ potential outcomes:
\begin{equation*}
    \begin{split}
         &r_{C^aijk}\equiv R_{ijk}(Z = 0_a),~\quad\quad~   
         r_{T^aijk}\equiv R_{ijk}(Z = 1_a),\\
         &r_{C^bijk}\equiv R_{ijk}(Z = 0_b),~ \quad\quad~   
         r_{T^bijk}\equiv R_{ijk}(Z = 1_b).
    \end{split}
\end{equation*}
\end{definition}
\noindent Let $R_{ijk} = r_{C^aijk}\mathbbm{1}\{{Z}_{ijk}=0_a\} + r_{T^aijk}\mathbbm{1}\{{Z}_{ijk}=1_a\}+r_{C^bijk}\mathbbm{1}\{{Z}_{ijk}=0_b\} + r_{T^bijk}\mathbbm{1}\{{Z}_{ijk}=1_b\}$ denote the observed outcome for $ijk$.

Following \citet{angrist1996identification}, participant $ijk$ is said to be a complier with respect to the IV dose pair $(0_g, 1_g)$, for $g \in \{a, b\}$, if $(d_{T^gijk}, d_{C^gijk}) = (1, 0)$, an always-taker if $(d_{T^gijk}, d_{C^gijk}) = (1, 1)$, a never-taker if $(d_{T^gijk}, d_{C^gijk}) = (0, 0)$, and a defier if $(d_{T^gijk}, d_{C^gijk}) = (0, 1)$, where $d_{T^gijk}$ and $d_{C^gijk}$ are defined in Definition \ref{def: potential treatment received}. In IV analysis, a common assumption is that participants are not less likely to take the treatment when offered treatment rather than control \citep{angrist1996identification}. We adopt this assumption for each IV pair ($\{0_a, 1_a\}$ or $\{0_b, 1_b\}$) and refer to it as the partial monotonicity assumption.

\begin{assumption}[Partial monotonicity]\label{ass: partial mono} For each participant indexed by $ijk$, we have
     $d_{T^gijk}\geq d_{C^gijk}$, $g \in \{a, b\}$.
\end{assumption}

Under this partial monotonicity assumption, defiers under either IV pair are assumed not to exist, and the compliance rate under IV pair $(0_g, 1_g)$ is defined as follows: 
\begin{equation*}
    \iota_g = \frac{1}{4I}\sum_{i = 1}^I \sum_{j = 1}^2\sum_{k = 1}^2 (d_{T^gijk} - d_{C^gijk}),~g \in \{a, b\},
\end{equation*}
which is the sample average treatment effect of the binary IV $(0_g, 1_g)$ on treatment received.

Finally, we make the \emph{nested IV assumption} below to impose an additional structure on potential treatment received under two versions of the IV assignment \citep{wang2024nested}. 

\begin{assumption}[Nested IV]\label{ass: nested IV} For each participant indexed by $ijk$, 
     $(d_{C^aijk}, d_{T^aijk}) = (0, 1)$ implies $(d_{C^bijk}, d_{T^bijk}) = (0, 1)$. 
\end{assumption}

\noindent In words, Assumption \ref{ass: nested IV} says that a participant who is a complier under a weaker IV pair, that is, $(0_a, 1_a)$, would remain a complier under a stronger IV pair, that is, $(0_b, 1_b)$. The nested IV assumption has a testable implication. Under Assumptions \ref{ass: partial mono} and \ref{ass: nested IV}, we would necessarily have $\iota_b \geq \iota_a$; hence, by testing the one-sided null hypothesis $H_{0}: \iota_b \geq \iota_a$, one could assess the validity of Assumptions \ref{ass: partial mono} and \ref{ass: nested IV}. 

A participant's compliance status is determined jointly by both IV pairs, $(0_a, 1_a)$ and $(0_b, 1_b)$. In principle, this yields $2^4 = 16$ possible principal strata. However, under Assumptions \ref{ass: partial mono} and \ref{ass: nested IV}, the number is reduced from $16$ to $6$\footnote{
More precisely, Assumptions \ref{ass: partial mono} and \ref{ass: nested IV} reduce the number of principal strata from 16 to 7. In what follows, we consider only 6 principal strata, as two of the original strata are combined into a single stratum of switchers.
} \citep{wang2024nested}. These $6$ principal strata are:  (i) SWitchers (SW), who are non-compliers (that is, always-takers or never-takers) under the weaker IV pair $(0_a, 1_a)$ but compliers under the stronger IV pair $(0_b, 1_b)$, with $(d_{T^aijk}, d_{C^aijk}) \in \{(1, 1), (0, 0)\}$ and $(d_{T^bijk}, d_{C^bijk}) = (1, 0)$; (ii) Always-COmpliers (ACO),  who comply with the assigned treatment under both IV pairs $(0_a, 1_a)$ and $(0_b, 1_b)$, with $(d_{T^aijk}, d_{C^aijk}) = (d_{T^bijk}, d_{C^bijk}) = (1, 0)$; (iii) Always-Taker-Never-Takers (AT-NT), who always take the treatment under $\{0_a, 1_a\}$ but never take the treatment under 
 $\{0_b, 1_b\}$, with $(d_{T^aijk}, d_{C^aijk})=(1, 1)$ and $ (d_{T^bijk}, d_{C^bijk}) = (0, 0)$; (iv) Always-Always-Takers (AAT), who are always-takers under both IV pairs $(0_a, 1_a)$ and $(0_b, 1_b)$, with $(d_{T^aijk}, d_{C^aijk}) = (d_{T^bijk}, d_{C^bijk}) = (1, 1)$;
 (v) Never-Taker-Always-Takers (NT-AT),  who never take treatment under $\{0_a, 1_a\}$ but always take it under $\{0_b, 1_b\}$, with 
$(d_{T^aijk}, d_{C^aijk})=(0, 0)$ and $ (d_{T^bijk}, d_{C^bijk}) = (1, 1)$; and (vi) Always-Never-Takers (ANT), who never take the treatment under either IV pair $(0_a, 1_a)$ or $(0_b, 1_b)$, with $(d_{T^aijk}, d_{C^aijk}) = (d_{T^bijk}, d_{C^bijk}) = (0, 0)$. 

Finally, Definition \ref{def: PP outcomes} defines the ``per-protocol" potential outcomes.

\begin{definition}[Per-protocol potential outcomes]
\label{def: PP outcomes}
Let 
\[
r_{d=0, ijk}\equiv R_{ijk}(D = 0) \quad\text{and}\quad r_{d = 1, ijk}\equiv R_{ijk}(D = 1)
\]
denote the potential outcomes had participant $ijk$ taken the control or treatment.
\end{definition}
\noindent Unlike $(r_{T^g ijk}, r_{C^g ijk}),~g \in \{a, b\}$, which depend on the IV pair, $(r_{d=0, ijk}, r_{d=1, ijk})$ do not depend on the IV and are arguably the ultimate quantity of interest.

\subsection{Always-complier sample average treatment effect (ACO-SATE)}
\label{subsec: ACO-SATE}
Equipped with the potential outcomes defined in Definitions \ref{def: potential treatment received} and \ref{def: potential clinical outcomes}, we are ready to define the target causal estimands. Our first estimand of interest is the following ``effect ratio"-style estimand \citep{baiocchi2010building,kang2018inference,zhang2022bridging,chen2024manipulating}:
\begin{equation*}
    \kappa = \frac{\sum_{i = 1}^I \sum_{j = 1}^2 \sum_{k = 1}^2 (r_{T^aijk} - r_{C^aijk})}{\sum_{i = 1}^I \sum_{j = 1}^2\sum_{k = 1}^2 (d_{T^aijk} - d_{C^aijk})},
\end{equation*}
which, under the partial monotonicity assumption $d_{T^aijk}\geq d_{C^aijk}$ and the nested IV assumption, equals the SATE among the always-complier subgroup, that is,
\[
\text{ACO-SATE} = \frac{1}{n_{\text{ACO}}}\sum_{ijk \in ACO}\{r_{d = 1, ijk} - r_{d = 0, ijk} \},
\]
where $n_{\textup{ACO}}$ denotes the number of always-compliers in the cohort, and $(r_{d = 1, ijk}, r_{d = 0, ijk})$ are ``per-protocol" outcomes defined in Definition \ref{def: PP outcomes}. 

\subsection{Switcher sample average treatment effect (SW-SATE)}
\label{subsec: S-SWATE}
Our second target estimand is the following ratio of difference-in-differences:
\begin{equation}\label{causal estimand}
    \lambda = \frac{\sum_{i = 1}^I \sum_{j = 1}^2 \sum_{k = 1}^2 (r_{T^bijk} - r_{C^bijk})-(r_{T^aijk} - r_{C^aijk})}{\sum_{i = 1}^I \sum_{j = 1}^2\sum_{k = 1}^2 (d_{T^bijk} - d_{C^bijk}) - (d_{T^aijk} - d_{C^aijk})},
\end{equation}
where the numerator denotes the difference in the sample average treatment effect of IV on the outcome between two IV pairs, and the denominator denotes that on the treatment received. 

The estimand $\lambda$ is well-defined provided that the denominator, $\sum_{i = 1}^I \sum_{j = 1}^2\sum_{k = 1}^2 (d_{T^bijk} - d_{C^bijk}) - (d_{T^aijk} - d_{C^aijk})$, does not equal $0$; that is, the sample average treatment effect of IV on treatment received is different between two IV pairs.

In the absence of the partial monotonicity assumption and the nested IV assumption, the parameter $\lambda$ can be interpreted as follows: for every hundred additional people converted by the stronger IV $(0_b, 1_b)$ to take the treatment, the intention-to-treat effect is increased by $100 \times \lambda$. 

We further define the sample average treatment effect among switchers as follows:
\begin{equation}\label{S-SATE}
    \text{SW-SATE} = \frac{1}{n_{\text{SW}}}\sum_{ijk \in SW}\{r_{d = 1, ijk} - r_{d = 0, ijk} \},
\end{equation}
where $n_{\text{SW}}$ denotes the number of switchers in the cohort, and $(r_{d = 1, ijk}, r_{d = 0, ijk})$ are ``per-protocol" outcomes defined in Definition \ref{def: PP outcomes}. 

Under the partial monotonicity assumption and the nested IV assumption, we have $\lambda = \text{SW-SATE}$, as formalized in Proposition $1$.

\begin{proposition}\label{prop: lambda = SWATE}
    Let 
    $\lambda$ and $\text{SW-SATE}$ be defined as in \eqref{causal estimand} and \eqref{S-SATE}, respectively. Under Assumptions \ref{ass: partial mono} and \ref{ass: nested IV}, we have $\lambda = \text{SW-SATE}$. 
\end{proposition}

Proposition \ref{prop: lambda = SWATE} holds because, for the principal strata AT–NT, AAT, NT–AT, and ANT, under either IV, the IV has no effect on treatment received and, consequently, no effect on the outcome. Moreover, for the always-complier (ACO) principal stratum---participants who comply under either IV pair--- both the numerator and denominator in $\lambda$ are equal to $0$. Therefore, only the switcher subgroup contributes to $\lambda$, and we have:

\begin{equation}
\begin{split}
      \lambda = ~\frac{\sum_{ijk \in SW}  (r_{T^bijk} - r_{C^bijk})}{\sum_{ijk \in SW} (d_{T^bijk} - d_{C^bijk}) }
      = ~\frac{\sum_{ijk \in SW}  R_{ijk}( {D} = 1) - R_{ijk}( {D} = 0)}{\sum_{ijk \in SW} (1 - 0) } 
      = ~\text{SW-SATE},
\end{split}
\end{equation}
where the first equality holds because under the weaker IV pair $\{0_a, 1_a\}$, switchers are never-takers or always-takers with a zero intention-to-treat effect, and the second equality holds because under the stronger IV pair $\{0_b, 1_b\}$, switchers become compliers.

Both causal parameters, $\kappa$ and $\lambda$, involve contrasts of potential outcomes. 
The estimand $\kappa$ depends on $(d_{T^aijk}, r_{T^aijk})$ and $(d_{C^aijk}, r_{C^aijk})$; for each participant assigned to the weaker IV pair, either $(d_{T^aijk}, r_{T^aijk})$ or $(d_{C^aijk}, r_{C^aijk})$ is observed, whereas for those assigned to the stronger IV pair, neither is observed. Similarly, the estimand $\lambda$ involves $4$ sets of potential outcomes--- $(d_{T^aijk}, r_{T^aijk})$, $(d_{C^aijk}, r_{C^aijk})$, $(d_{T^bijk}, r_{T^bijk})$, $(d_{C^bijk}, r_{C^bijk})$---of which only one set is observed per participant. Consequently, statistical inference is required for both estimands. We develop design-based inference for $\kappa$ and $\lambda$ in the next section.

%% file: 4_Inference.tex
\section{Design-based inference for ACO-SATE and SW-SATE}
\label{sec: inference}
\subsection{Inference under the randomization assumption}
\label{subsec: inference under RA}
Write $\mathcal{F} = \{(d_{C^a ijk}, d_{T^a ijk},d_{C^b ijk}, d_{T^b ijk}, r_{C^a ijk}, r_{T^a ijk}, r_{C^b ijk}, r_{T^b ijk}, \boldsymbol{x}_{ijk}),~i = 1, \dots, I,~j = 1, 2,~k = 1, 2\}$ and $\mathcal{Z}$ the set containing all $8^I$ possible IV assignments. To facilitate making inference of  $\kappa$ and $\lambda$, we introduce an auxiliary random variable $Z^g_{ijk}$ for each $g \in \{a,b\}$ as follows: 
$Z^g_{ijk} = 1$ if $Z_{ijk} = 1_g$, $Z^g_{ijk} = 0$ if $Z_{ijk} = 0_g$, and $Z^g_{ijk} = 1/2$ otherwise. By design, this construction ensures that $Z^g_{ij1} + Z^g_{ij2} = 1$ for $j \in \{1,2\}$ and $g \in \{a,b\}$. The IV assignment mechanism for  $\boldsymbol{Z_i} = (Z_{i11}, Z_{i12}, Z_{i21}, Z_{i22})$, as formalized by Assumption~\ref{ass: Randomization Assumption}, determines the joint distribution of $Z^g_{ijk}$. 

We first derive a valid level-$\alpha$ confidence interval for the difference in compliance rates between the two IV pairs, which can be interpreted as the proportion of switchers in the sample under Assumptions~\ref{ass: partial mono} and \ref{ass: nested IV}. To that end, we conduct a finite-population asymptotic analysis in which the $4I$ units are embedded into a sequence of finite populations of increasing size \citep{fcltxlpd2016}. This asymptotic framework forms the basis for the statistical inference developed throughout the remainder of the article.

\begin{proposition}[Inference for the proportion of switchers]
\label{prop: compliance rate}
Let
    \begin{equation*}
    \begin{split}
V_{i} &= \sum_{j = 1}^2  \sum_{k = 1}^2 
    \left\{ \{Z^b_{ijk} D_{ijk} -(1-Z^b_{ijk}) D_{ijk}\}
    -  \{Z^a_{ijk} D_{ijk} -(1-Z^a_{ijk})D_{ijk}\}\right\}, \\
    \overline{V} &= I^{-1}\sum_{i = 1}^I V_{i},\\
    \end{split}
\end{equation*}
and define the variance estimator 
$$
S^2 = \frac{1}{I(I-1)}\sum_{i = 1}^I\left(V_{i}-\overline{V}\right)^2.
$$
Under mild regularity conditions and conditional on $\mathcal{F}$ and $\mathcal{Z}$, the interval 
$$
\left\{c: \left|\overline{V} - c \right| \leq z_{1-\alpha/2} \sqrt{S^2 }\right\}
$$ 
is an asymptotically valid, level-$\alpha$ confidence interval for $\iota_b - \iota_a$, the proportion of switchers in the sample,  where $z_{1-\alpha/2}$ denotes the $(1-\alpha/2)$-th quantile of the standard normal distribution.
\end{proposition}

\begin{remark}[Inference of $\iota_g$]
    By letting $V_{i} = \sum_{j = 1}^2  \sum_{k = 1}^2 \{Z^g_{ijk} D_{ijk} -(1-Z^g_{ijk})D_{ijk}\}$ in Proposition \ref{prop: compliance rate}, we can similarly construct a level-$\alpha$ confidence interval for $\iota_g$, the compliance rate under IV pair $(0_g, 1_g),~g \in \{a,b\}$.
\end{remark}

Next, we consider testing the null hypothesis
\[
H^{\mathrm{ACO}}_0: \kappa = \kappa_0
\]
for the ACO-SATE estimand. Theorem \ref{thm: RI for ACO-SATE} provides an asymptotically valid, level-$\alpha$ test under the randomization scheme specified in Assumption~\ref{ass: Randomization Assumption}.

\begin{theorem}[Inference for the SATE among always-compliers]\label{thm: RI for ACO-SATE}
Let
\begin{align*}
V_{i, \mathrm{ACO}}(\kappa_0) &=  \sum_{j = 1}^2  \sum_{k = 1}^2 \left\{Z^a_{ijk}(R_{ijk} - \kappa_0 D_{ijk}) -(1-Z^a_{ijk})(R_{ijk} - \kappa_0 D_{ijk})\right\}. 
\end{align*}
    and define the test statistic
\begin{align*}
    T_{\mathrm{ACO}}(\kappa_0) = I^{-1}\sum_{i = 1}^I V_{i, \mathrm{ACO}}(\kappa_0),
\end{align*}
with the corresponding variance estimator \[
S_{\mathrm{ACO}}^2(\kappa_0) = \frac{1}{I(I-1)}\sum_{i = 1}^I\left\{V_{i,\mathrm{ACO}}(\kappa_0)-I^{-1}\sum_{i = 1}^I V_{i,\mathrm{ACO}}(\kappa_0)\right\}^2.
\] 
Conditional on $\mathcal{F}$ and $\mathcal{Z}$,  the test that rejects $H^{\mathrm{ACO}}_0: \kappa = \kappa_0$ when $$|T_{\mathrm{ACO}}(\kappa_0)| \geq z_{1-\alpha/2} \sqrt{S_{\mathrm{ACO}}^2(\kappa_0) }$$ is an asymptotically valid, level-$\alpha$ test, under mild regularity conditions.
\end{theorem}

\begin{remark}\label{rmk: ACO-SATE CI also valid for subgroup CO}
    The level-$\alpha$ confidence interval established in Theorem~\ref{thm: RI for ACO-SATE} is also a conditionally valid level-$\alpha$ confidence interval for the effect ratio estimand restricted to the $2I$ units assigned to $\{0_a, 1_a\}$, noting that these $2I$ units constitute a random sample from the full set of $4I$ units; see, for example, \citet{branson2018randomization}, for a related discussion.
\end{remark}

\begin{remark}\label{rmk: ACO-SATE conservative variance}
     The variance estimator $S_{\mathrm{ACO}}^2(\kappa_0)$ generally overestimates the true variance \citep{imai2008variance}. It is unbiased for the true variance if and only if
    \[
         \sum_{j = 1}^2 \sum_{k = 1}^2 
    \{r_{d = 1, ijk} - r_{d=0, ijk} - \kappa_0 \}\mathbbm{1}\{S_{ijk} = \mathrm{ACO}\}
    \]
    is constant across all strata $i = 1,\dots,I$. A sufficient condition for this to hold is that the treatment effect among always-compliers is constant. 
\end{remark}

Finally, we move on to making inference for the SW-SATE. Theorem \ref{thm: RI} establishes an asymptotically valid, level-$\alpha$ test for the null hypothesis 
\[
H^{\mathrm{SW}}_0: \lambda = \lambda_0.
\]

\begin{theorem}[Inference for the SATE among switchers]\label{thm: RI}
Let
\begin{align*}
V_{i,\mathrm{SW}}(\lambda_0) &= \sum_{j = 1}^2  \sum_{k = 1}^2 
    \{Z^b_{ijk}(R_{ijk} - \lambda_0 D_{ijk}) -(1-Z^b_{ijk})(R_{ijk} - \lambda_0 D_{ijk})\}\nonumber
    \\&- \sum_{j = 1}^2  \sum_{k = 1}^2 \{Z^a_{ijk}(R_{ijk} - \lambda_0 D_{ijk}) -(1-Z^a_{ijk})(R_{ijk} - \lambda_0 D_{ijk})\}, 
\end{align*}
and define the test statistic
\begin{align*}
    T_{\mathrm{SW}}(\lambda_0) = I^{-1}\sum_{i = 1}^I V_{i,\mathrm{SW}}(\lambda_0),
\end{align*}
with the corresponding variance estimator 
\[
S_{\mathrm{SW}}^2(\lambda_0) = \frac{1}{I(I-1)}\sum_{i = 1}^I\left\{V_{i,\mathrm{SW}}(\lambda_0)-I^{-1}\sum_{i = 1}^I V_{i,\mathrm{SW}}(\lambda_0)\right\}^2. 
\] Under mild regularity conditions and conditional on $\mathcal{F}$ and  $\mathcal{Z}$, the test that rejects $H^{SW}_0: \lambda = \lambda_0$ when $$|T_{\mathrm{SW}}(\lambda_0)| \geq z_{1-\alpha/2} \sqrt{S_{\mathrm{SW}}^2(\lambda_0) }$$ is an asymptotically valid, level-$\alpha$ test.
\end{theorem}

\begin{remark}[Conservativeness of the variance estimator]\label{remark: SW-SATE variance}
    As with ACO-SATE, the variance estimator  $S_{\mathrm{SW}}^2(\lambda_0)$ for SW-SATE in general  overestimates the true variance of the test statistic. It is unbiased if and only if 
    \[
        \sum_{j = 1}^2 \sum_{k = 1}^2 
        \{r_{d = 1, ijk} - r_{d = 0, ijk} - \lambda_0 \}\mathbbm{1}(S_{ijk} = \mathrm{SW})
    \]
    is constant across all strata $i = 1, \dots, I$. A sufficient condition for this to hold is when the treatment effect is constant
    among switchers. In practice, such conditions may be stringent, and therefore $S_{\mathrm{SW}}^2(\lambda_0)$ typically overestimates the true variance of the test statistic $T_{\mathrm{SW}}(\lambda_0)$. However, as we will demonstrate in the simulation studies, the extent to which $S_{\mathrm{ACO}}^2(\kappa_0)$ in Theorem \ref{thm: RI for ACO-SATE} and $S_{\mathrm{SW}}^2(\lambda_0)$ are overestimated is relatively small across many practical scenarios; consequently, inference procedures based on Theorems~\ref{thm: RI for ACO-SATE} and~\ref{thm: RI} are often not overly conservative.
\end{remark}

In addition to estimating the treatment effects among always-compliers and switchers, the nested IV framework also makes the effect heterogeneity across different complier subpopulations 
an object of formal inference rather than an untestable assumption. To this end, we develop a test of the null hypothesis 
\[
H_0^{\mathrm{homog}}: \kappa = \lambda,
\]
which states that ACO-SATE and SW-SATE coincide. 
 Under Assumptions~2 and~3, testing $H_0^{\mathrm{homog}}$ can be reduced to testing whether the effect ratio under the stronger IV pair $(0_b, 1_b)$ equals that under the weaker IV pair $(0_a, 1_a)$; see Supplemental Material~A.5 for details. Theorem~\ref{thm:homog} exploits this reformulation to construct a valid, level-$\alpha$ test.

\begin{theorem}[Test of treatment effect homogeneity]
\label{thm:homog}
For each stratum $i$, define the vector
$W_i=(D_i^a, Y_i^a, D_i^b, Y_i^b)^\top$, with components
\begin{equation*}\label{eq:Wcomponents}
  D_i^g = \sum_{j=1}^2\sum_{k=1}^2
    \bigl\{Z^g_{ijk}D_{ijk}-(1-Z^g_{ijk})D_{ijk}\bigr\},\quad
  Y_i^g = \sum_{j=1}^2\sum_{k=1}^2
    \bigl\{Z^g_{ijk}R_{ijk}-(1-Z^g_{ijk})R_{ijk}\bigr\},
\end{equation*}
for $g\in\{a,b\}$. 
 Let $\bar{W}_I = I^{-1}\sum_{i=1}^I W_i$ and $S_W^2 = (I-1)^{-1}\sum_{i=1}^I (W_i - \bar{W}_I)(W_i - \bar{W}_I)^\top$ denote the sample mean and sample covariance matrix, respectively. Let $g: \mathbb{R}^4 \to \mathbb{R}$ be defined by $g(x_1,x_2,x_3,x_4) = x_4/x_3 - x_2/x_1$, with gradient
$\nabla g(x)=(x_2/x_1^2,\,-1/x_1,\,-x_4/x_3^2,\,1/x_3)^\top$.
Define the test statistic
\begin{equation*}\label{eq:Thomog}
  T_{\mathrm{homog}}
  \;=\; \frac{\sqrt{I}\,g(\bar{W}_I)}
         {\sqrt{[\nabla g(\bar{W}_I)]^\top S_W^2\,[\nabla g(\bar{W}_I)]}}.
\end{equation*}
Under mild regularity conditions and conditional on $\mathcal{F}$ and $\mathcal{Z}$, the test that rejects $H_{0,\mathrm{homog}}$ when $|T_{\mathrm{homog}}| \geq z_{1-\alpha/2}$ is an asymptotically valid, level-$\alpha$ test.
\end{theorem}
 
\begin{remark}\label{rem:homog_conserv}
 The variance estimator $S_W^2$ overestimates the true asymptotic variance of $\sqrt{I}\,\bar W_I$ in the Loewner (positive semi-definite) order, so that the resulting test is, in general, conservative. The estimator is asymptotically unbiased if the conditional mean $E[W_i \mid \mathcal{F}]$ is constant across strata $i = 1, \ldots, I$. 
A sufficient condition is that (i) both the number of always-compliers and the number of switchers are constant across strata, and (ii) the unit-level treatment effect $r_{d=1,ijk} - r_{d=0,ijk}$ is constant within the always-complier and switcher strata, respectively, though not necessarily equal between them.
\end{remark}

\subsection{A partly biased randomization scheme}
\label{subsec: biased randomization scheme}
In an integrated analysis of data from multiple sites in a multicenter clinical trial, treatment assignment is randomized within each site---either with or without matching on observed covariates---whereas ``site selection” may not be randomized, even after accounting for observed covariates. To be more specific, among the $8$ possible IV assignment configurations in $\Omega$, those in 
\begin{equation*}
    \Omega_1 = \{(1_a, 0_a, 1_b, 0_b), (0_a, 1_a, 1_b, 0_b), (1_a, 0_a, 0_b, 1_b), (0_a, 1_a, 0_b, 1_b)\}
\end{equation*}
are equiprobable, and we denote this probability as $\pi_{ia} = \Pr(\boldsymbol{Z_i} = \boldsymbol{e} \mid \boldsymbol{X}, \textsf{M})$ for $\boldsymbol{e} \in \Omega_1$. Similarly, the $4$ IV assignments in 
\begin{equation*}
    \Omega_2 = \{(1_b, 0_b, 1_a, 0_a), ( 1_b, 0_b, 0_a, 1_a), (0_b, 1_b, 1_a, 0_a), (0_b, 1_b, 0_a, 1_a)\}
\end{equation*}
are equiprobable. If we use $\pi_{ib} = \Pr(\boldsymbol{Z_i} = \boldsymbol{e} \mid \boldsymbol{X}, \textsf{M})$ for $\boldsymbol{e} \in \Omega_2$ to denote this probability, then $\pi_{ia} + \pi_{ib} = 1/4.$

We follow the Rosenbaum sensitivity analysis model \citep{rosenbaum2002observational} and quantify the maximum deviation from a randomization scheme for each matched strata $i$ as follows:
\begin{equation}
\label{eq: site level biased randomization 1}
    \frac{1}{4(1+\Gamma)} \leq \pi_{ia}\leq \frac{\Gamma}{4(1+\Gamma)},
\end{equation}
or equivalently:
\begin{equation}
\label{eq: site level biased randomization 2}
    \frac{1}{\Gamma} \leq \frac{\pi_{ia}}{\frac{1}{4} - \pi_{ia}} \leq \Gamma.
\end{equation}

Under this model, $\Gamma = 1$ corresponds to randomization at the ``trial" or ``clinical site" level, in which case all elements in $\Omega = \Omega_1 \cup \Omega_2$ are equiprobable and Assumption \ref{ass: Randomization Assumption} holds. As $\Gamma > 1$ increases, the potential deviation from randomization within each matched stratum $i$ grows, and Equation \eqref{eq: site level biased randomization 1} (equivalently, Equation \eqref{eq: site level biased randomization 2}) specifies a uniform upper bound on that deviation. 

Definition \ref{def: partly biased randomization} formalizes this partly biased randomization scheme.

\begin{definition}[Partly biased randomization scheme]\label{def: partly biased randomization}
   A partly biased randomization scheme indexed by $\Gamma \geq 1$, denoted as $\mathcal{M}_\Gamma$, refers to the set of probability distributions on $\Omega = \Omega_1 \cup \Omega_2$ that satisfy the restriction in Equation \eqref{eq: site level biased randomization 1} or equivalently Equation \eqref{eq: site level biased randomization 2}.
\end{definition}

\subsection{Inference under the partly biased randomization scheme}
\label{subsec: inference under biased randomization}
To facilitate inference under the partly biased randomization scheme $\mathcal{M}_\Gamma$, we first introduce an alternative characterization of the partly biased randomization scheme. 

Fix $\Gamma \geq 1$. Define $A_i = 1$ if, in stratum $i$, pair $j = 1$ is assigned to the IV pair $\{0_a, 1_a\}$ and pair $j = 2$ to $\{0_b, 1_b\}$, and let $A_i = 0$ if the assignments are reversed. Let $A_i$ satisfy 
\[
\frac{1}{1+\Gamma} \leq P(A_i = 1) \leq \frac{\Gamma}{1 + \Gamma},
\]
for $i = 1, \dots, I$.

For all $i = 1, \dots, I, ~ j = 1, 2$, further define the random variable $T_{ij}$ as follows: 
\[
T_{ij} = 
\begin{cases}
    +1, ~~\text{if the unit } ij1 \text{ is assigned to } 1_a \text{ or } 1_b \text{ and } ij2 \text{ to } 0_a \text{ or } 0_b, \\
    -1, ~~\text{if the unit } ij1 \text{ is assigned to } 0_a \text{ or } 0_b \text{ and } ij2 \text{ to } 1_a \text{ or } 1_b.
\end{cases}
\]
By construction, $T_{ij}$ is independent of $A_i$, and randomization within each pair $ij$ ensures $P(T_{ij} = +1) = P(T_{ij} = -1) = 1/2$. One can readily see that the partly biased randomization model $\mathcal{M}_{\Gamma}$ is then fully characterized by $\{(A_i, T_{i1}, T_{i2}),~i = 1, \dots, I\}$.

We consider testing the null hypothesis $H^{\mathrm{ACO}}_0: \kappa = \kappa_0$ 
for always-compliers
under $\mathcal{M}_\Gamma$. 
For each matched stratum $i$, define:
\begin{equation}\label{eq: aco tau_i bias}
\begin{split}
\hat\tau_{i, \mathrm{ACO}}(\kappa_0)
= &A_i\cdot T_{i1}\cdot\left\{R_{i11} - \kappa_0 D_{i11} - (R_{i12} -\kappa_0 D_{i12})\right\} \\
    + &(1-A_i)\cdot T_{i2}\cdot\left\{R_{i21} - \kappa_0 D_{i21} - (R_{i22} -\kappa_0 D_{i22})\right\}, 
\end{split}
\end{equation}
and
\begin{align}\label{eq: aco D_i bias}
     D^{\mathrm{ACO}}_{i,\Gamma}(\kappa_0)=  \hat\tau_{i, \mathrm{ACO}}(\kappa_0) -\left(\frac{\Gamma-1}{\Gamma+1}\right) |\hat\tau_{i, \mathrm{ACO}}(\kappa_0)|.
\end{align}
Theorem \ref{thm: aco biased_inference} shows that a valid test for $H^{\mathrm{ACO}}_0: \kappa = \kappa_0$ under $\mathcal{M}_\Gamma$ can be constructed using $\{D^{\mathrm{ACO}}_{i,\Gamma}(\kappa_0),~i = 1, \dots, I\}$. 

\begin{theorem}\label{thm: aco biased_inference}
For each matched stratum $i$, let $\hat\tau_{i, \mathrm{ACO}}(\kappa_0)$ and $D^{\mathrm{ACO}}_{i, \Gamma}(\kappa_0)$ be defined as in \eqref{eq: aco tau_i bias} and \eqref{eq: aco D_i bias}, respectively. Let $\overline{D}_{\Gamma, \mathrm{ACO}}(\kappa_0) = I^{-1}\sum_{i=1}^I D^{\mathrm{ACO}}_{i,\Gamma}(\kappa_0)$ be the sample mean across $I$ matched strata and $
S^{2}_{\Gamma,\mathrm{ACO}}(\kappa_0) = \frac{1}{I(I-1)} \sum_{i=1}^I\left(D^{\mathrm{ACO}}_{i, \Gamma}(\kappa_0)-\overline{D}_{\Gamma, \mathrm{ACO}}(\kappa_0)\right)^2
$ the corresponding variance estimator. Consider testing $H^{\mathrm{ACO}}_0: \kappa = \kappa_0$ against a greater-than alternative under the partly biased IV assignment model $\mathcal{M}_\Gamma$. The test that rejects $H^{\mathrm{ACO}}_0$ when
$\overline{D}_{\Gamma,\mathrm{ACO}}(\kappa_0)/S_{\Gamma, \mathrm{ACO}}(\kappa_0) \geq z_{1-\alpha}$ is an asymptotically valid level-$\alpha$ test, and the corresponding one-sided, level-$\alpha$ confidence interval can be obtained by inverting the test.
\end{theorem}

Lastly, we consider testing the null hypothesis  $H^{\mathrm{SW}}_0: \lambda = \lambda_0$ under the model $\mathcal{M}_\Gamma$. 
For each matched stratum $i$, define:
{\small\begin{equation*}\label{eq: tau_i bias}
\begin{split}
\hat\tau_{i, \mathrm{SW}}(\lambda_0) 
= &A_i\cdot\left\{T_{i2}\cdot\left\{R_{i21} - \lambda_0 D_{i21} - (R_{i22} - \lambda_0 D_{i22})\right\}
    - T_{i1}\cdot\left\{R_{i11} - \lambda_0 D_{i11} - (R_{i12} - \lambda_0 D_{i12})\right\}\right\} \\
    + &(1-A_i)\cdot\left\{T_{i1}\cdot\left\{R_{i11} - \lambda_0 D_{i11} - (R_{i12} - \lambda_0 D_{i12})\right\}
    - T_{i2}\cdot\left\{R_{i21} - \lambda_0 D_{i21} - (R_{i22} - \lambda_0 D_{i22})\right\}\right\}, 
\end{split}
\end{equation*}}
and 
\begin{align*}\label{eq: D_i bias}
     D^{\mathrm{SW}}_{i,\Gamma}(\lambda_0) =  \hat\tau_{i, \mathrm{SW}}(\lambda_0) -\left(\frac{\Gamma-1}{\Gamma+1}\right) |\hat\tau_{i, \mathrm{SW}}(\lambda_0)|.
\end{align*}
Theorem \ref{thm: biased_inference} derives a valid, level-$\alpha$ test for the null hypothesis $H^{\mathrm{SW}}_0: \lambda = \lambda_0$ 
under 
$\mathcal{M}_\Gamma$.  

\begin{theorem}\label{thm: biased_inference}
For each matched pair $i$, let 
$\hat\tau_{i, \mathrm{SW}}(\lambda_0)$ and $D^{\mathrm{SW}}_{i,\Gamma}(\lambda_0)$ be defined as above.
Let $\overline{D}_{\Gamma, \mathrm{SW}}(\lambda_0) = I^{-1}\sum_{i=1}^I D^{\mathrm{SW}}_{i,\Gamma}(\lambda_0)$ be the sample mean across $I$ matched strata and $
S^2_{\Gamma, \mathrm{SW}}(\lambda_0) = \frac{1}{I(I-1)} \sum_{i=1}^I\left(D^{\mathrm{SW}}_{i, \Gamma}(\lambda_0)-\overline{D}_{\Gamma, \mathrm{SW}}(\lambda_0)\right)^2
$ be the usual variance estimator. Consider testing $H^{\mathrm{SW}}_0: \lambda = \lambda_0$ against a greater-than alternative under the partly biased IV assignment model $\mathcal{M}_\Gamma$. The test that rejects $H^{\mathrm{SW}}_0$ when
$\overline{D}_{\Gamma, \mathrm{SW}}(\lambda_0)/S_{\Gamma, \mathrm{SW}}(\lambda_0) \geq z_{1-\alpha}$ is an asymptotically valid level-$\alpha$ test, and the corresponding one-sided, level-$\alpha$ confidence interval can be obtained by inverting the test.
\end{theorem}

\begin{proof}
    We provide a sketch of the proof for Theorem \ref{thm: biased_inference}. The proof of Theorem \ref{thm: aco biased_inference} is analogous. Proving Theorem \ref{thm: biased_inference} consists of two key steps. The first step is to construct a bounding variable for the random variable 
     $D^{\mathrm{SW}}_{i,\Gamma}(\lambda_0)$
    conditional on each of the $4$ realizations of $(T_{i1}, T_{i2})$, which can be achieved using arguments in \citet{fogarty2020studentized}. Let $U_{i, \Gamma}^{(+1, +1)}$, $U_{i, \Gamma}^{(+1, -1)}$, $U_{i, \Gamma}^{(-1, +1)}$, and $U_{i, \Gamma}^{(-1, -1)}$ denote the bounding variable for 
    $D^{\mathrm{SW}}_{i,\Gamma}(\lambda_0)$
    when $(T_{i1}, T_{i2}) = (+1, +1)$, $(+1, -1)$, $(-1, +1)$, and $(-1, -1)$, respectively. Define a random variable $U_{i, \Gamma}$ that equals $U_{i, \Gamma}^{(+1, +1)}$, $U_{i, \Gamma}^{(+1, -1)}$, $U_{i, \Gamma}^{(-1, +1)}$, and $U_{i, \Gamma}^{(-1, -1)}$, each with probability $1/4$. In the second step, it can be shown that, under the null hypothesis $H_0: \lambda = \lambda_0$, $\mathbb{E}\left[\sum_{i=1}^I U_{i, \Gamma}\right] \leq 0,$ which then implies $\mathbb{E}[\sum_{i=1}^I 
     D^{\mathrm{SW}}_{i,\Gamma}(\lambda_0)
    ] \leq 0$ because 
     $D^{\mathrm{SW}}_{i,\Gamma}(\lambda_0)$
    is stochastically dominated by $U_{i, \Gamma}.$ The full details can be found in the Supplemental Material A. 
\end{proof}

%% file: 5_Simulation.tex
\section{Simulation}
\label{sec: simulation}

We have two goals in the simulation studies. Section \ref{subsec: simu randomization} evaluates the finite-sample performance of the proposed tests for the SW-SATE and ACO-SATE under a randomization scheme. Section \ref{subsec: simulation biased RI} summarizes the performance of the proposed test under the partly biased randomization scheme.

\subsection{Validity and power of the proposed tests under a randomization scheme}
\label{subsec: simu randomization}
\subsubsection{Simulation setup}
\label{subsubsec: simu rand DGP}
We first verify the level and assess the power of the hypothesis test for the SW-SATE. We generate data for a PoP-NIV design consisting of $I$ strata. The first factor we vary is the number of matched strata:
\begin{description}
    \item[Factor 1:] Number of matched strata $I$: $100$, $500$, and $1000$.
\end{description}

For each unit $k$ in matched pair $j$ within stratum $i \in [I]$, its principal stratum membership $S_{ijk}$ is sampled from the following multinomial distribution: 
\begin{equation*}
\begin{split}
    &P(S_{ijk} = \text{SW}) = p; \\
    &P(S_{ijk} \in \{\text{ACO, AT-NT, AAT, NT-AT, ANT}\}) = (1-p)/5.
\end{split}
\end{equation*}

Recall switchers entail those who switch from a never-taker to a complier and from an always-taker to a complier. Half of switchers in our data-generating process are never-takers under a weaker IV, and the other half are always-takers. For each unit $ijk$, the potential outcomes $\{d_{T^a ijk}, d_{C^a ijk}, d_{T^b ijk}, d_{C^b ijk}\}$ are determined by $S_{ijk}$.

The second factor we vary is the proportion of switchers:
\begin{description}
    \item[Factor 2:] Proportion of switchers $p:$  $0.3$, $0.5$, and $0.7$.
\end{description}

Next, we generate potential ``per-protocol" outcomes $\{r_{d=0, ijk}, r_{d=1, ijk}\}$ as follows:
$$r_{d=0, ijk} \sim \mathcal{N}(0,1), ~~r_{d=1, ijk} = r_{d=0, ijk} + \tau_{ijk}, $$
where the treatment effect $\tau_{ijk}$ is heterogeneous and depends on $ijk$'s principal stratum membership: $\tau_{ijk} \sim \mathcal{N}(0.5,1)$ for always-compliers, $\tau_{ijk} \sim \mathcal{N}(0.1,1)$ for other non-switcher strata (AT-NT, AAT, NT-AT, or ANT), and $\tau_{ijk}$ is drawn from one of the following distributions for switchers: 
\begin{description}
    \item[Factor 3:] Treatment effect among switchers, $\tau^{\text{SW}}_{ijk}$, follows one of the three distributions: (i) $\tau^{\text{SW}}_{ijk} \sim \text{Unif}~[\mu - \sqrt{3}, \mu + \sqrt{3}]$; (ii) $\tau^{\text{SW}}_{ijk} \sim \mathcal{N}(\mu, 1)$; and (iii) $\tau^{\text{SW}}_{ijk} \sim \text{Exp}(1/\mu)$.
\end{description}

In each distribution, $\mu$ controls the magnitude of the effect size among switchers, which is varied as the fourth factor:
\begin{description}
    \item[Factor 4:] Effect size among switchers $\mu$: $0, 0.25, 0.5,$ and $1$.
\end{description}

After we generate participant $ijk$'s principal stratum $S_{ijk}$ and potential ``per-protocol" outcomes $\{r_{d=0, ijk}, r_{d=1, ijk}\}$, potential outcomes $\{r_{T^a ijk}, r_{C^a ijk}, r_{T^b ijk}, r_{C^b ijk}\}$ are then determined by $\{r_{d=0, ijk}, r_{d=1, ijk}, S_{ijk}\}$. 

Finally, we generate IV assignment vector $\boldsymbol{Z_i} = (Z_{i11}, Z_{i12}, Z_{i21}, Z_{i22})$ for each stratum $i$ independently according to the randomization scheme in Assumption \ref{ass: Randomization Assumption}. For each unit $ijk$, the observed treatment received and outcome are determined as follows:
\begin{align*}
     &D_{ijk} = d_{C^aijk}\mathbbm{I}\{{Z}_{ijk}=0_a\} + d_{T^aijk}\mathbbm{I}\{{Z}_{ijk}=1_a\}+d_{C^bijk}\mathbbm{I}\{{Z}_{ijk}=0_b\} + d_{T^bijk}\mathbbm{I}\{{Z}_{ijk}=1_b\}, \\
    &R_{ijk} = r_{d=1, ijk}\mathbbm{I}\{D_{ijk}=1\} + r_{d=0, ijk}\mathbbm{I}\{D_{ijk}=0\}.
\end{align*}

To summarize, Factors 1-4 define $108$ data-generating processes. For each DGP, we generated $1000$ datasets and recorded the true sample average treatment effect $\lambda_{\text{true}}$ among switchers in each dataset. We then evaluated the level and power of the randomization-based test in Theorem \ref{thm: RI} by conducting two tests: (i) $H_0: \lambda = \lambda_{\text{true}}$ and (ii) $H_0: \lambda = 0$. We also constructed confidence intervals for $\lambda_{\text{true}}$ by inverting the test.

\subsubsection{Simulation results}
\label{subsec: subsec: simu results rand}
Table \ref{tbl: simulation RI} summarizes the simulation results for different choices of the sample size $I$, switcher proportion $p$, and effect size $\mu$, when $\tau^{\text{SW}}_{ijk} \sim \mathcal{N}(\mu, 1)$. Reported metrics include type I error rate (level), empirical power, average length of the confidence intervals, and empirical coverage of the $95\%$ confidence intervals constructed according to Theorem \ref{thm: RI}. In addition, Table \ref{tbl: simulation RI} reports the average compliance rate $\bar{\iota}_g$ for each IV pair, the average true SW-SATE $\bar{\lambda}_{\text{true}}$, the average true standard deviation of the test statistic $\text{SD}[T(\lambda_{\text{true}})]$ under $\lambda = \lambda_{\text{true}}$, and the average standard error estimator $S(\lambda_{\text{true}})$, all averaged over $1000$ simulation replicates. Tables S2-S3 in the Supplemental Material B.1 summarize the same simulation results but for $\tau^{\text{SW}}_{ijk} \sim \text{Unif}~[\mu - \sqrt{3}, \mu + \sqrt{3}]$ and $\tau^{\text{SW}}_{ijk} \sim \text{Exp}(1/\mu)$.

We observed several consistent trends in the simulation studies. First, across all data-generating settings, the test maintained good type I error control, even with $I$ as small as $100$. The estimated standard errors slightly overestimated the true standard deviation, indicating that the variance estimator is mildly conservative. Second, as expected, the test’s power increased and the confidence interval width decreased as the number of strata $I$ grew and the proportion of switchers $p$ became larger. For example, when the effect distribution among switchers follows a normal distribution with mean $\mu = 1$, the empirical power is $0.76$ with a confidence interval width of $1.63$ for $I = 500$ strata and $p = 0.3$. Increasing the number of strata to 
$I=1000$ raised the empirical power to $0.96$ and narrowed the confidence interval to $1.10$. Further increasing the proportion of switchers to $p = 0.7$ (with $I = 1000$) yielded empirical power of 
$1.00$ and reduced the interval width to $0.44$. Third, fixing $I$ and $p$, the test showed higher power with larger effect sizes. For instance, with $I=500$ and $p=0.3$, when the effect distribution among switchers was uniform and the mean value increased from $0.25$ to $1.0$, empirical power rose from $0.11$ to $0.74$. Finally, the proposed method demonstrated robustness across a variety of effect size distributions.

In addition to the data-genrating processes described in Section \ref{subsubsec: simu rand DGP}, we also examined the scenario discussed in Remark \ref{remark: SW-SATE variance}, where each stratum contains exactly two switchers, and the treatment effect among switchers is constant. We verified that the variance estimator $S^2(\lambda_{\text{true}})$ is unbiased for the true variance under this setup; see Table~S4 in Supplementary Material B.1 for details.

\subsubsection{Hypothesis testing for the ACO-SATE}\label{subsubsec: addition simu rand}
We also evaluated the validity and power of the hypothesis test for the ACO-SATE, constructed according to Theorem~\ref{thm: RI for ACO-SATE}. The data-generating processes parallel those for SW-SATE in Section~\ref{subsubsec: simu rand DGP}, except that the roles of switchers and always-compliers are swapped. For each simulated dataset, we computed the true sample average treatment effect $\kappa_{\text{true}}$ for always-compliers and evaluated two hypothesis tests, $H_0: \kappa = \kappa_{\text{true}}$ and $H_0: \kappa = 0$. We also constructed confidence intervals by inverting the test. Performance was assessed in terms of type I error, empirical power, confidence interval length, and coverage. Supplemental Material~B.2.1 describes the simulation settings.

Tables~S5--S7 in Supplemental Material B.2.2 report the simulation results, which closely parallel those for SW-SATE. First, the test maintains a type I error rate near the nominal level across all combinations of number of strata $I$, proportion of always-compliers $p$, and the effect size $\mu$. Even with $I = 100$, the coverage of the 95\% confidence intervals stays close to the nominal level across all treatment effect distributions, demonstrating that the proposed randomization-based inference procedure performs reliably under finite samples and heterogeneous effect-size distributions. Second, increasing the number of strata $I$ or the proportion of always-compliers $p$ improves power and narrows confidence intervals. Finally, for fixed $I$ and $p$, larger effect sizes $\mu$ yield higher power, as expected.

\begin{table}[ht]
\centering
\caption{Simulation results over 1000 replicates for varying strata sizes ($I$), proportions of switchers ($p$), and effect size among switchers ($\mu$), under a normal distribution of treatment effects among switchers: $\tau^{\text{SW}}_{ijk} \sim \mathcal{N}(\mu, 1)$. Reported metrics include level, power, average compliance rate under the weaker IV pair $(\overline{\iota}_a)$, average compliance rate under the stronger IV pair $(\overline{\iota}_b)$, average SW-SATE $(\overline{\lambda}_{\text{true}})$, average length of the confidence interval, coverage probability of the confidence interval (cov. (\%)), the standard deviation of the test statistic under $\lambda = \lambda_{\text{true}}$ (SD[T($\lambda_{\text{true}}$)]), and the conservative standard error estimator ($S(\lambda_{\text{true}})$).   
}\label{tbl: simulation RI}
\resizebox{\textwidth}{!}{
\begin{tabular}{ccccccccccccc}
\midrule
\multirow{2}{*}{$I$}  & \multirow{2}{*}{$p$} & \multirow{2}{*}{$\mu$} & \multirow{2}{*}{level} & \multirow{2}{*}{power} & \multirow{2}{*}{$\bar{\iota}_a$} & \multirow{2}{*}{$\bar{\iota}_b$} & \multirow{2}{*}{$\bar{\lambda}_{\text{true}}$} & \multicolumn{2}{c}{CI}  & 100$\times$  & 100$\times$ \\ \cline{9-10}
  &  &  &  &  &  & & & length & cov.($\%$)&  SD[T($\lambda_{\text{true}}$)] & $S(\lambda_{\text{true}})$ \\
\midrule
  100 & 0.30 & 0.00 & 0.05 & 0.05 & 0.14 & 0.44 & 0.00 & 4.96 & 95.2 & 6.11 & 6.13 \\ 
  100 & 0.30 & 0.25 & 0.05 & 0.07 & 0.14 & 0.44 & 0.25 & 4.92 & 94.8 & 6.11 & 6.13  \\ 
  100 & 0.30 & 0.50 & 0.06 & 0.10 & 0.14 & 0.44 & 0.50 & 4.88 & 93.7 & 6.14 & 6.16 \\ 
  100 & 0.30 & 1.00 & 0.05 & 0.21 & 0.14 & 0.44 & 1.00 & 5.01 & 94.8 & 6.41 & 6.44  \\ 
  100 & 0.50 & 0.00 & 0.05 & 0.05 & 0.10 & 0.60 & 0.00 & 2.21 & 94.7 & 6.09 & 6.14  \\ 
  100 & 0.50 & 0.25 & 0.06 & 0.08 & 0.10 & 0.60 & 0.25 & 2.23 & 93.5 & 6.09 & 6.12 \\ 
  100 & 0.50 & 0.50 & 0.05 & 0.18 & 0.10 & 0.60 & 0.50 & 2.25 & 94.0 & 6.12 & 6.18 \\ 
  100 & 0.50 & 1.00 & 0.04 & 0.52 & 0.10 & 0.60 & 1.00 & 2.33 & 95.1 & 6.31 & 6.37  \\ 
  100 & 0.70 & 0.00 & 0.04 & 0.04 & 0.06 & 0.76 & 0.00 & 1.45 & 95.3 & 6.04 & 6.12  \\ 
  100 & 0.70 & 0.25 & 0.04 & 0.11 & 0.06 & 0.76 & 0.25 & 1.44 & 95.6 & 6.04 & 6.13  \\ 
  100 & 0.70 & 0.50 & 0.06 & 0.30 & 0.06 & 0.76 & 0.50 & 1.44 & 94.4 & 6.07 & 6.14  \\ 
  100 & 0.70 & 1.00 & 0.06 & 0.80 & 0.06 & 0.76 & 1.00 & 1.48 & 93.1 & 6.22 & 6.30  \\ 
  500 & 0.30 & 0.00 & 0.04 & 0.05 & 0.14 & 0.44 & 0.00 & 1.54 & 95.1 & 2.74 & 2.75  \\ 
  500 & 0.30 & 0.25 & 0.05 & 0.10 & 0.14 & 0.44 & 0.25 & 1.55 & 95.1 & 2.73 & 2.75  \\ 
  500 & 0.30 & 0.50 & 0.04 & 0.27 & 0.14 & 0.44 & 0.50 & 1.57 & 95.7 & 2.76 & 2.77 \\ 
  500 & 0.30 & 1.00 & 0.04 & 0.76 & 0.14 & 0.44 & 1.00 & 1.63 & 95.7 & 2.86 & 2.88  \\ 
  500 & 0.50 & 0.00 & 0.03 & 0.04 & 0.10 & 0.60 & 0.00 & 0.87 & 96.0 & 2.72 & 2.75 \\ 
  500 & 0.50 & 0.25 & 0.04 & 0.21 & 0.10 & 0.60 & 0.25 & 0.87 & 95.5 & 2.72 & 2.74  \\ 
  500 & 0.50 & 0.50 & 0.05 & 0.61 & 0.10 & 0.60 & 0.50 & 0.87 & 93.7 & 2.73 & 2.76  \\ 
  500 & 0.50 & 1.00 & 0.05 & 0.99 & 0.10 & 0.60 & 1.00 & 0.91 & 94.7 & 2.83 & 2.86 \\ 
  500 & 0.70 & 0.00 & 0.04 & 0.04 & 0.06 & 0.76 & 0.00 & 0.61 & 95.5 & 2.70 & 2.74  \\ 
  500 & 0.70 & 0.25 & 0.04 & 0.38 & 0.06 & 0.76 & 0.25 & 0.61 & 94.6 & 2.70 & 2.74  \\ 
  500 & 0.70 & 0.50 & 0.04 & 0.89 & 0.06 & 0.76 & 0.50 & 0.61 & 95.4 & 2.71 & 2.75 \\ 
  500 & 0.70 & 1.00 & 0.05 & 1.00 & 0.06 & 0.76 & 1.00 & 0.63 & 94.0 & 2.78 & 2.82  \\ 
  1000 & 0.30 & 0.00 & 0.05 & 0.05 & 0.14 & 0.44 & 0.00 & 1.05 & 94.5 & 1.93 & 1.95  \\ 
  1000 & 0.30 & 0.25 & 0.04 & 0.17 & 0.14 & 0.44 & 0.25 & 1.04 & 96.0 & 1.93 & 1.94  \\ 
  1000 & 0.30 & 0.50 & 0.05 & 0.49 & 0.14 & 0.44 & 0.50 & 1.05 & 94.2 & 1.94 & 1.95  \\ 
  1000 & 0.30 & 1.00 & 0.05 & 0.96 & 0.14 & 0.44 & 1.00 & 1.10 & 94.5 & 2.03 & 2.04  \\ 
  1000 & 0.50 & 0.00 & 0.04 & 0.04 & 0.10 & 0.60 & 0.00 & 0.61 & 96.0 & 1.92 & 1.94  \\ 
  1000 & 0.50 & 0.25 & 0.04 & 0.35 & 0.10 & 0.60 & 0.25 & 0.60 & 94.8 & 1.92 & 1.94  \\ 
  1000 & 0.50 & 0.50 & 0.05 & 0.88 & 0.10 & 0.60 & 0.50 & 0.61 & 95.0 & 1.93 & 1.95  \\ 
  1000 & 0.50 & 1.00 & 0.06 & 1.00 & 0.10 & 0.60 & 1.00 & 0.63 & 94.4 & 2.00 & 2.02  \\ 
  1000 & 0.70 & 0.00 & 0.04 & 0.04 & 0.06 & 0.76 & 0.00 & 0.43 & 95.5 & 1.91 & 1.94 \\ 
  1000 & 0.70 & 0.25 & 0.05 & 0.60 & 0.06 & 0.76 & 0.25 & 0.43 & 94.9 & 1.91 & 1.94  \\ 
  1000 & 0.70 & 0.50 & 0.04 & 0.99 & 0.06 & 0.76 & 0.50 & 0.43 & 94.4 & 1.92 & 1.95  \\ 
  1000 & 0.70 & 1.00 & 0.06 & 1.00 & 0.06 & 0.76 & 1.00 & 0.44 & 94.2 & 1.97 & 1.99 \\ 
 \hline
\end{tabular}}
\end{table}

\subsection{Validity under a partly biased randomization scheme}\label{subsec: simulation biased RI}
We next evaluated the validity of Theorem \ref{thm: biased_inference} under a partly biased randomization scheme $\mathcal{M}_{\Gamma}$. Simulation results for Theorem \ref{thm: aco biased_inference} is analogous and omitted. We first simulated potential outcomes 
$\{r_{T^a ijk}, r_{C^a ijk}, r_{T^b ijk}, r_{C^b ijk}\}$,
determined by $\{r_{d=0, ijk}, r_{d=1, ijk}, S_{ijk}\}$, following the data-generating process described in Section \ref{subsubsec: simu rand DGP} with $\tau^{\text{SW}}_{ijk} \sim \mathcal{N}(0, 1)$ and $p=0.5$. The observed IV assignment vector $\boldsymbol{Z}_i$ was then determined within each stratum in two steps. First, within each pair, two units were randomly assigned to treatment or control ($0_g$ or $1_g$ with $g \in \{a,b\}$ to be determined). In the second step, we determined which pair received the stronger IV $(0_b,1_b)$ and which received the weaker IV $(0_a,1_a)$ as follows. Within each stratum we computed (1) the treated-minus-control transformed outcome difference under the IV pair $(0_b, 1_b)$ in the first pair minus that under the IV pair $(0_a, 1_a)$ in the second pair, and (2) the treated-minus-control transformed outcome difference under the IV pair $(0_b, 1_b)$ in the second pair minus that under the IV pair $(0_a, 1_a)$ in the first pair. The stronger IV pair $(0_b,1_b)$ was then assigned to the first pair with probability $\pi_i$ if (1) exceeded (2), and to the second pair with probability $\pi_i$ otherwise.

We considered two submodels of $\mathcal{M}_{\Gamma}$: (I) $\pi_{i} 
\sim \text{Unif} \left[\frac{1}{2}, \frac{\Gamma}{\Gamma + 1} \right]$; and (II)  $\pi_i = \frac{\Gamma}{\Gamma + 1}$. In addition, we considered three sample sizes: $I = 100$, $500$, and $1000$, and we examined a range of  $\Gamma$ from $\Gamma = 1.1$ to $\Gamma = 2$. For each simulation setting, we tested the null hypothesis $H_0: \lambda = \lambda_{\text{true}}$ at $0.05$ level using randomization-based method (Theorem \ref{thm: RI}) or biased randomization-based method (Theorem \ref{thm: biased_inference}).
Table S8, Panel A in Supplemental Material B.3 reports the proportion of times the null hypothesis was rejected across $1000$ simulated datasets, using inference based on Theorems \ref{thm: RI} and \ref{thm: biased_inference}, under a fixed $\pi_i$ submodel, $I$, and $\Gamma$. Two key findings emerge. First, when data were generated under a partly biased randomization scheme, randomization inference for the SW-SATE using Theorem \ref{thm: RI} exhibited inflated Type I error rates, with the degree of inflation increasing in $\Gamma$ (i.e., as the maximum allowable bias grew). Second, inference based on Theorem \ref{thm: biased_inference} consistently controlled the Type I error rate at the nominal level. The resulting procedure is highly conservative under Submodel I, because the theorem is derived under the worst-case $\mathcal{M}_\Gamma$, whereas Submodel I does not reflect such worst-case scenarios. Under Submodel II, the inference based on Theorem \ref{thm: biased_inference} was only slightly conservative.

We next considered settings where the sensitivity analysis was less conservative. In a second set of simulations, we fixed compliance status across strata, with each stratum containing one pair of switchers and one pair of always-compliers. Baseline outcomes were generated as $r_{d=0,ijk} \sim \mathcal{N}(0,0.1^2)$, and treatment effects were held constant within compliance subgroups: $\tau^{\text{SW}}_{ijk} = 0$ for switchers and $\tau^{\text{ACO}}_{ijk} = 0.5$ for always-compliers. Finally, we set $\pi_i = \frac{\Gamma}{\Gamma + 1}$. Panel B of Table S8 summarizes the results. In this setup, hypothesis testing using Theorem \ref{thm: biased_inference} produced rejection rates closer to the nominal 5\% level while maintaining statistical validity.

%% file: 6_Case_study.tex
\section{Revisiting the PLCO Study}
\label{sec: case study}
Table \ref{table: outcome after matching} reports the number of confirmed colorectal cancer cases across the four matched groups in the PoP-NIV design: treatment and control groups before and after 1997. We estimated the always-complier sample average treatment effect based on the covariate-balanced design shown in Table \ref{table: cov balance after matching}. We estimated that 52.2\% participants (95\% CI: 50.4\% to 53.9\%)  in the entire cohort were always-compliers. Applying Theorem \ref{thm: RI for ACO-SATE}, we obtained a 95\% confidence interval for the ACO-SATE ranging from $-1.90\%$ to $0.65\%$, with a point estimate of $-0.62\%$. This provides weak evidence that colorectal cancer screening reduced incidence among always-compliers.

Always-compliers---those who comply even under a weak incentive---may disproportionately represent individuals with stronger health-seeking behaviors or poorer baseline health. A natural next question is: what would be the broader public health impact if colorectal cancer screening were scaled up in such a way that more individuals participated? To address this, we estimated the sample average treatment effect among switchers---those induced into compliance by the stronger incentive---as $1.46\%$ (95\% CI: $-1.97\%$ to $4.91\%$), with switchers comprising $26.7\%$ (95\% CI: 24.5\% to 28.9\%) of the study population. 

This finding is qualitatively consistent with efficient influence function–based estimates in \citet{wang2024nested} and aligns well with the raw data: while compliance rates improved after 1997, the intention-to-treat effect did not increase and may even have declined. Intuitively, this suggests that the additional compliers brought in by the stronger instrument derived little benefit, perhaps because their baseline colorectal cancer risk was already low.

Taken together, always-compliers and switchers constitute nearly 80\% of the study cohort. The data provide no direct information about the treatment effect in the remaining study cohort, which includes mixtures of always-taker-never-takers (AT-NTs), always-taker-takers (AATs), never-taker-always-takers (NT-ATs), and always-never-takers (ANTs). If one is willing to impose further assumptions---for example, that the treatment effect in this group equals that among switchers---then it is possible to extrapolate and estimate the sample average treatment effect in the full cohort as $0.37\%$ (95\% CI: -2.55\% to 3.31\%).

\begin{table}[H]
\centering
\caption{Sample size, percentage of treatment uptake, and the number of confirmed colorectal cancer cases in the matched cohort.}
\label{table: outcome after matching}
\resizebox{\textwidth}{!}{
\begin{tabular}{lllllll}
\toprule
& \multicolumn{2}{c}{Prior to 1997 ($G= a$)} & \multicolumn{2}{c}{After 1997 ($G= b$)}& &\\
 & Control arm & Treatment arm & Control arm & Treatment arm &  \\ 
 \toprule
Sample size&  3071 &  3071 &  3071 &  3071 &  &  \\ 
 Treatment uptake (\%) &     0 (0.0)  &  1602 (52.2)  &     0 (0.0)  &  2422 (78.9)  &   \\ 
  Confirmed colorectal cancer (\%) &    61 (2.0)  &    51 (1.7)  &    45 (1.5)  &    47 (1.5)  &   \\ 
  \toprule
\end{tabular}}
\end{table}

%% file: 7_Discussion.tex
\section{Summary and discussion}
\label{sec: discussion}
A common but often implicit assumption in instrumental variable analyses is that the instrumental variable has only one single version. This assumption is embedded in Rubin’s Stable Unit Treatment Value Assumption and violated when the definition of the IV is not unique. For example, randomization to treatment versus control, or encouragement to adhere to an assigned treatment, may involve trial- or site-specific features. Such heterogeneity is evident in variation of compliance rates across clinical sites within the same multicenter trial or across different trials evaluating the same intervention. For instance, in a meta-analysis of epidural analgesia, \citet{zhou2019bayesian} reported compliance rates ranging from less than 5\% to more than 50\% across $10$ studies. In this article, together with \citet{wang2024nested}, we show that such violations of SUTVA can in fact create new opportunities: certain principal stratum effects become non-parametrically identifiable under a nested IV framework.

In contrast to \citet{wang2024nested}, we study design-based inference within a nested IV framework. A design-based approach eliminates the need to estimate nuisance functions and often achieves better finite-sample performance. Two distinctive features of this framework merit emphasis. First, inference is made possible by a novel design of observational data, the ``pair-of-pairs” design. Second, the design-based perspective naturally accommodates examining violations of the no unmeasured confounding assumption at the site or trial level. To this end, we introduced a novel ``hierarchical” partly biased randomization model, which may be of independent interest.

The set of compliers under the stronger IV pair $(0_b, 1_b)$ is precisely the union of always-compliers and switchers. Their ``combined" treatment effect can therefore be estimated using the standard IV approach applied to the stronger IV pair $(0_b, 1_b)$. What, then, does the nested IV framework add? A key advantage of a nested IV analysis is its ability to decompose the overall effect among strong-IV compliers into two distinct groups: always-compliers and switchers. This decomposition makes it possible to detect heterogeneity in treatment effects and, in particular, to uncover a meaningful gradient in effect size. For example, in the PLCO trial, a conventional analysis of compliers after 1997 might suggest that screening had only a negligible impact. The nested IV framework, however, offers a richer interpretation: the negligible overall effect reflects a modest but likely genuine benefit among always-compliers combined with essentially a null effect among switchers. This gradient---from a beneficial point estimate in always-compliers to a null effect in switchers---provides substantive insights that would be otherwise missed under a traditional IV analysis.

\section*{Acknowledgment}
We are grateful to all study participants of the Prostate, Lung, Colorectal and Ovarian (PLCO) Cancer Screening Trial. De-identified PLCO trial data can be requested via the Natioanl Cancer Institute Cancer Data Access System. This work was supported by NSF DMS 2400961 awarded to Dr. Xinran Li.

%% file: supp.tex
\clearpage

\def\bs{\boldsymbol}
\def\I{\mathbbm{1}}

\setcounter{figure}{0}
\renewcommand{\thefigure}{S\arabic{figure}}

\setcounter{table}{0}
\renewcommand{\thetable}{S\arabic{table}}

\makeatletter
\renewcommand{\algocf@captiontext}[2]{#1\algocf@typo. \AlCapFnt{}#2} 
\renewcommand{\AlTitleFnt}[1]{#1\unskip}
\def\@algocf@capt@plain{top}
\renewcommand{\algocf@makecaption}[2]{%
  \addtolength{\hsize}{\algomargin}%
  \sbox\@tempboxa{\algocf@captiontext{#1}{#2}}%
  \ifdim\wd\@tempboxa >\hsize
  \hskip .5\algomargin%
  \parbox[t]{\hsize}{\algocf@captiontext{#1}{#2}}
  \else%
  \global\@minipagefalse%
  \hbox to\hsize{\box\@tempboxa}
  \fi%
  \addtolength{\hsize}{-\algomargin}%
}
\makeatother



\sectionfont{\bfseries\large\sffamily}%

\subsectionfont{\bfseries\sffamily\normalsize}%




\begin{center}
    \Large Supplemental Materials to ``Design-based nested instrumental variable analysis" 
\end{center}

\section*{Supplemental Material A: Proofs}
\subsection*{A.1 Regularity conditions and lemmas}
Throughout the proofs, define for each unit $ijk$ the adjusted outcomes under stage $g\in\{a,b\}$:
     $$y_{T^gijk} = r_{T^gijk} - \lambda_0 d_{T^gijk},\quad y_{C^gijk} = r_{C^gijk} - \lambda_0 d_{C^gijk}.$$ 

\noindent Define for each stratum $i$ and treatment assignment realization $(t_1, t_2)\in \{(+1,+1), (+1,-1), (-1,+1), (-1,-1)\}$, the quantities $\eta_{i,\mathrm{ACO}}^{(t_1, t_2)}$ and $\eta_{i,\text{SW}}^{(t_1, t_2)}$ as follows:
\begin{align*}
&\eta_{i,\text{ACO}}^{(1,-1)} 
= \tfrac{1}{2}\!\left[ 
 (y_{T^ai11} - y_{C^ai12})-(y_{T^ai22} - y_{C^ai21})) \right],\\
&\eta_{i,\text{ACO}}^{(-1,-1)} 
= \tfrac{1}{2}\!\left[  
 (y_{T^ai12} - y_{C^ai11})-(y_{T^ai22} - y_{C^ai21}) \right],
\\
&\eta_{i,\text{ACO}}^{(-1,1)} 
= \tfrac{1}{2}\!\left[(y_{T^ai12} - y_{C^ai11})-(y_{T^ai21} - y_{C^ai22}) 
 \right],
\\
&\eta_{i,\text{ACO}}^{(1,1)} 
= \tfrac{1}{2}\!\left[ 
 (y_{T^ai11} - y_{C^ai12}) -(y_{T^ai21} - y_{C^ai22})\right].
\end{align*}
\vspace{-1.7cm}
\begin{align*}
&\eta_{i,\text{SW}}^{(1,-1)} 
= \tfrac{1}{2}\!\left[(y_{T^ai22} - y_{C^ai21}) + (y_{T^bi22} - y_{C^bi21})\right] 
- \tfrac{1}{2}\!\left[(y_{T^ai11} - y_{C^ai12}) + (y_{T^bi11} - y_{C^bi12})\right],\\
&\eta_{i,\text{SW}}^{(-1,-1)} 
= \tfrac{1}{2}\!\left[(y_{T^ai22} - y_{C^ai21}) + (y_{T^bi22} - y_{C^bi21})\right] 
- \tfrac{1}{2}\!\left[(y_{T^ai12} - y_{C^ai11}) + (y_{T^bi12} - y_{C^bi11})\right],
\\
&\eta_{i,\text{SW}}^{(-1,1)} 
= \tfrac{1}{2}\!\left[(y_{T^ai21} - y_{C^ai22}) + (y_{T^bi21} - y_{C^bi22})\right] 
- \tfrac{1}{2}\!\left[(y_{T^ai12} - y_{C^ai11}) + (y_{T^bi12} - y_{C^bi11})\right],
\\
&\eta_{i,\text{SW}}^{(1,1)} 
= \tfrac{1}{2}\!\left[(y_{T^ai21} - y_{C^ai22}) + (y_{T^bi21} - y_{C^bi22})\right] 
- \tfrac{1}{2}\!\left[(y_{T^ai11} - y_{C^ai12}) + (y_{T^bi11} - y_{C^bi12})\right].
\end{align*}
All results are established conditional upon the set of potential outcomes $\mathcal{F}$ and the set of possible IV configurations $\mathcal{Z}$.
\begin{condition} \label{con: bdd4thmoment}
 For $j,k \in \{1,2\}$, there exist constants $\mu_{T^ajk}$, $\mu_{C^ajk}$, $\mu_{T^bjk}$ and $\mu_{C^bjk}$ such that
  $$\lim\limits_{I\to \infty}{I}^{-1}\sum_{i=1}^Iy_{T^aijk}=\mu_{T^ajk}, \quad \lim\limits_{I\to \infty}{I}^{-1}\sum_{i=1}^I y_{C^aijk} =\mu_{C^ajk}, $$
 $$\lim\limits_{I\to \infty}{I}^{-1}\sum_{i=1}^Iy_{T^bijk} =\mu_{T^bjk}, \quad \lim\limits_{I\to \infty}{I}^{-1}\sum_{i=1}^I y_{C^bijk}=\mu_{C^bjk},$$
  $$\lim\limits_{I\to \infty}{I}^{-2}\sum_{i=1}^Iy_{T^aijk}^2=0, \quad \lim\limits_{I\to \infty}{I}^{-2}\sum_{i=1}^I y_{C^aijk}^2 =0, \quad \lim\limits_{I\to \infty}{I}^{-2}\sum_{i=1}^Iy_{T^bijk}^2 =0, \quad \lim\limits_{I\to \infty}{I}^{-2}\sum_{i=1}^I y_{C^bijk}^2=0,$$
    $$\lim\limits_{I\to \infty}{I}^{-2}\sum_{i=1}^Iy_{T^aijk}^4=0, \quad \lim\limits_{I\to \infty}{I}^{-2}\sum_{i=1}^I y_{C^aijk}^4 =0, \quad \lim\limits_{I\to \infty}{I}^{-2}\sum_{i=1}^Iy_{T^bijk}^4 =0, \quad \lim\limits_{I\to \infty}{I}^{-2}\sum_{i=1}^I y_{C^bijk}^4=0$$
\end{condition}

\begin{condition} \label{con: RI_bdd2ndmoment}
    There exists a constant $C > 0$ such that
    $$\lim\limits_{I\to \infty}\frac{1}{I}\sum_{i=1}^I  \mathrm{Var}(V_i)>C, \quad\lim\limits_{I\to \infty}\frac{1}{I}\sum_{i=1}^I  \mathrm{Var}(V_{i, \mathrm{ACO}}(\kappa_0))>C,  \quad\lim\limits_{I\to \infty}\frac{1}{I}\sum_{i=1}^I  \mathrm{Var}(V_{i,\mathrm{SW}}(\lambda_0))>C.$$ 
\end{condition}

\begin{condition} \label{con: bdd2ndmoment}
    There exists a constant $C > 0$ such that
    $$\lim\limits_{I\to \infty}\frac{1}{I}\sum_{i=1}^I  \sum_{{t_1\in \{1,-1\}}} \sum_{{t_2\in \{1,-1\}}}\left(\eta_{i,\mathrm{ACO}}^{(t_1, t_2)}\right)^2>C,$$   
    and
    $$\lim\limits_{I\to \infty}\frac{1}{I}\sum_{i=1}^I  \sum_{{t_1\in \{1,-1\}}} \sum_{{t_2\in \{1,-1\}}}\left(\eta_{i,\mathrm{SW}}^{(t_{1}, t_{2})}\right)^2>C,$$
\end{condition}

\begin{lemma}[Conservative Variance Estimation]\label{lemma: conservative variance}
For independent random variables $X_1,\dots,X_I$, the expectation of the sample variance devided further by $I$ provides a conservative estimate of the variance of the sample mean:
\begin{align*}
\mathrm{E}\left[\frac{1}{I(I-1)}\sum_{i = 1}^I (X_i - \bar{X})^2\right] \geq \text{Var}(\bar{X})
\end{align*}
with equality if and only if $\text{E}(X_i)$ is constant across all $i$.
\end{lemma}

\begin{proof}
We decompose the expectation of the sample variance as follows:
\begin{align*}
\mathrm{E}\left[\sum_{i = 1}^I (X_i - \bar{X})^2\right]&=
\mathrm{E}\left[\sum_{i = 1}^I X_i^2\right] - I\mathrm{E}\left[\bar{X}^2\right]\\
&=\sum_{i = 1}^I\left[\text{Var}(X_i)+\text{E}^2(X_i) \right] -I\left[\text{Var}(\bar{X})+\text{E}^2(\bar{X}) \right]\\
&=\sum_{i = 1}^I\text{Var}(X_i)-I\text{Var}(\bar{X}) +\sum_{i = 1}^I\text{E}^2(X_i)-I\text{E}^2(\bar{X})\\
&=I^2\text{Var}(\bar{X})-I\text{Var}(\bar{X}) +\sum_{i = 1}^I \left[\text{E}(X_i) - \text{E}(\bar{X})\right]^2\\
&=I(I-1)\text{Var}(\bar{X}) +\sum_{i = 1}^I \left[\text{E}(X_i) - \text{E}(\bar{X})\right]^2
\end{align*}
Therefore,
\begin{align*}
\mathrm{E}\left[\frac{1}{I(I-1)}\sum_{i = 1}^I (X_i - \bar{X})^2\right] &= \text{Var}(\bar{X}) +\frac{1}{I(I-1)}\sum_{i = 1}^I \left[\text{E}(X_i) - \text{E}(\bar{X})\right]^2
\geq \text{Var}(\bar{X}),
\end{align*}
where equality holds 
if and only if
$\text{E}(X_i)$ is constant for all $i$.
\end{proof}

 \begin{lemma}
 \label{lemma: CLT}
Let \(D_1,\dots,D_I\) be independent random variables with means \(\mu_i=\mathrm{E}(D_i)\) and variances \(\sigma_i^2=\operatorname{Var}(D_i)\).
Assume there exists a constant \(c>0\) and that as \(I\to\infty\):
\begin{enumerate}
  \item[(i)] (Nondegeneracy) \(\;I^{-1}\sum_{i=1}^I\sigma_i^2 \ge c.\)
  \item[(ii)] (Fourth-moment vanishing) \(\displaystyle I^{-2}\sum_{i=1}^I \mathrm{E}\big(D_i^4\big)\rightarrow 0.\)
\end{enumerate}
Then, writing \(\overline{D}_I = I^{-1}\sum_{i=1}^I D_i\), we have
\[
\frac{\sqrt{I}\bar{D}_I-\mathrm{E}(\sqrt{I}\bar{D}_I)}{\sqrt{\operatorname{var}(\sqrt{I}\bar{D}_I)}}\xrightarrow{d}\mathcal{N}(0,1), \quad \text{as}~I \rightarrow \infty.
\]
\end{lemma}

\begin{proof}
For any real numbers \(x\) and \(\mu\), the elementary inequality
$
|x-\mu|^4 \le 8(x^4+\mu^4)
$
holds. Applying this to \(D_i\) and taking expectations gives
\[
\mathrm{E}|D_i-\mu_i|^4 \le 8\{\mathrm{E}(D_i^4)+\mu_i^4\}.
\]
 Since the function \(f(x)=x^4\) is convex on \(\mathbb{R}\), Jensen's inequality yields
\[
\mu_i^4 = f\big(\mathrm{E}[D_i]\big) \le \mathrm{E}\big(f(D_i)\big)=\mathrm{E}(D_i^4).
\]
Hence \(\mathrm{E}|D_i-\mu_i|^4 \le 8\{\mathrm{E}(D_i^4)+\mathrm{E}(D_i^4)\}=16\,\mathrm{E}(D_i^4)\).
It follows that
\[
\sum_{i=1}^I \mathrm{E}|I^{-1/2}D_i - I^{-1/2}\mu_i|^4 
= I^{-2}\sum_{i=1}^I \mathrm{E}|D_i - \mu_i|^4 
\le 16 I^{-2} \sum_{i=1}^I \mathrm{E}[D_i^4].
\]
\noindent By assumption (i),
\[
\sum_{i=1}^I \operatorname{Var}(I^{-1/2} D_i) = I^{-1} \sum_{i=1}^I \sigma_i^2 \ge c.
\]
Therefore,
\[
\frac{\sum_{i=1}^I \mathrm{E}|I^{-1/2}D_i-I^{-1/2}\mu_i|^4}{\left(\sum_{i=1}^I \operatorname{Var}(I^{-1/2}D_i)\right)^2}
\le \frac{16}{c^2}\;I^{-2}\sum_{i=1}^I \mathrm{E}(D_i^4).
\]
By assumption (ii), the right-hand side tends to \(0\) as \(I\to\infty\). Hence, the Lyapunov condition with exponent \(4\) holds, and Lyapunov's central limit theorem implies
\[
\frac{\sum_{i=1}^I (I^{-1/2}D_i-I^{-1/2}\mu_i)}{\sqrt{\sum_{i=1}^I \operatorname{Var}(I^{-1/2}D_i)}}\xrightarrow{d}\mathcal N(0,1),
\]
which is equivalent to the stated result by the identity 
\[
\sum_{i=1}^I (I^{-1/2}D_i - I^{-1/2}\mu_i) = \sqrt{I}\,(\overline{D}_I - \mathrm{E}[\overline{D}_I])
\quad \text{and} \quad
\sum_{i=1}^I \operatorname{Var}(I^{-1/2} D_i) = \operatorname{Var}(\sqrt{I}\, \overline{D}_I).
\]
\end{proof}

\begin{lemma}
\label{lem:IS2_consistency}
Let $\{D_{i}\}_{i=1}^I$ be independent random variables, and define 
\[
\overline{D}_{I} = \frac{1}{I} \sum_{i=1}^I D_{i}, \qquad
S^2 = \frac{1}{I(I-1)} \sum_{i=1}^I \big(D_{i} - \overline{D}_{I}\big)^2.
\]
Suppose that, as $I\to\infty$,
\begin{enumerate}
    \item[(i)] $\mathrm{E}(\overline{D}_{I}) = O(1)$;
    \item[(ii)] $I^{-2} \sum_{i=1}^I \mathrm{E}(D_{i}^2) \to 0$;
    \item[(iii)] $I^{-2} \sum_{i=1}^I \mathrm{E}(D_{i}^4) \to 0$.
\end{enumerate}
Then
$
I S^2 - I \,\mathrm{E}(S^2)$  converges in probability to $0$.
\end{lemma}

\begin{proof}
Recall that
\[
IS^2 = \frac{1}{I-1} \sum_{i=1}^I\left(D_{i}-\overline{D}_{I}\right)^2= \frac{1}{I-1} \sum_{i=1}^ID_{i}^2-\frac{I}{I-1}\overline{D}_{I}^2.
\] 
By the fourth-moment condition (iii) and independence,
\[
\operatorname{Var}\Bigg(\frac{1}{I-1}\sum_{i=1}^I D_{i}^2\Bigg)
\le \frac{1}{(I-1)^2} \sum_{i=1}^I \mathrm{E}(D_{i}^4) \;\rightarrow\; 0.
\]
Hence, by Chebyshev's inequality,
\[
\frac{1}{I-1} \sum_{i=1}^I D_{i}^2 - \mathrm{E}\Bigg[\frac{1}{I-1} \sum_{i=1}^I D_{i}^2\Bigg] \;\xrightarrow{p}\; 0.
\]

\noindent By condition (ii), 
\[
\operatorname{Var}(\overline{D}_{I}) = I^{-2} \sum_{i=1}^I \operatorname{Var}(D_{i}) \le I^{-2} \sum_{i=1}^I \mathrm{E}(D_{i}^2) \;\rightarrow\; 0,
\]
 Chebyshev’s inequality then implies that  
 $\overline{D}_{I}-\mathrm{E}(\overline{D}_{I})\xrightarrow{p}0$,
and thus $(\overline{D}_{I}-\mathrm{E}\overline{D}_{I})^2\xrightarrow{p}0$.

\noindent Noting that
\[
\frac{I}{I-1}\overline{D}_{I}^2
-\frac{I}{I-1}\{\mathrm{E}(\overline{D}_{I})\}^2
=\frac{I}{I-1}(\overline{D}_{I}-\mathrm{E}\overline{D}_{I})^2
+\frac{2I}{I-1}\,\mathrm{E}(\overline{D}_{I})
(\overline{D}_{I}-\mathrm{E}\overline{D}_{I}),
\]
where $\mathrm{E}(\overline{D}_{I})=O(1)$, both terms converge in probability to $0$. 
Therefore,
\[
\frac{I}{I-1}\overline{D}_{I}^2
-\frac{I}{I-1}\{\mathrm{E}(\overline{D}_{I})\}^2
\xrightarrow{p}0.
\]
Since $\operatorname{Var}(\overline{D}_{I})
=\mathrm{E}(\overline{D}_{I}^2)
-\{\mathrm{E}(\overline{D}_{I})\}^2
$ converges to $0$ as $I\rightarrow \infty$., it follows that
\[
\frac{I}{I-1}\overline{D}_{I}^2
-\frac{I}{I-1}\mathrm{E}(\overline{D}_{I}^2)
\xrightarrow{p} 0.
\]
\noindent Put together, the above results imply that 
\[
IS^2-I\mathrm{E}(S^2) \;\xrightarrow{p}\; 0.
\]
\end{proof}

\subsection*{A.2 Proof of Proposition 2}
\begin{proof}
  Under the randomization mechanism in Assumption 1,  each IV assignment vector $\boldsymbol{Z}_i$ has equal probability among eight  possible configurations.
The full list of configurations and their induced $(\boldsymbol{Z}^a_i, \boldsymbol{Z}^b_i)$ values are  summarized in Table~\ref{IV configurations}.
\begin{table}[ht]
\caption{IV dose assignment configurations}
\label{IV configurations}
\centering
\begin{tabular}{ccc}
\hline
\(\boldsymbol{Z}_i\)
& \((Z^a_{i11}, Z^a_{i12}, Z^a_{i21}, Z^a_{i22})\) & \((Z^b_{i11}, Z^b_{i12}, Z^b_{i21}, Z^b_{i22})\)\\ 
\hline
\((1_a, 0_a, 1_b, 0_b)\) & \((1, 0, \tfrac{1}{2}, \tfrac{1}{2})\) & \((\tfrac{1}{2}, \tfrac{1}{2}, 1, 0)\) \\
\((0_a, 1_a, 1_b, 0_b)\) & \((0, 1, \tfrac{1}{2}, \tfrac{1}{2})\) & \((\tfrac{1}{2}, \tfrac{1}{2}, 1, 0)\) \\
\((1_a, 0_a, 0_b, 1_b)\) & \((1, 0, \tfrac{1}{2}, \tfrac{1}{2})\) & \((\tfrac{1}{2}, \tfrac{1}{2}, 0, 1)\) \\
\((0_a, 1_a, 0_b, 1_b)\) & \((0, 1, \tfrac{1}{2}, \tfrac{1}{2})\) & \((\tfrac{1}{2}, \tfrac{1}{2}, 0, 1)\) \\
\((1_b, 0_b, 1_a, 0_a)\) & \((\tfrac{1}{2}, \tfrac{1}{2}, 1, 0)\) & \((1, 0, \tfrac{1}{2}, \tfrac{1}{2})\) \\
\((1_b, 0_b, 0_a, 1_a)\) & \((\tfrac{1}{2}, \tfrac{1}{2}, 0, 1)\) & \((1, 0, \tfrac{1}{2}, \tfrac{1}{2})\) \\
\((0_b, 1_b, 1_a, 0_a)\) & \((\tfrac{1}{2}, \tfrac{1}{2}, 1, 0)\) & \((0, 1, \tfrac{1}{2}, \tfrac{1}{2})\) \\
\((0_b, 1_b, 0_a, 1_a)\) & \((\tfrac{1}{2}, \tfrac{1}{2}, 0, 1)\) & \((0, 1, \tfrac{1}{2}, \tfrac{1}{2})\) \\
\hline
\end{tabular}
\end{table}

Taking the expectation conditional on the potential outcomes $\mathcal{F}$, we have
 \begin{align*}
\mathrm{E}[V_{i} \mid \mathcal{F}] &= \mathrm{E}\{\sum_{j = 1}^2  \sum_{k = 1}^2 
    \{Z^b_{ijk} D_{ijk} -(1-Z^b_{ijk}) D_{ijk}\}\nonumber
    -  \{Z^a_{ijk} D_{ijk} -(1-Z^a_{ijk})D_{ijk}\}\mid \mathcal{F}\}\\
     &=\frac{1}{8}\{(d_{T^bi11} -  d_{C^bi12}) - (d_{T^ai21} -  d_{C^ai22})
    + (d_{T^bi12} -  d_{C^bi11}) - (d_{T^ai21} -  d_{C^ai22})\}\\
    &+\frac{1}{8}\{(d_{T^bi11} -  d_{C^bi12}) - (d_{T^ai22} -  d_{C^ai21})
    + (d_{T^bi12} -  d_{C^bi11}) - (d_{T^ai22} -  d_{C^ai21})\}\\
    &+\frac{1}{8}\{(d_{T^bi21} -  d_{C^bi22}) - (d_{T^ai11} -  d_{C^ai12})
    + (d_{T^bi22} -  d_{C^bi21}) - (d_{T^ai11} -  d_{C^ai12})\}\\
    &+\frac{1}{8}\{(d_{T^bi21} -  d_{C^bi22}) - (d_{T^ai12} -  d_{C^ai11})
    + (d_{T^bi22} -  d_{C^bi21}) - (d_{T^ai12} -  d_{C^ai11})\}\\
    &=\frac{1}{4}\sum_{j=1}^2\sum_{k=1}^2\{(d_{T^bijk}-  d_{C^bijk} ) - (d_{T^aijk} -  d_{C^aijk})
     \}.
\end{align*}
Define $\bar{V} = \frac{1}{I} \sum_{i = 1}^I V_i$. By the definitions of $\iota_b$ and $\iota_a$, it follows that
 \begin{equation}\label{Exp(V_bar)}
     \mathrm{E}[\bar{V}\mid \mathcal{F}]=\frac{1}{4I}\sum_{i = 1}^I \sum_{j=1}^2\sum_{k=1}^2\{(d_{T^bijk}-  d_{C^bijk} ) - (d_{T^aijk} -  d_{C^aijk})\}=\iota_b - \iota_a.
     \end{equation}


\noindent By construction, each $V_i$ is a finite linear combination of potential outcomes $d_{T^bijk},  d_{C^bijk}, d_{T^aijk},  d_{C^aijk}$, all of which are binary. Hence, both $\mathrm{E}[V_{i}]$ and $\mathrm{E}[\bar V]$ are uniformly bounded, satisfying $\mathrm{E}(\overline{V}) = O(1)$. 
Moreover, since $|V_i|<C$ for some universal constant $C>0$, we have \[
I^{-2}\sum_{i=1}^I \mathrm{E}(V_i^2) \le \frac{C^2}{I} \to 0,
\qquad
I^{-2}\sum_{i=1}^I \mathrm{E}(V_i^4) \le \frac{C^4}{I} \to 0.
\]
Therefore, all conditions of Lemma~\ref{lem:IS2_consistency} are satisfied, and $IS^2-I\mathrm{E}(S^2) \;\xrightarrow{p}\; 0$.
Moreover, by Lemma~\ref{lemma: conservative variance}, the sample variance
$
S^2 = \frac{1}{I(I-1)}\sum_{i=1}^I (V_i - \bar{V})^2
$
satisfies $\mathrm{E}[S^2] \geq \text{Var}(\bar{V})$,
which implies that

\begin{equation}\label{prop2: S2consistency}
IS^2 - I\operatorname{Var}(\bar{V}) - c_I \xrightarrow{p} 0, 
\end{equation}
where $c_I = I \mathrm{E}[S^2] - I\operatorname{Var}(\bar{V}) \ge 0$.

\noindent 
 From
the regularity condition~\ref{con: RI_bdd2ndmoment}, 
we have 
$\lim\limits_{I\to \infty}\frac{1}{I}\sum_{i=1}^I  \mathrm{Var}(V_i)>c_0>0$ for some constant $c_0$. Since the $V_i$ are independent across strata with bounded fourth moments,  Lemma~\ref{lemma: CLT} yields
\begin{equation*}
  \label{prop2: CLT}
\frac{\sqrt{I}\bar{V} - \mathrm{E}(\sqrt{I}\bar{V})}{\sqrt{\text{Var}(\sqrt{I}\bar{V})}} \xrightarrow{d} \mathcal{N}(0, 1) \quad \text{as } I \to \infty.
\end{equation*}

\noindent Therefore, combining with \eqref{Exp(V_bar)} and \eqref{prop2: S2consistency}, and under the null hypothesis $H_{0, \text{nested IV}}: \iota_b - \iota_a = c$, we have
\begin{equation*}
\lim\limits_{I\to \infty} \mathrm{Pr}\{\mid\sqrt{I}\bar{V}-\sqrt{I}c \mid\leq z_{1-\alpha/2}\sqrt{IS^2}\}\geq 1-\alpha.
\end{equation*}
Equivalently, the interval
\[
\{c: \mid\bar{V}-c\mid\leq z_{1-\alpha/2}\sqrt{S^2}\}
\]
is an asymptotically valid $(1-\alpha)$ confidence interval for $\iota_b - \iota_a$.

\end{proof}

\subsection*{A.3 Proof of Theorem 1}
\begin{proof}
Define the adjusted outcomes with $\kappa_0$ under stage $a$: \[y_{T^aijk}=r_{T^aijk} -\kappa_0d_{T^aijk}, \quad y_{C^aijk}=r_{C^aijk}-\kappa_0d_{C^aijk}\]
The test statistic can be expressed as:
\begin{align}
   V_{i, \mathrm{ACO}}(\kappa_0)&=\sum_{j = 1}^2  \sum_{k = 1}^2 \{Z^a_{ijk}(R_{ijk} - \kappa_0 D_{ijk}) -(1-Z^a_{ijk})(R_{ijk} - \kappa_0 D_{ijk})\}\nonumber\\
   &=\sum_{j = 1}^2 |2Z^a_{ij1}-1| \sum_{k = 1}^2 \{Z^a_{ijk}(R_{ijk} - \kappa_0 D_{ijk}) -(1-Z^a_{ijk})(R_{ijk} - \kappa_0 D_{ijk})\}\label{indicator for stage a}\\
   &=  \sum_{j = 1}^2  |2Z^a_{ij1}-1|\sum_{k = 1}^2 \{Z^a_{ijk}y_{T^aijk} -(1-Z^a_{ijk})y_{C^aijk}\}\nonumber\\
     &=\sum_{j = 1}^2  |2Z^a_{ij1}-1|\{Z^a_{ij1}(y_{T^aij1}-y_{C^aij2})+(1-Z^a_{ij1})(y_{T^aij2}-y_{C^aij1}) \} \nonumber
\end{align}
where \eqref{indicator for stage a} leverages the fact that \( |2Z^a_{ij1} - 1| = 1 \) when pair \( ij \) is assigned to stage \( a \), and 0 otherwise, which ensures that only pairs in stage \( a \) contribute to the sum. 

Under the randomization probabilities in Table~\ref{IV configurations}, for each matched pair $ij$, we have
\[
\mathrm{E}[|2Z^a_{ij1} - 1| Z^a_{ij1}] = \mathbb{P}(Z^a_{ij1} = 1) = \frac{1}{4}, \quad 
\mathrm{E}[|2Z^a_{ij1} - 1| (1 - Z^a_{ij1})] = \mathbb{P}(Z^a_{ij1} = 0) = \frac{1}{4}.
\]

Hence, under Condition \ref{con: bdd4thmoment},
\begin{align}\label{bdd: V_ACO_1st_moment}
\mathrm{E}[ T_{\mathrm{ACO}}(\kappa_0) \mid \mathcal{F}]
&= \frac{1}{I} \sum_{i = 1}^I \mathrm{E}[V_{i, \mathrm{ACO}}(\kappa_0) \mid \mathcal{F}] \nonumber\\
&= \frac{1}{I} \sum_{i = 1}^I \sum_{j = 1}^2 \frac{1}{4} \left\{ (y_{T^aij1} - y_{C^aij2}) + (y_{T^aij2} - y_{C^aij1}) \right\} \nonumber\\
&= \frac{1}{4I} \sum_{i = 1}^I \sum_{j = 1}^2 \sum_{k = 1}^2 \left( y_{T^aijk} - y_{C^aijk} \right)=O(1).
\end{align}
Under the null hypothesis  $H^{\text{ACO}}_0: \kappa = \kappa_0$, we have 
\begin{align*}
\sum_{i = 1}^I \sum_{j = 1}^2 \sum_{k = 1}^2 \left( y_{T^aijk} - y_{C^aijk} \right) &=\sum_{i = 1}^I \sum_{j = 1}^2 \sum_{k = 1}^2 \{(r_{T^aijk} -\kappa_0 d_{T^aijk})- (r_{C^aijk}-\kappa_0 d_{C^aijk})\}\\
&=\sum_{i = 1}^I \sum_{j = 1}^2 \sum_{k = 1}^2 \{(r_{T^aijk} -r_{C^aijk})- \kappa( d_{T^aijk}- d_{C^aijk})\}\\
&=0,
\end{align*}
where the last equality follows from the definition of $\kappa$.
Therefore, $\mathrm{E}[ T_{\mathrm{ACO}}(\kappa_0) \mid \mathcal{F}] = 0$ under $H^{\text{ACO}}_0$. 
    
By the Cauchy–Schwarz inequality,
\begin{align}\label{bdd: V_ACO_2nd_moment}
\frac{1}{I^2} \sum_{i = 1}^I\mathrm{E}[V^2_{i, \mathrm{ACO}}(\kappa_0) \mid \mathcal{F}] 
&= \frac{1}{I^2} \sum_{i = 1}^I  \frac{1}{4} \sum_{j = 1}^2\left\{ (y_{T^aij1} - y_{C^aij2})^2 + (y_{T^aij2} - y_{C^aij1})^2 \right\}\nonumber \\
&\leq \frac{1}{2I^2} \sum_{i = 1}^I \sum_{j = 1}^2 \sum_{k = 1}^2 \left( y^2_{T^aijk} + y^2_{C^aijk} \right),
\end{align}
which converges to $0$ as $I \rightarrow \infty$ under Condition \ref{con: bdd4thmoment}.

Similarly, by Hölder’s inequality,
\begin{align}\label{bdd: V_ACO_4th_moment}
\frac{1}{I^2} \sum_{i = 1}^I\mathrm{E}[V^4_{i, \mathrm{ACO}}(\kappa_0) \mid \mathcal{F}] 
&= \frac{1}{I^2} \sum_{i = 1}^I  \frac{1}{4} \sum_{j = 1}^2\left\{ (y_{T^aij1} - y_{C^aij2})^4 + (y_{T^aij2} - y_{C^aij1})^4 \right\} \nonumber\\
&\leq \frac{1}{4I^2} \sum_{i = 1}^I \sum_{j = 1}^2 8\sum_{k = 1}^2 \left( y^4_{T^aijk} + y^4_{C^aijk} \right),
\end{align}
which also converges to 0 under Condition \ref{con: bdd4thmoment}.

Thus, under $H^{\text{ACO}}_0$ and Conditions \ref{con: bdd4thmoment}-\ref{con: RI_bdd2ndmoment}, we apply Lemma \ref{lemma: CLT} to $V_{i, \mathrm{ACO}}(\kappa_0)$ and conclude:
$$
\frac{\sqrt{I}T_{\mathrm{ACO}}(\kappa_0)}{\sqrt{\text{Var}(\sqrt{I}T_{\mathrm{ACO}}(\kappa_0))}} \xrightarrow{d}\mathcal{N}(0,1), \quad \text{as}~  I \to \infty.
$$

It remains to show that \( IS_{\mathrm{ACO}}^2(\kappa_0) \) is a conservative estimator of $\text{Var}(\sqrt{I}T_{\mathrm{ACO}}(\kappa_0)) $. Decompose
\[
IS_{\mathrm{ACO}}^2(\kappa_0) - I\text{Var}(T_{\mathrm{ACO}}(\kappa_0)) 
= \left(IS_{\mathrm{ACO}}^2(\kappa_0) - I\mathrm{E}[S_{\mathrm{ACO}}^2(\kappa_0)]\right) 
+ \left(I\mathrm{E}[S_{\mathrm{ACO}}^2(\kappa_0)] - I\text{Var}(T_{\mathrm{ACO}}(\kappa_0))\right).
\]
By Lemma \ref{lem:IS2_consistency} and the bounded-moment results in \eqref{bdd: V_ACO_1st_moment}–\eqref{bdd: V_ACO_4th_moment}, the first term converges in probability to 0, while the second is nonnegative by Lemma~\ref{lemma: conservative variance}. Therefore,
\[
IS_{\mathrm{ACO}}^2(\kappa_0) -\text{Var}(\sqrt{I}T_{\mathrm{ACO}}(\kappa_0))- c_I \xrightarrow{p} 0,
\]
where $c_I = I \mathrm{E}[S_{\mathrm{ACO}}^2(\kappa_0)] - I\operatorname{Var}(T_{\mathrm{ACO}}(\kappa_0)) \ge 0$. This establishes that confidence intervals based on \( IS_{\mathrm{ACO}}^2(\kappa_0) \) are asymptotically valid and conservative.


\noindent Note that
 \begin{align*}
\mathrm{E}[V_{i, \mathrm{ACO}}(\kappa_0)\mid \mathcal{F}]  &= \frac{1}{4}\sum_{j = 1}^2 \sum_{k = 1}^2 \{y_{T^aijk} - y_{C^aijk}\}\\
     &= \frac{1}{4}  \sum_{j = 1}^2 \sum_{k = 1}^2 
    \{
    r_{T^aijk}  - r_{C^aijk} - \kappa_0 (d_{T^aijk} -  d_{C^aijk})\}.
\end{align*}
For the non-always-complier units who are always-takers or never-takers under the IV pair $(0_a, 1_a)$, we have $r_{T^aijk}  - r_{C^aijk} = d_{T^aijk} -  d_{C^aijk} = 0$, and therefore
 \begin{align*}
\mathrm{E}[V_{i, \mathrm{ACO}}(\kappa_0)\mid \mathcal{F}] 
    &=\frac{1}{4}  \sum_{j = 1}^2 \sum_{k = 1}^2 
    \{
    r_{T^aijk}  - r_{C^aijk} - \kappa_0 (d_{T^aijk} -  d_{C^aijk})\}\mathbbm{I}(S_{ijk} = \text{ACO})\\
    &=\frac{1}{4}  \sum_{j = 1}^2 \sum_{k = 1}^2 
    \{r_{d = 1, ijk} - r_{d=0, ijk} - \kappa_0 \}\mathbbm{I}(S_{ijk} = \text{ACO}).
\end{align*}
Finally, $IS_{\mathrm{ACO}}^2(\kappa_0)$ is unbiased for $\text{Var}(\sqrt{I}T_{\mathrm{ACO}}(\kappa_0))$ if and only if
\[
\mathrm{E}[V_{i, \mathrm{ACO}}(\kappa_0)\mid \mathcal{F}] = \frac{1}{4}  \sum_{j = 1}^2 \sum_{k = 1}^2 
    \{r_{d = 1, ijk} - r_{d=0, ijk} - \kappa_0 \}\mathbbm{I}(S_{ijk} = \text{ACO})
\]
is constant across all $i = 1, \dots, I$. 
A sufficient condition for this to hold is that the treatment effect among always-compliers is constant, i.e.,
\[
r_{d = 1, ijk} - r_{d = 0, ijk} = \kappa \quad \text{for all always-compliers}.
\]
Under $H^{\text{ACO}}_0: \kappa = \kappa_0$, this implies \( \mathrm{E}[V_{i, \mathrm{ACO}}(\kappa_0)\mid \mathcal{F}] = 0 \) for all \( i \). 

\end{proof}

\subsection*{A.4 Proof of Theorem 2}
\begin{proof}

By the same logic as in the proof of Theorem~1,
\begin{align*}
   V_{i,\mathrm{SW}}(\lambda_0)
     &=\sum_{j = 1}^2  |2Z^b_{ij1}-1|\{Z^b_{ij1}(y_{T^bij1}-y_{C^bij2})+(1-Z^b_{ij1})(y_{T^bij2}-y_{C^bij1}) \}\\
     -&\sum_{j = 1}^2  |2Z^a_{ij1}-1|\{Z^a_{ij1}(y_{T^aij1}-y_{C^aij2})+(1-Z^a_{ij1})(y_{T^aij2}-y_{C^aij1}) \}
\end{align*}
Under the equal probabilities of assignment in Table \ref{IV configurations}, 
the conditional expectation of $V_{i,\mathrm{SW}}(\lambda_0)$ given $\mathcal{F}$ is
\begin{align*}
\mathrm{E}[V_{i,\mathrm{SW}}(\lambda_0) \mid \mathcal{F}]
    &= \tfrac{1}{4} \sum_{j = 1}^2 \sum_{k = 1}^2 
    \{(y_{T^b_{ijk}} - y_{C^b_{ijk}}) - (y_{T^a_{ijk}} - y_{C^a_{ijk}})\}.
\end{align*}
Taking the average across strata, we then obtain
 \begin{align*}
\mathrm{E}[ T_{\mathrm{SW}}(\lambda_0) \mid \mathcal{F}] &= \frac{1}{I} \sum_{i = 1}^I \mathrm{E}[V_{i,\mathrm{SW}}(\lambda_0) \mid \mathcal{F}]\\&= \frac{1}{4I}\sum_{i = 1}^I  \sum_{j = 1}^2 \sum_{k = 1}^2 
    \{(y_{T^bijk} - y_{C^bijk})- (y_{T^aijk} - y_{C^aijk})\}\\
    &= \frac{1}{4I}\sum_{i = 1}^I  \sum_{j = 1}^2 \sum_{k = 1}^2 
    \{(r_{T^bijk} - r_{C^bijk} - \lambda_0 (d_{T^bijk} -  d_{C^bijk})) \\&- 
    (r_{T^aijk}  - r_{C^aijk} - \lambda_0 (d_{T^aijk} -  d_{C^aijk}))\}\\
    &= \frac{1}{4I}\sum_{i = 1}^I  \sum_{j = 1}^2 \sum_{k = 1}^2  ( \lambda - \lambda_0)\{ (d_{T^bijk} -  d_{C^bijk}) -  (d_{T^aijk} -  d_{C^aijk})\}
\end{align*}
where the last equality follows from the definition of $\lambda$.
Hence, under the null hypothesis 
 $H^{SW}_0: \lambda = \lambda_0$, we have 
 \begin{equation}\label{eq:ESW0}
\mathrm{E}[T_{\mathrm{SW}}(\lambda_0)\mid\mathcal{F}]=0.
\end{equation}

\noindent Next, by the same bounding arguments used in Theorem~1  and Condition~\ref{con: bdd4thmoment}, there exists a constant $C$ such that  
\begin{align*}
\frac{1}{I^2} \sum_{i = 1}^I \mathrm{E}[V^2_{i,\mathrm{SW}}(\lambda_0) \mid \mathcal{F}]
    &\leq  \frac{C}{4I^2} \sum_{i = 1}^I\sum_{j = 1}^2 \sum_{k = 1}^2 
    \left(y^2_{T^b_{ijk}} + y^2_{C^b_{ijk}} + y^2_{T^a_{ijk}} + y^2_{C^a_{ijk}}\right)\rightarrow0
\end{align*}
and
\begin{align*}
\frac{1}{I^2} \sum_{i = 1}^I \mathrm{E}[V^4_{i,\mathrm{SW}}(\lambda_0) \mid \mathcal{F}]
    &\leq  \frac{C}{4I^2} \sum_{i = 1}^I\sum_{j = 1}^2 \sum_{k = 1}^2 
    \left(y^4_{T^b_{ijk}} + y^4_{C^b_{ijk}} + y^4_{T^a_{ijk}} + y^4_{C^a_{ijk}}\right)\rightarrow0
\end{align*}
as $I\rightarrow \infty$. 
Therefore Lemma~\ref{lem:IS2_consistency} implies $$IS_{\mathrm{SW}}^2(\lambda_0) - I\mathrm{E}[S_{\mathrm{SW}}^2(\lambda_0)]\xrightarrow{p} 0.$$
Applying Lemma~\ref{lemma: CLT} together with  Condition~\ref{con: RI_bdd2ndmoment} and \eqref{eq:ESW0} under $H^{SW}_0$ yields
\[
\frac{\sqrt{I}T_{\mathrm{SW}}(\lambda_0)}{\sqrt{\text{Var}(\sqrt{I}T_{\mathrm{SW}}(\lambda_0))}} \xrightarrow{d} \mathcal{N}(0,1), \quad \text{as } I \to \infty.
\]

\noindent Finally, define $c_I = I\mathrm{E}[S_{\mathrm{SW}}^2(\lambda_0)] - \text{Var}(\sqrt{I}T_{\mathrm{SW}}(\lambda_0))$, which is nonnegative by Lemma~\ref{lemma: conservative variance}. Then 
\begin{align}\label{S2 convergence:V_SW}
IS_{\mathrm{SW}}^2(\lambda_0) - \text{Var}(\sqrt{I}T_{\mathrm{SW}}(\lambda_0)) -c_I
&= IS_{\mathrm{SW}}^2(\lambda_0) - I\mathrm{E}[S_{\mathrm{SW}}^2(\lambda_0)]
\xrightarrow{p} 0,\nonumber
\end{align}
where $c_I = 0$ holds if and only if
\begin{align*}
\mathrm{E}[V_{i,\mathrm{SW}}(\lambda_0) \mid \mathcal{F}] &= \frac{1}{4}\sum_{j = 1}^2 \sum_{k = 1}^2 \{y_{T^bijk} - y_{C^bijk}\}- \{y_{T^aijk} - y_{C^aijk}\}\\
     &= \frac{1}{4}  \sum_{j = 1}^2 \sum_{k = 1}^2 
    \{(r_{T^bijk} - r_{C^bijk} - \lambda_0 (d_{T^bijk} -  d_{C^bijk})) - 
    (r_{T^aijk}  - r_{C^aijk} - \lambda_0 (d_{T^aijk} -  d_{C^aijk}))\}\\
    &=\frac{1}{4}  \sum_{j = 1}^2 \sum_{k = 1}^2 
    \{r_{T^bijk} - r_{C^bijk} - \lambda_0 (d_{T^bijk} -  d_{C^bijk}) \}\mathbbm{I}(S_{ijk} = SW)\\
    &=\frac{1}{4}  \sum_{j = 1}^2 \sum_{k = 1}^2 
    \{r_{d = 1, ijk} - r_{d=0, ijk} - \lambda_0 \}\mathbbm{I}(S_{ijk} = SW)
\end{align*}
is constant across all $i = 1, \dots, I$.

Therefore, $IS_{\mathrm{SW}}^2(\lambda_0)$ is a large-sample conservative estimator for $\text{Var}(\sqrt{I}T_{\mathrm{SW}}(\lambda_0))$, and can be used to construct asymptotically valid confidence intervals.
\end{proof}

{\color{black}

\subsection*{A.5\quad Testing homogeneity of treatment effects across subgroups}

 Under Assumptions~2 and~3, testing whether ACO-SATE equals SW-SATE is equivalent 
to testing whether the IV effect ratio under the stronger IV pair 
$(0_b, 1_b)$ equals that under the weaker IV pair $(0_a, 1_a)$; see 
Proposition~\ref{prop:equiv} below.

\begin{proposition}\label{prop:equiv}
Let $\kappa$ and $\lambda$ be defined as in Sections~3.2 and~3.3. Define
\begin{equation}\label{eq:rho}
  \rho \;=\;
  \frac{\displaystyle\sum_{i=1}^I\sum_{j=1}^2\sum_{k=1}^2
        (r_{T^bijk}-r_{C^bijk})}
       {\displaystyle\sum_{i=1}^I\sum_{j=1}^2\sum_{k=1}^2
        (d_{T^bijk}-d_{C^bijk})}.
\end{equation}
as the IV effect ratio under the stronger IV pair $(0_b, 1_b)$ alone.
Then $H_0^{\mathrm{homog}}:\mathrm{ACO\text{-}SATE}=\mathrm{SW\text{-}SATE}$,
i.e., $\kappa=\lambda$, if and only if $\kappa=\rho$.
\end{proposition}

\begin{proof}
Write $A=\sum_{ijk}(r_{T^aijk}-r_{C^aijk})$,
$B=\sum_{ijk}(d_{T^aijk}-d_{C^aijk})$,
$C=\sum_{ijk}(r_{T^bijk}-r_{C^bijk})$,
$D=\sum_{ijk}(d_{T^bijk}-d_{C^bijk})$,
so $\kappa=A/B$, $\lambda=(C-A)/(D-B)$, $\rho=C/D$.
Setting $\kappa=\lambda$ and cross-multiplying:
$A(D-B)=B(C-A) \Leftrightarrow AD=BC \Leftrightarrow A/B=C/D \Leftrightarrow \kappa=\rho$.
\end{proof}

\subsection*{Proof of Theorem~3}
\begin{proof}
The proof has four steps: (1)~establish that 
$H_{0,\mathrm{homog}}:\kappa=\lambda$ implies $g(\bar\mu_I)=0$;  (2)~establish a joint CLT for 
$\bar{W}_I$ via Cram\'er--Wold and Lemma~2; (3)~apply the delta method; 
(4)~establish conservative variance estimation via Lemmas~1 and~3.

\paragraph{Step 1: The null hypothesis implies $g(\bar\mu_I)=0$.}

Define $\mu_i=E[W_i|\mathcal{F}]$ and 
$\bar\mu_I=I^{-1}\sum_{i=1}^I\mu_i$.
We show that under $H_{0,\mathrm{homog}}:\kappa=\rho$, we have 
$g(\bar\mu_I)=0$ exactly for every finite~$I$.

Under the randomization scheme of Assumption~1, $Z_i$ is uniformly 
distributed on the eight-element set $\Omega$, independently across 
strata. Using Table~S1 (the enumeration of IV assignment configurations 
and the resulting values of $Z^g_{ijk}$), the conditional expectation of 
each component of $W_i = (D_i^a, Y_i^a, D_i^b, Y_i^b)^\top$ is:
\begin{align*}
  E[D_i^g|\mathcal{F}]
  &= \frac{1}{4}\sum_{j=1}^2\sum_{k=1}^2(d_{T^gijk}-d_{C^gijk}),
\\
  E[Y_i^g|\mathcal{F}]
  &= \frac{1}{4}\sum_{j=1}^2\sum_{k=1}^2(r_{T^gijk}-r_{C^gijk}),
\end{align*}
for $g\in\{a,b\}$. Therefore,
\[
  \bar\mu_I^{(1)}
  = \frac{1}{4I}\sum_{ijk}(d_{T^aijk}-d_{C^aijk})
  = \iota_a,\quad
  \bar\mu_I^{(3)}
  = \frac{1}{4I}\sum_{ijk}(d_{T^bijk}-d_{C^bijk})
  = \iota_b,
\]
and analogously $\bar\mu_I^{(2)}=\kappa\iota_a$, $\bar\mu_I^{(4)}=\rho\iota_b$.  
Evaluating the function $g$ at the expected vector $\bar\mu_I$ yields:
\[
  g(\bar\mu_I) = \frac{\bar\mu_I^{(4)}}{\bar\mu_I^{(3)}} - \frac{\bar\mu_I^{(2)}}{\bar\mu_I^{(1)}} = \rho - \kappa.
\]
Under $H_{0,\mathrm{homog}}:\kappa=\rho$, this gives $g(\bar\mu_I) = 0$ exactly.

\paragraph{Step 2: Joint CLT via Cram\'er--Wold.}
By Assumption~1, the IV assignment vectors are independent across strata, so $W_1,\ldots,W_I$ are independent conditional on $\mathcal{F}$. We establish
\begin{equation}\label{eq:jointCLT}
  \sqrt{I}(\bar{W}_I-\bar\mu_I)
  \xrightarrow{d} N(\mathbf{0},\Sigma),
\end{equation}
where 
$\Sigma=\lim_{I\to\infty}I^{-1}\sum_{i=1}^I\mathrm{Var}(W_i|\mathcal{F})$
(which exists by Condition~1 and boundedness).

By the \emph{Cram\'er--Wold theorem}, \eqref{eq:jointCLT} holds if and 
only if for every fixed $c\in\mathbb{R}^4$, the scalar sequence
\[
  \sqrt{I}\,c^\top(\bar{W}_I-\bar\mu_I)
  = \sqrt{I}\cdot I^{-1}\sum_{i=1}^I[c^\top W_i - c^\top\mu_i]
\]
is asymptotically $N(0, c^\top\Sigma c)$. The case $c=\mathbf{0}$ is trivial, since both sides vanish identically.
Fix such a $c\neq\mathbf{0}$.
 The scalars $c^\top W_i$ and $c^\top W_{i'}$ are independent for $i\neq i'$ by Assumption~1.
Each component of $W_i$ is a finite linear combination of bounded potential outcomes, so there exists $B_c<\infty$, depending only on $c$ and the bound on the potential outcomes, such that $|c^\top W_i | \leq B_c < \infty$ for all $i$. Consequently,
\[
I^{-2}\sum_{i=1}^I E[(c^\top W_i )^4 \mid \mathcal{F}] \;\leq\; B_c^4 / I \;\to\; 0,
\]
verifying condition (ii) of Lemma~2. Condition~(i) of Lemma~2 (nondegeneracy), 
$I^{-1}\sum_{i=1}^I \mathrm{Var}(c^\top W_i\mid\mathcal{F}) \to c^\top\Sigma\,c > 0$, follows from Condition~2. 
Lemma~2 then yields the univariate CLT for $c^\top\bar{W}_I$, and since $c$ was arbitrary, the Cram\'er--Wold theorem gives \eqref{eq:jointCLT}.
\medskip

\paragraph{Step 3: Delta method.}

The function $g(x_1,x_2,x_3,x_4) = x_4/x_3-x_2/x_1$ is continuously 
differentiable on any open set where $x_1\neq 0$ and $x_3\neq 0$.
Throughout we assume the implicit regularity condition that the compliance rates $\iota_a$ and $\iota_b$ are bounded away from zero for all sufficiently large $I$; this is the same condition that ensures the IV effect ratios $\kappa$ and $\rho$ are well-defined. Under this condition, $\bar\mu_I^{(1)} = \iota_a$ and $\bar\mu_I^{(3)} = \iota_b$ lie in the differentiability region of $g$. 

The gradient at $\bar\mu_I$ is:
\[
  \nabla g(\bar\mu_I)
  = \Bigl(\frac{\bar\mu_I^{(2)}}{(\bar\mu_I^{(1)})^2},\;
          -\frac{1}{\bar\mu_I^{(1)}},\;
          -\frac{\bar\mu_I^{(4)}}{(\bar\mu_I^{(3)})^2},\;
          \frac{1}{\bar\mu_I^{(3)}}\Bigr)^\top
 =
  \Bigl(\frac{\kappa}{\iota_a},\;
        -\frac{1}{\iota_a},\;
        -\frac{\rho}{\iota_b},\;
        \frac{1}{\iota_b}\Bigr)^\top.
\]
Applying the delta method to \eqref{eq:jointCLT} at the 
sequence of centering points $\bar\mu_I$ gives
\begin{equation}\label{eq:deltamethod}
  \sqrt{I}\bigl(g(\bar{W}_I)-g(\bar\mu_I)\bigr)
  \;\xrightarrow{d}\;
  N\!\bigl(0,\,\sigma_H^2\bigr),
\end{equation}
where, with $\boldsymbol\mu:=\lim_{I\to\infty}\bar\mu_I$,
\begin{equation*}\label{eq:sigmaH2}
  \sigma_H^2
  := [\nabla g(\boldsymbol\mu)]^\top\Sigma\,[\nabla g(\boldsymbol\mu)]
  = \lim_{I\to\infty}\frac{1}{I}\sum_{i=1}^I
    \mathrm{Var}\bigl([\nabla g(\bar\mu_I)]^\top W_i\bigm|\mathcal{F}\bigr)
\end{equation*}
exists and is finite under Condition~1.
Under $H_{0,\mathrm{homog}}$,  $g(\bar\mu_I)=0$, and 
\eqref{eq:deltamethod} reduces to
\[
  \sqrt{I}\,g(\bar{W}_I) \;\xrightarrow{d}\; N(0,\sigma_H^2).
\]

\paragraph{Step 4: Conservative variance estimation.}

We establish that $\hat{\sigma}_H^2 := [\nabla g(\bar{W}_I)]^\top S_W^2[\nabla g(\bar{W}_I)]$ is a conservative estimate of $\sigma_H^2$. 

\noindent\textit{Step 4a: Matrix version of Lemma~1 
(overestimation in expectation).}
Define $\bar{\Sigma}_I = \frac{1}{I} \sum_{i=1}^I \mathrm{Var}(W_i \mid \mathcal{F})$. We establish that $E[S_W^2 \mid \mathcal{F}] \succeq \bar{\Sigma}_I$ in the positive semi-definite order. Expanding the sample covariance:
\begin{align*}
(I-1)\, S_W^2 
&= \sum_{i=1}^I (W_i - \bar{W}_I)(W_i - \bar{W}_I)^\top \\
&= \sum_{i=1}^I \bigl[(W_i - \mu_i) + (\mu_i - \bar{\mu}_I) - (\bar{W}_I - \bar{\mu}_I)\bigr]\bigl[(W_i - \mu_i) + (\mu_i - \bar{\mu}_I) - (\bar{W}_I - \bar{\mu}_I)\bigr]^\top.
\end{align*}
Taking expectations conditional on $\mathcal{F}$, cross-terms pairing $(\mu_i - \bar\mu_I)$ with either $(W_i - \mu_i)$ or 
$(\bar W_I - \bar\mu_I)$ vanish, since $\mu_i - \bar\mu_I$ is non-random 
given $\mathcal{F}$ while both $E[W_i - \mu_i \mid \mathcal{F}] = 0$ and 
$E[\bar W_I - \bar\mu_I \mid \mathcal{F}] = 0$. 
The surviving terms are
\begin{align*}
E\!\left[(I-1)\, S_W^2 \;\middle|\; \mathcal{F}\right] 
&= \sum_{i=1}^I \mathrm{Var}(W_i \mid \mathcal{F}) \;+\; \sum_{i=1}^I (\mu_i - \bar{\mu}_I)(\mu_i - \bar{\mu}_I)^\top \;-\; I\, \mathrm{Var}(\bar{W}_I \mid \mathcal{F}) \nonumber\\
&= I\, \bar{\Sigma}_I \;+\; \sum_{i=1}^I (\mu_i - \bar{\mu}_I)(\mu_i - \bar{\mu}_I)^\top \;-\; \bar{\Sigma}_I, \label{eq:expand_SW}
\end{align*}
 Dividing by $(I-1)$:
\begin{equation}\label{eq:SW_bias}
E[S_W^2 \mid \mathcal{F}] = \bar{\Sigma}_I + \frac{1}{I-1} \sum_{i=1}^I (\mu_i - \bar{\mu}_I)(\mu_i - \bar{\mu}_I)^\top.
\end{equation}
Since the second term is positive semi-definite, 
$E[S_W^2 \mid \mathcal{F}] \succeq \bar{\Sigma}_I$, meaning
for any vector $v \in \mathbb{R}^4$: 
\[
v^\top E[S_W^2 \mid \mathcal{F}]\, v \geq v^\top \bar{\Sigma}_I\, v.
\]
In particular, taking $v=\nabla g(\boldsymbol\mu)$:
\begin{equation}\label{eq:expconserv}
  [\nabla g(\boldsymbol\mu)]^\top E[S_W^2 \mid\mathcal{F}]\,[\nabla g(\boldsymbol\mu)]
  \;\geq\;
  [\nabla g(\boldsymbol\mu)]^\top\bar\Sigma_I\,[\nabla g(\boldsymbol\mu)]
  \;\to\; \sigma_H^2.
\end{equation}

\noindent\textit{Step 4b: Multivariate Lemma~3 applied to $S_W^2$ via 
linear projections.}
We show that $\,S_W^2 - \,E[S_W^2 \mid \mathcal{F}] \xrightarrow{p} 0$ 
entry-wise. Lemma~3 is a scalar statement and does not apply directly to 
the matrix $S_W^2$; instead, we apply it to the scalar projection 
$a^\top S_W^2\,a$ for an arbitrary fixed $a \in \mathbb{R}^4$, and then 
recover entry-wise convergence.

Fix $a \in \mathbb{R}^4$ and set $D_i^a := a^\top W_i$. Then 
$D_1^a, \ldots, D_I^a$ are independent conditional on $\mathcal{F}$ by 
Assumption~1, and
\[
  a^\top S_W^2\,a 
  \;=\; \frac{1}{I-1}\sum_{i=1}^I \bigl(a^\top (W_i - \bar W_I)\bigr)^2
  \;=\; \frac{1}{I-1}\sum_{i=1}^I \bigl(D_i^a - \bar D_I^a\bigr)^2
  \;=\; I\,S_a^2,
\]
where $S_a^2 := \frac{1}{I(I-1)}\sum_{i=1}^I (D_i^a - \bar D_I^a)^2$ is 
the scalar variance estimator to which Lemma~3 applies. We verify the three conditions of Lemma~3.
Each component of $W_i$ is a finite linear combination of bounded 
potential outcomes, so there exists $B_a < \infty$, depending only on $a$ 
and the bound on the potential outcomes, such that $|D_i^a| \leq B_a$ for 
all $i$.

\emph{Condition (i):} $E[\bar D_I^a \mid \mathcal{F}] = a^\top \bar\mu_I$, 
which is bounded uniformly in $I$ since the entries of $\bar\mu_I$ are 
uniformly bounded under Condition~1. Hence $E[\bar D_I^a \mid \mathcal{F}] = O(1)$.

\emph{Condition (ii):} 
$I^{-2}\sum_{i=1}^I E[(D_i^a)^2 \mid \mathcal{F}] \leq I^{-2} \cdot I \cdot B_a^2 = B_a^2/I \to 0$.

\emph{Condition (iii):} 
$I^{-2}\sum_{i=1}^I E[(D_i^a)^4 \mid \mathcal{F}] \leq B_a^4/I \to 0$.

Lemma~3 then yields $I\,S_a^2 - I\,E[S_a^2 \mid \mathcal{F}] 
\xrightarrow{p} 0$, which since $a^\top S_W^2 a = I\,S_a^2$ is exactly
\begin{equation}\label{eq:projection_concentration}
  a^\top S_W^2\,a - a^\top E[S_W^2 \mid \mathcal{F}]\,a \;\xrightarrow{p}\; 0.
\end{equation}
For each pair $(\ell,\ell')\in\{1,2,3,4\}^2$, the $({\ell},\ell')$ entry of 
$S_W^2$ is
$S_W^{2,(\ell\ell')} = (I-1)^{-1}\sum_i(W_i^{(\ell)}-\bar{W}^{(\ell)})(W_i^{(\ell')}-\bar{W}^{(\ell')})$.
Since~\eqref{eq:projection_concentration} holds for every fixed 
$a \in \mathbb{R}^4$, applying it to $a = e_\ell$ (the $\ell$-th 
standard basis vector) gives the entry-wise convergence
$S_W^{2,(\ell\ell)} - E[S_W^{2,(\ell\ell)} \mid \mathcal{F}] \xrightarrow{p} 0$. For the off-diagonal entries of  $M_I:=S_W^2-E[S_W^2 \mid \mathcal{F}]$, we note that
\[
  M_I^{(\ell\ell')}
  \;=\; \tfrac{1}{2}\bigl[(e_\ell + e_{\ell'})^\top M_I (e_\ell + e_{\ell'}) 
                          - e_\ell^\top M_I e_\ell - e_{\ell'}^\top M_I e_{\ell'}\bigr]
\]
can be expressed as a fixed linear combination of three scalar projections of $M_I$, each 
of which converges to zero in probability by the scalar result applied 
to $a = e_\ell + e_{\ell'}$, $a = e_\ell$, and $a = e_{\ell'}$ in turn. 
By continuous mapping theorem, 
$M_I^{(\ell\ell')} \xrightarrow{p} 0$ for all $\ell \neq \ell'$.
We conclude that
\begin{equation}\label{eq:SW_concentration}
  \,S_W^2 - \,E[S_W^2 \mid \mathcal{F}] \;\xrightarrow{p}\; 0
  \quad \text{entry-wise.}
\end{equation}

\noindent\textit{Step 4c: Consistent (conservative) estimation of 
$\sigma_H^2$ via the continuous mapping theorem.}
Combining Steps~4a and~4b, define 
$\Lambda_I := (I-1)^{-1}\sum_{i=1}^I (\mu_i - \bar\mu_I)(\mu_i - \bar\mu_I)^\top$, 
so that~\eqref{eq:SW_bias} reads $E[S_W^2 \mid \mathcal{F}] = \bar\Sigma_I + \Lambda_I$. 
Together with~\eqref{eq:SW_concentration}:
\begin{equation}\label{eq:SW_decomp}
  S_W^2 \;=\; E[S_W^2 \mid \mathcal{F}] + o_p(1)
        \;=\; \bar\Sigma_I + \Lambda_I + o_p(1)
        \;\xrightarrow{p}\; \Sigma + \Lambda,
\end{equation}
where $\Lambda := \lim_{I\to\infty} \Lambda_I$, which exists 
and is positive semi-definite under Condition~1.

 Since $W_1, \ldots, W_I$ are independent with uniformly bounded entries, we have
\[
  \mathrm{Var}(\bar W_I^{(\ell)}\mid\mathcal{F}) 
  \;=\; I^{-2}\sum_{i=1}^I \mathrm{Var}(W_i^{(\ell)}\mid\mathcal{F}) 
  \;=\; O(I^{-1}),
\]
for each component $\ell$. By Chebyshev's inequality, 
$\bar W_I - \bar\mu_I \xrightarrow{p} 0$ entry-wise, and hence 
$\bar W_I \xrightarrow{p} \boldsymbol\mu$.  The continuous mapping theorem gives
\begin{equation}\label{eq:sigmahat_limit}
 \hat\sigma_H^2 := [\nabla g(\bar W_I)]^\top S_W^2\,[\nabla g(\bar W_I)]
  \;\xrightarrow{p}\;
  [\nabla g(\boldsymbol\mu)]^\top(\Sigma + \Lambda)\,[\nabla g(\boldsymbol\mu)]
  \;=\;  \sigma_H^2 + \delta,
\end{equation}
where 
$\delta := [\nabla g(\boldsymbol\mu)]^\top \Lambda\, [\nabla g(\boldsymbol\mu)] \geq 0$ because $\Lambda$ is positive semi-definite; 
$\delta = 0$ if $\mu_i$ is constant across strata $i=1,\dots,I$.


\medskip
\noindent\textit{Step 4d: Slutsky's theorem and asymptotic level.}
Under $H_{0,\mathrm{homog}}$, Step~3 gives 
$\sqrt{I}\,g(\bar W_I) \xrightarrow{d} N(0, \sigma_H^2)$, and Step~4c 
gives $\hat\sigma_H^2 \xrightarrow{p} \sigma_H^2 + \delta$ with 
$\delta \geq 0$.  Slutsky's theorem then yields
\[
  T_{\mathrm{homog}}
  \;=\; \frac{\sqrt{I}\,g(\bar W_I)}{\sqrt{\hat\sigma_H^2}}
  \;=\; \frac{\sqrt{I}\,g(\bar W_I)}{\sqrt{\sigma_H^2}}
        \;\cdot\;
        \frac{\sqrt{\sigma_H^2}}{\sqrt{\hat\sigma_H^2}}
  \;\xrightarrow{d}\;
  N\!\Bigl(0,\;\frac{\sigma_H^2}{\sigma_H^2+\delta}\Bigr).
\]
The asymptotic variance $\sigma_H^2/(\sigma_H^2+\delta)$ is at most $1$, 
with equality if and only if $\delta = 0$. Therefore, for 
$Z \sim N(0,1)$:
\[
  \limsup_{I\to\infty}\,
  \Pr\!\left\{|T_{\mathrm{homog}}| \geq z_{1-\alpha/2}
              \;\Big|\; \mathcal{F},\,\mathcal{Z}\right\}
  \;\leq\;
  \Pr\!\left\{|Z| \geq z_{1-\alpha/2}\right\}
  \;=\; \alpha,
\]
establishing that the test is asymptotically valid at level~$\alpha$. The inequality is strict whenever $\delta > 0$, in which case the test is conservative; equality holds when $\mu_i$ is constant across strata.
\end{proof}
}

\subsection*{A.6 Proof of Theorem 5}

\noindent To streamline exposition, we first present the proof of Theorem 5, as it contains the key technical ingredients used in biased randomization inference. Theorem 4 is proved by the same argument applied to a simpler test statistic; hence we prove Theorem 5 below and omit the nearly identical details for Theorem 4.

\begin{proof}







The proof proceeds in three main steps: (1) establish stochastic dominance conditional on each treatment assignment realization, (2) construct moment bounds to replace the unknown expectation and variance under the null, and (3) verify asymptotic normality via a Lyapunov's condition.


\paragraph{Step 1: Stochastic dominance construction.}
There are four possible realizations of $(T_{i1}, T_{i2})$, each occurring with probability $1/4$ under randomized treatment assignment within each pair. We consider each realization separately.
Let $\hat\tau^{(t_1, t_2)}_{i,\text{SW}}$ denote the value of $\hat\tau_{i, \mathrm{SW}}(\lambda_0) $ given the realization $(T_{i1}, T_{i2})=(t_1, t_2)$, where $(t_1, t_2)\in \{(+1,+1), (+1,-1), (-1,+1), (-1,-1)\}$. For each realization, we decompose:
$$\hat\tau^{(t_1, t_2)}_{i,\text{SW}} = \bar\tau_{i,\text{SW}}^{(t_1, t_2)} + (2A_i - 1)\eta_{i,\text{SW}}^{(t_1, t_2)}$$
where $\eta_{i,\text{SW}}^{(t_1, t_2)}$ are defined in Section A.1, and the quantities $\hat\tau^{(t_1, t_2)}_{i,\text{SW}}$ and $\bar\tau_{i,\text{SW}}^{(t_1, t_2)}$ are specified for each case as follows:

---

\textit{Case 1:} $(t_1, t_2) = (1,-1)$
\begin{align*}
\hat\tau_{i,\text{SW}}^{(1,-1)}
&= A_i\{(y_{T^bi22} - y_{C^bi21}) - (y_{T^ai11}-y_{C^ai12})\}  \\
&\quad+ (1-A_i)\{(y_{T^bi11} - y_{C^bi12}) - (y_{T^ai22} - y_{C^ai21})\}, \\[6pt]
\bar\tau_{i,\text{SW}}^{(1,-1)} 
&= \tfrac{1}{2}\!\left[(y_{T^bi11} - y_{C^bi12}) - (y_{T^ai11} - y_{C^ai12})
+ (y_{T^bi22} - y_{C^bi21}) - (y_{T^ai22} - y_{C^ai21})\right]. 
\end{align*}

---

\textit{Case 2:} $(t_1, t_2) = (-1,-1)$
\begin{align*}
\hat\tau_{i,\text{SW}}^{(-1,-1)}
&= A_i\{(y_{T^bi22} - y_{C^bi21}) - (y_{T^ai12}-y_{C^ai11})\} \\
&\quad+ (1-A_i)\{(y_{T^bi12} - y_{C^bi11}) - (y_{T^ai22} - y_{C^ai21})\}, \\[6pt]
\bar\tau_{i,\text{SW}}^{(-1,-1)} 
&= \tfrac{1}{2}\!\left[(y_{T^bi12} - y_{C^bi11}) - (y_{T^ai12} - y_{C^ai11})
+ (y_{T^bi22} - y_{C^bi21}) - (y_{T^ai22} - y_{C^ai21})\right]. 
\end{align*}

---

\textit{Case 3:} $(t_1, t_2) = (-1,1)$
\begin{align*}
\hat\tau_{i,\text{SW}}^{(-1,1)}
&= A_i\{(y_{T^bi21} - y_{C^bi22}) - (y_{T^ai12}-y_{C^ai11})\} \\
&\quad+ (1-A_i)\{(y_{T^bi12} - y_{C^bi11}) - (y_{T^ai21} - y_{C^ai22})\}, \\[6pt]
\bar\tau_{i,\text{SW}}^{(-1,1)} 
&= \tfrac{1}{2}\!\left[(y_{T^bi12} - y_{C^bi11}) - (y_{T^ai12} - y_{C^ai11})
+ (y_{T^bi21} - y_{C^bi22}) - (y_{T^ai21} - y_{C^ai22})\right]. 
\end{align*}

---

\textit{Case 4:} $(t_1, t_2) = (1,1)$
\begin{align*}
\hat\tau_{i,\text{SW}}^{(1,1)}
&= A_i\{(y_{T^bi21} - y_{C^bi22}) - (y_{T^ai11} - y_{C^ai12})\} \\
&\quad+ (1-A_i)\{(y_{T^bi11} - y_{C^bi12}) - (y_{T^ai21} - y_{C^ai22})\}, \\[6pt]
\bar\tau_{i,\text{SW}}^{(1,1)} 
&= \tfrac{1}{2}\!\left[(y_{T^bi11} - y_{C^bi12}) - (y_{T^ai11} - y_{C^ai12})
+ (y_{T^bi21} - y_{C^bi22}) - (y_{T^ai21} - y_{C^ai22})\right]. 
\end{align*}

\noindent Under the partly biased randomization scheme, assignment probabilities satisfy:
\begin{align}
\frac{1}{1+\Gamma} \leq \pi_i = \Pr(A_i = 1) \leq \frac{\Gamma}{1+\Gamma}, \quad \frac{1}{1+\Gamma} \leq 1-\pi_i = \Pr(A_i = 0) \leq \frac{\Gamma}{1+\Gamma}
\end{align}
This implies:
$$ \frac{1}{1+\Gamma} \leq \Pr\{(2A_i- 1)\eta_{i,\text{SW}}^{(t_1, t_2)}=|\eta_{i,\text{SW}}^{(t_1, t_2)}|\}\leq \frac{\Gamma}{1+\Gamma}$$
 for $\eta_{i,\text{SW}}^{(t_1, t_2)}\neq 0$. If $\eta_{i,\text{SW}}^{(t_1, t_2)}= 0$ then $(2A_i- 1)\eta_{i,\text{SW}}^{(t_1, t_2)}=|\eta_{i,\text{SW}}^{(t_1, t_2)}|=0$, and all inequalities below hold trivially.

\noindent Define a binary random variable $V_{i,\Gamma}$ taking values $+1$ or $-1$ with:
\begin{align}
\Pr(V_{i,\Gamma} = +1) = \frac{\Gamma}{1+\Gamma}, \quad \Pr(V_{i,\Gamma} = -1) = \frac{1}{1+\Gamma}.
\end{align}
For the fixed $\eta_{i,\text{SW}}^{(t_1, t_2)}$, the variable $V_{i,\Gamma}|\eta_{i,\text{SW}}^{(t_1, t_2)}|$ takes values $\pm|\eta_{i,\text{SW}}^{(t_1, t_2)}|$ with mass on $+|\eta_{i,\text{SW}}^{(t_1, t_2)}|$ equal to $\Gamma/(1+\Gamma)$, 
which is at least the mass that $(2A_i - 1)\eta_{i,\text{SW}}^{(t_1, t_2)}$ places on $+|\eta_{i,\text{SW}}^{(t_1, t_2)}|$. 
Therefore, for each fixed $(t_1, t_2)$, 
$V_{i,\Gamma}|\eta_{i,\text{SW}}^{(t_1, t_2)}|$ stochastically dominates $(2A_i - 1)\eta_{i,\text{SW}}^{(t_1, t_2)}$.

Consider the function $f(x) = x - \frac{\Gamma-1}{\Gamma+1} |x|$. For $x \geq 0$, $f(x)=\frac{2}{\Gamma+1} x$ and for $x < 0$, $f(x)=\frac{2\Gamma}{\Gamma+1} x$. Since $\Gamma>0$ and  $f$ is continuous at $x=0$ with $f(0)=0$, $f$ is strictly increasing on $\mathbbm{R}$.

\noindent Define
$$U_{i,\Gamma}^{(t_1, t_2)} \equiv \bar\tau_{i,\text{SW}}^{(t_1, t_2)} + V_{i,\Gamma}|\eta_{i,\text{SW}}^{(t_1, t_2)}| - \left(\frac{\Gamma-1}{\Gamma+1}\right)|\bar\tau_{i,\text{SW}}^{(t_1, t_2)} + V_{i,\Gamma}|\eta_{i,\text{SW}}^{(t_1, t_2)}||$$
and
$$D_{i,\Gamma}^{(t_1, t_2)} \equiv  \hat\tau^{(t_1, t_2)}_{i,\text{SW}} - \left(\frac{\Gamma-1}{\Gamma+1}\right)|\hat\tau^{(t_1, t_2)}_{i,\text{SW}}|$$
Since stochastic dominance is preserved under the monotone increasing $f$, we have:
$$\Pr(U_{i,\Gamma}^{(t_1, t_2)} \geq c) \geq\Pr(D_{i,\Gamma}^{(t_1, t_2)} \geq c) = \Pr(D^{\mathrm{SW}}_{i, \Gamma}(\lambda_0)\geq c \mid (T_{i1}, T_{i2}) = (t_1, t_2)),$$
for any scalar $c$ and each realization $(T_{i1}, T_{i2}) = (t_1, t_2)$, where the last equality uses the definition $D_{i,\Gamma}^{(t_1, t_2)} = D_{i,\Gamma}$ on the event $(T_{i1}, T_{i2}) = (t_1, t_2)$.

 By the law of total probability:
\begin{align}\label{law of total prob: D}
\Pr(D^{\mathrm{SW}}_{i, \Gamma}(\lambda_0)\geq c) &= \sum_{(t_1,t_2)} \Pr(D^{\mathrm{SW}}_{i, \Gamma}(\lambda_0)\geq c \mid (T_{i1}, T_{i2}) = (t_1, t_2)) \Pr((T_{i1}, T_{i2}) = (t_1, t_2))\nonumber\\
&= \frac{1}{4} \sum_{(t_1,t_2)} \Pr(D^{\mathrm{SW}}_{i, \Gamma}(\lambda_0)\geq c \mid (T_{i1}, T_{i2}) = (t_1, t_2))\nonumber\\
&\leq \frac{1}{4} \sum_{(t_1,t_2)} \Pr(U_{i,\Gamma}^{(t_1, t_2)} \geq c),
\end{align}

Define $U_{i,\Gamma}$ as a random variable with:
\begin{align*}
\Pr(U_{i,\Gamma} = U_{i,\Gamma}^{(t_1, t_2)}) = \frac{1}{4} \quad \text{for } (t_1, t_2) \in \{(+1,+1), (+1,-1), (-1,+1), (-1,-1)\}
\end{align*}

 Using the law of total probability for $U_{i,\Gamma}$:
\begin{align}\label{law of total prob: U}
\Pr(U_{i,\Gamma} \geq c) &= \sum_{(t_1,t_2)} \Pr(U_{i,\Gamma} \geq c \mid U_{i,\Gamma} = U_{i,\Gamma}^{(t_1, t_2)}) \Pr(U_{i,\Gamma} = U_{i,\Gamma}^{(t_1, t_2)})\nonumber\\
&= \frac{1}{4} \sum_{(t_1,t_2)} \Pr(U_{i,\Gamma}^{(t_1, t_2)} \geq c)
\end{align}


Combining the results from \eqref{law of total prob: D} and \eqref{law of total prob: U}, we prove that $U_{i,\Gamma}$ stochastically dominates $D^{\mathrm{SW}}_{i, \Gamma}(\lambda_0)$ for each $i$, i.e., for any scalar $c$:
$$\Pr(U_{i,\Gamma} \geq c) \geq \Pr(D^{\mathrm{SW}}_{i, \Gamma}(\lambda_0)\geq c)$$

\paragraph{Step 2: Moment bounds.}
By Lemma 1, $S^2_{\Gamma}$ provides an estimator of $\text{var}(\overline{D}_{\Gamma})$ which is conservative in expectation such that
$\mathrm{E}(S^2_{\Gamma}) \geq \text{var}(\overline{D}_{\Gamma})$ for any $\Gamma \geq 1$. 
We then bound the expectation of $\overline{D}_{\Gamma}$ by bounding the expectation of its stochastic upper bound $\overline{U}_{\Gamma}$.

\noindent For each $i$,
\begin{align*}
    \mathrm{E}(U_{i,\Gamma}^{(t_1, t_2)})&= \bar\tau_{i,\text{SW}}^{(t_1, t_2)}+\frac{\Gamma-1}{1+\Gamma}|\eta_{i,\text{SW}}^{(t_1, t_2)}|\\&- \left(\frac{\Gamma-1}{\Gamma+1}\right) 
    \{ \frac{\Gamma}{\Gamma+1}|\bar\tau_{i,\text{SW}}^{(t_1, t_2)} + |\eta_{i,\text{SW}}^{(t_1, t_2)}||+ \frac{1}{1+\Gamma}|\bar\tau_{i,\text{SW}}^{(t_1, t_2)} -|\eta_{i,\text{SW}}^{(t_1, t_2)}||\}
\end{align*}
\noindent Following Lemma 2 in \cite{fogarty2020studentized}, with $\Gamma \geq 1$, the inequality $$\frac{\Gamma}{\Gamma+1}|\bar\tau_{i,\text{SW}}^{(t_1, t_2)} + |\eta_{i,\text{SW}}^{(t_1, t_2)}||+ \frac{1}{1+\Gamma}|\bar\tau_{i,\text{SW}}^{(t_1, t_2)} -|\eta_{i,\text{SW}}^{(t_1, t_2)}|| \geq |\eta_{i,\text{SW}}^{(t_1, t_2)}| + \frac{\Gamma-1}{\Gamma+1}\bar\tau_{i,\text{SW}}^{(t_1, t_2)}$$ always holds. Hence, 
\begin{align*}
    \mathrm{E}(U_{i,\Gamma}^{(t_1, t_2)})
    &\leq \bar\tau_{i,\text{SW}}^{(t_1, t_2)}+\frac{\Gamma-1}{1+\Gamma}|\eta_{i,\text{SW}}^{(t_1, t_2)}|- \left(\frac{\Gamma-1}{\Gamma+1}\right)\{|\eta_{i,\text{SW}}^{(t_1, t_2)}| + \frac{\Gamma-1}{\Gamma+1}\bar\tau_{i,\text{SW}}^{(t_1, t_2)}\}\\
    &=\left[1- \left(\frac{\Gamma-1}{\Gamma+1}\right)^2\right]\bar\tau_{i,\text{SW}}^{(t_1, t_2)}.
\end{align*}
Averaging over the four treatment assignment cases of $(T_{i1},T_{i2})$,
\begin{align*}
\mathrm{E}[U_{i,\Gamma}] &= \frac{1}{4}\sum_{t_{1} \in \{1,-1\}} \sum_{t_{2} \in \{1,-1\}} \mathrm{E}[U_{i,\Gamma}^{(t_1, t_2)}]\\
&\leq \frac{1}{4}\left[1- \left(\frac{\Gamma-1}{\Gamma+1}\right)^2\right]\sum_{t_{1} \in \{1,-1\}} \sum_{t_{2} \in \{1,-1\}} \bar\tau_{i,\text{SW}}^{(t_1, t_2)}
\end{align*}

\noindent Summing over all strata:
\begin{align*}
    \mathrm{E}\left[\sum_{i=1}^I U_{i,\Gamma}\right] &\leq \frac{1}{4}\left[1- \left(\frac{\Gamma-1}{\Gamma+1}\right)^2\right] \sum_{i=1}^I \sum_{t_{1} \in \{1,-1\}} \sum_{t_{2} \in \{1,-1\}} \bar\tau_{i,\text{SW}}^{(t_1, t_2)} 
    \\&= \frac{1}{4}\left[1- \left(\frac{\Gamma-1}{\Gamma+1}\right)^2\right]\sum_{i = 1}^I\sum_{j = 1}^2 \sum_{k = 1}^2 (\tau^b_{ijk} - \tau^a_{ijk}),
\end{align*}
where the second equality follows from the definition of $\bar\tau_{i,\text{SW}}^{(t_1, t_2)}$ with $\tau^g_{ijk} = y_{T^gijk}-y_{C^gijk}$.  

\noindent Under $H^{SW}_0: \lambda = \lambda_0$, we have $$\sum_{i = 1}^I \sum_{j = 1}^2 \sum_{k = 1}^2 (\tau^b_{ijk} - \tau^a_{ijk})=0,$$ which implies
$\mathrm{E}\left[\sum_{i=1}^I U_{i,\Gamma}\right] \leq 0$.
Since $D^{\mathrm{SW}}_{i, \Gamma}(\lambda_0)$ is stochastically dominated by $U_{i,\Gamma}$ for each $i$,  it follows that $\mathrm{E}[\overline{D}_{\Gamma, \mathrm{SW}}(\lambda_0)] \leq 0$ under $H^{SW}_0$.


\paragraph{Step 3: Establish the asymptotic normality.}

Given the independence across strata, to prove asymptotic normality of $\sqrt{I}\overline{D}_{\Gamma, \mathrm{SW}}(\lambda_0)$, it suffices to show that Lyapunov’s condition holds for $\delta=2$.

\medskip
\noindent For the conditional variance, we compute
\begin{align*}
   &\operatorname{Var}( D^{\mathrm{SW}}_{i,\Gamma}(\lambda_0)\mid T_{i1} = t_1, T_{i2} =t_2)  
  \\=&  \operatorname{Var}(\hat\tau_{i,\text{SW}}^{(t_1, t_2)} - \left(\frac{\Gamma-1}{\Gamma+1}\right)\lvert\hat\tau_{i,\text{SW}}^{(t_1, t_2)}\lvert\mid T_{i1} = t_1, T_{i2} =t_2) \\
   =&  \operatorname{Var}(\bar\tau_{i,\text{SW}}^{(t_1, t_2)} + (2A_i - 1)\eta_{i,\text{SW}}^{(t_1, t_2)} - \left(\frac{\Gamma-1}{\Gamma+1}\right)|\bar\tau_{i,\text{SW}}^{(t_1, t_2)} + (2A_i - 1)\eta_{i,\text{SW}}^{(t_1, t_2)}| \mid T_{i1} = t_1, T_{i2} =t_2) 
\end{align*}
Since $(2A_i-1)$ takes values $\pm1$ with probability $\pi_i$ and $1-\pi_i$, respectively, we have
\begin{align*}
&\operatorname{E}\biggr[\bar\tau_{i,\text{SW}}^{(t_1, t_2)} + (2A_i - 1)\eta_{i,\text{SW}}^{(t_1, t_2)} - \frac{\Gamma-1}{\Gamma+1}|\bar\tau_{i,\text{SW}}^{(t_1, t_2)} + (2A_i - 1)\eta_{i,\text{SW}}^{(t_1, t_2)}|\biggr] \\
=& \pi_i\left(\bar\tau_{i,\text{SW}}^{(t_1, t_2)} + \eta_{i,\text{SW}}^{(t_1, t_2)} - \frac{\Gamma-1}{\Gamma+1}|\bar\tau_{i,\text{SW}}^{(t_1, t_2)} + \eta_{i,\text{SW}}^{(t_1, t_2)}|\right) 
+ (1-\pi_i)\left(\bar\tau_{i,\text{SW}}^{(t_1, t_2)} - \eta_{i,\text{SW}}^{(t_1, t_2)} - \frac{\Gamma-1}{\Gamma+1}|\bar\tau_{i,\text{SW}}^{(t_1, t_2)} - \eta_{i,\text{SW}}^{(t_1, t_2)}|\right) 
\end{align*}
and
\begin{align*}
&\operatorname{E}\biggr[\left(\bar\tau_{i,\text{SW}}^{(t_1, t_2)} + (2A_i - 1)\eta_{i,\text{SW}}^{(t_1, t_2)} - \frac{\Gamma-1}{\Gamma+1}|\bar\tau_{i,\text{SW}}^{(t_1, t_2)} + (2A_i - 1)\eta_{i,\text{SW}}^{(t_1, t_2)}|\right)^2\biggr] \\
=& \pi_i\left(\bar\tau_{i,\text{SW}}^{(t_1, t_2)} + \eta_{i,\text{SW}}^{(t_1, t_2)} - \frac{\Gamma-1}{\Gamma+1}|\bar\tau_{i,\text{SW}}^{(t_1, t_2)} + \eta_{i,\text{SW}}^{(t_1, t_2)}|\right)^2 
+ (1-\pi_i)\left(\bar\tau_{i,\text{SW}}^{(t_1, t_2)} - \eta_{i,\text{SW}}^{(t_1, t_2)} - \frac{\Gamma-1}{\Gamma+1}|\bar\tau_{i,\text{SW}}^{(t_1, t_2)} - \eta_{i,\text{SW}}^{(t_1, t_2)}|\right)^2
\end{align*}
Then by the definition of variance, we obtain
\begin{align*}
   &\operatorname{Var}( D^{\mathrm{SW}}_{i,\Gamma}(\lambda_0)\mid T_{i1} = t_1, T_{i2} =t_2)  
  \\=&  \operatorname{Var}(\bar\tau_{i,\text{SW}}^{(t_1, t_2)} + (2A_i - 1)\eta_{i,\text{SW}}^{(t_1, t_2)} - \left(\frac{\Gamma-1}{\Gamma+1}\right)|\bar\tau_{i,\text{SW}}^{(t_1, t_2)} + (2A_i - 1)\eta_{i,\text{SW}}^{(t_1, t_2)}| ) \\
  =& \operatorname{E}\biggr[\left(\bar\tau_{i,\text{SW}}^{(t_1, t_2)} + (2A_i - 1)\eta_{i,\text{SW}}^{(t_1, t_2)} - \frac{\Gamma-1}{\Gamma+1}|\bar\tau_{i,\text{SW}}^{(t_1, t_2)} + (2A_i - 1)\eta_{i,\text{SW}}^{(t_1, t_2)}|\right)^2\biggr] \\
  -&\left(\operatorname{E}\biggr[\bar\tau_{i,\text{SW}}^{(t_1, t_2)} + (2A_i - 1)\eta_{i,\text{SW}}^{(t_1, t_2)} - \frac{\Gamma-1}{\Gamma+1}|\bar\tau_{i,\text{SW}}^{(t_1, t_2)} + (2A_i - 1)\eta_{i,\text{SW}}^{(t_1, t_2)}|\biggr]\right)^2 \\
   =&  \pi_i(1-\pi_i)\left(\bar\tau_{i,\text{SW}}^{(t_1, t_2)} + \eta_{i,\text{SW}}^{(t_1, t_2)} - \frac{\Gamma-1}{\Gamma+1}|\bar\tau_{i,\text{SW}}^{(t_1, t_2)} + \eta_{i,\text{SW}}^{(t_1, t_2)}|\right)^2 \\
+& \pi_i(1-\pi_i)\left(\bar\tau_{i,\text{SW}}^{(t_1, t_2)} - \eta_{i,\text{SW}}^{(t_1, t_2)} - \frac{\Gamma-1}{\Gamma+1}|\bar\tau_{i,\text{SW}}^{(t_1, t_2)} - \eta_{i,\text{SW}}^{(t_1, t_2)}|\right)^2\\
-& 2\pi_i(1-\pi_i)\left(\bar\tau_{i,\text{SW}}^{(t_1, t_2)} + \eta_{i,\text{SW}}^{(t_1, t_2)} - \frac{\Gamma-1}{\Gamma+1}|\bar\tau_{i,\text{SW}}^{(t_1, t_2)} + \eta_{i,\text{SW}}^{(t_1, t_2)}|\right) \left(\bar\tau_{i,\text{SW}}^{(t_1, t_2)} - \eta_{i,\text{SW}}^{(t_1, t_2)} - \frac{\Gamma-1}{\Gamma+1}|\bar\tau_{i,\text{SW}}^{(t_1, t_2)} - \eta_{i,\text{SW}}^{(t_1, t_2)}|\right) \\
 =& \pi_i(1-\pi_i)\{2\eta_{i,\text{SW}}^{(t_1, t_2)}-\frac{\Gamma-1}{\Gamma+1}(|\bar\tau_{i,\text{SW}}^{(t_1, t_2)} + \eta_{i,\text{SW}}^{(t_1, t_2)} | -|\bar\tau_{i,\text{SW}}^{(t_1, t_2)} - \eta_{i,\text{SW}}^{(t_1, t_2)} |)\}^2 
\end{align*}

\noindent Under the sensitivity model $\pi_i \in [1/(1+\Gamma), \Gamma/(1+\Gamma)]$ with $\Gamma \ge 1$,  the product $\pi_i(1-\pi_i)$ is minimized at the endpoints, yielding
\[
\pi_i(1-\pi_i) \ge \frac{\Gamma}{\Gamma+1}\cdot\frac{1}{\Gamma+1}.
\]

\noindent Next, we bound the absolute-value term using the triangle inequality: for any real numbers $x$ and $y$,
\[
|x+y| - |x-y| \le |(x+y) - (x-y)| =  2|y|.
\]
Applying this with $x = \bar\tau_{i,\text{SW}}^{(t_1, t_2)}$ and $y = \eta_{i,\text{SW}}^{(t_1, t_2)}$ gives
\[
2\eta_{i,\text{SW}}^{(t_1, t_2)}-\frac{\Gamma-1}{\Gamma+1}(|\bar\tau_{i,\text{SW}}^{(t_1, t_2)} + \eta_{i,\text{SW}}^{(t_1, t_2)} | -|\bar\tau_{i,\text{SW}}^{(t_1, t_2)} - \eta_{i,\text{SW}}^{(t_1, t_2)} |)
\;\ge\; 
\begin{cases}
2\eta_{i,\text{SW}}^{(t_1, t_2)} \cdot \frac{2}{\Gamma+1}, & \eta_{i,\text{SW}}^{(t_1, t_2)} > 0,\\[2mm]
2\eta_{i,\text{SW}}^{(t_1, t_2)} \cdot \frac{2\Gamma}{\Gamma+1}, & \eta_{i,\text{SW}}^{(t_1, t_2)} < 0.
\end{cases}
\]
\noindent Since $\Gamma \ge 1$, in either case we have the conservative bound
\[
\{2\eta_{i,\text{SW}}^{(t_1, t_2)}-\frac{\Gamma-1}{\Gamma+1}(|\bar\tau_{i,\text{SW}}^{(t_1, t_2)} + \eta_{i,\text{SW}}^{(t_1, t_2)} | -|\bar\tau_{i,\text{SW}}^{(t_1, t_2)} - \eta_{i,\text{SW}}^{(t_1, t_2)} |)\}^2
\;\ge\;  \frac{16 \left(\eta_{i,\text{SW}}^{(t_1, t_2)}\right)^2}{(1+\Gamma)^2}.
\]
\noindent Combining with the bound on $\pi_i(1-\pi_i)$, we obtain
\begin{align*}
   \operatorname{Var}(D^{\mathrm{SW}}_{i,\Gamma}(\lambda_0)\mid T_{i1} = t_1, T_{i2} =t_2)  
  &\geq 16\frac{\Gamma}{\Gamma+1}\left(\frac{1}{\Gamma+1}\right)^3\left(\eta_{i,\text{SW}}^{(t_1, t_2)}\right)^2.
\end{align*}

\noindent By the law of total variance, we have
\begin{align*}
\operatorname{Var}(D^{\mathrm{SW}}_{i,\Gamma}(\lambda_0)) 
&= \mathrm{E}\!\left[ \operatorname{Var}(D^{\mathrm{SW}}_{i,\Gamma}(\lambda_0)\mid T_{i1}, T_{i2}) \right] 
  + \operatorname{Var}\!\left( \mathrm{E}[D^{\mathrm{SW}}_{i,\Gamma}(\lambda_0) \mid T_{i1}, T_{i2}] \right)\\
  &\geq \mathrm{E}\!\left[ \operatorname{Var}(D^{\mathrm{SW}}_{i,\Gamma}(\lambda_0)\mid T_{i1}, T_{i2}) \right] \\
  &\geq \mathrm{E}\!\left[ 16\frac{\Gamma}{\Gamma+1}\left(\frac{1}{\Gamma+1}\right)^3\left(\eta_{i,\text{SW}}^{(T_{i1}, T_{i2})}\right)^2 \right]\\
  &=16\frac{\Gamma}{\Gamma+1}\left(\frac{1}{\Gamma+1}\right)^3  \sum_{{t_1\in \{1,-1\}}} \sum_{{t_2\in \{1,-1\}}}\Pr\left\{(T_{i1}, T_{i2}) = (t_1, t_2)\right\}\left(\eta_{i,\text{SW}}^{(t_1, t_2)}\right)^2\\
  &=\frac{4\Gamma}{(\Gamma+1)^4} \sum_{{t_1\in \{1,-1\}}} \sum_{{t_2\in \{1,-1\}}}\left(\eta_{i,\text{SW}}^{(t_1, t_2)}\right)^2.
\end{align*}

\noindent Under the partly biased randomization scheme,
\begin{align*}
\mathrm{E}(D^{\mathrm{SW}}_{i,\Gamma}(\lambda_0))
&= \frac{\pi_i}{4}\sum_{k_1\neq k_1'}\sum_{k_2\neq k_2'}
   \Bigg[
      \big(y_{C^ai1k_1}-y_{T^ai1k_1'} + y_{T^bi2k_2}-y_{C^bi2k_2'}\big)  \\
&\hspace{4cm}
      -\frac{\Gamma-1}{\Gamma+1}
       \,\big|y_{C^ai1k_1}-y_{T^ai1k_1'} + y_{T^bi2k_2}-y_{C^bi2k_2'}\big|
   \Bigg] \\[6pt]
&\quad+ \frac{1-\pi_i}{4}\sum_{k_1\neq k_1'}\sum_{k_2\neq k_2'}
   \Bigg[
      \big(y_{C^ai2k_1}-y_{T^ai2k_1'} + y_{T^bi1k_2}-y_{C^bi1k_2'}\big) \\
&\hspace{4cm}
      -\frac{\Gamma-1}{\Gamma+1}
       \,\big|y_{C^ai2k_1}-y_{T^ai2k_1'} + y_{T^bi1k_2}-y_{C^bi1k_2'}\big|
   \Bigg].
\end{align*}
By H\"older’s inequality, for each stratum $i$ there exists $C_i<\infty$ such that
\begin{align*}
 \sum_{i=1}^I\mathrm{E}([D^{\mathrm{SW}}_{i,\Gamma}(\lambda_0)]^2)
& \leq \sum_{i=1}^I  C_{i}\sum_{j =1}^2 \sum_{k =1}^2 \left\{  y_{T^aijk}^2 + y_{C^aijk}^2 +  y_{T^bijk}^2 + y_{C^bijk}^2 \right\},
\end{align*}

\noindent Similarly,
\begin{align}\label{fourth moment bound}
\mathrm{E}([D^{\mathrm{SW}}_{i,\Gamma}(\lambda_0)]^4) &= \sum_{{t_1\in \{1,-1\}}} \sum_{{t_2\in \{1,-1\}}} \mathrm{E}[[D^{\mathrm{SW}}_{i,\Gamma}(\lambda_0)]^4\mid (T_{i1}, T_{i2}) = (t_1, t_2)] \cdot\Pr\left\{(T_{i1}, T_{i2}) = (t_1, t_2)\right\}\nonumber\\
&= \frac{1}{4} \sum_{{t_1\in \{1,-1\}}} \sum_{{t_2\in \{1,-1\}}} \mathrm{E}[[D^{\mathrm{SW}}_{i,\Gamma}(\lambda_0)]^4\mid (T_{i1}, T_{i2}) = (t_1, t_2)] \nonumber\\
&\leq \frac{1}{4} \sum_{{t_1\in \{1,-1\}}} \sum_{{t_2\in \{1,-1\}}}128(\frac{\Gamma}{1+\Gamma})^4\left( (\bar\tau_{i,\text{SW}}^{(t_1, t_2)})^4 + (\eta_{i,\text{SW}}^{(t_1, t_2)})^4\right)\\
& \leq  C'\sum_{j =1}^2 \sum_{k =1}^2 \left\{  y_{T^aijk}^4 + y_{C^aijk}^4 +  y_{T^bijk}^4 + y_{C^bijk}^4 \right\}  \nonumber
\end{align}
for some constant $C'<\infty$, where \eqref{fourth moment bound} follows from Lemma 5 of \cite{fogarty2020studentized}.

\noindent Therefore, under Conditions~\ref{con: bdd4thmoment} and \ref{con: bdd2ndmoment}, we have
$$\mathrm{E}(\overline{D}_{\Gamma, \mathrm{SW}}(\lambda_0))=O(1), \quad\lim\limits_{I\to \infty} I^{-2}\sum_{i=1}^I\mathrm{E}([D^{\mathrm{SW}}_{i,\Gamma}(\lambda_0)]^2)  =0, \quad\lim\limits_{I\to \infty} I^{-2}\sum_{i=1}^I\mathrm{E}([D^{\mathrm{SW}}_{i,\Gamma}(\lambda_0)]^4) =0, $$
and 
$$\lim\limits_{I\to \infty}I^{-1}\sum_{i=1}^I \operatorname{var}(D^{\mathrm{SW}}_{i,\Gamma}(\lambda_0)) \geq \frac{4\Gamma}{(\Gamma+1)^4} C. $$
Then Lemma \ref{lem:IS2_consistency} implies
\begin{equation}\label{SW: S2consistency}
IS^2_{\Gamma, \mathrm{SW}}(\lambda_0) -I\mathrm{E}(S^2_{\Gamma, \mathrm{SW}}(\lambda_0)) \;\xrightarrow{p}\; 0,
\end{equation}
and Lemma \ref{lemma: CLT} yields

\begin{equation*}\label{asymp_normal}
     \frac{\sqrt{I}\overline{D}_{\Gamma, \mathrm{SW}}(\lambda_0) -\mathrm{E}(\sqrt{I}\overline{D}_{\Gamma, \mathrm{SW}}(\lambda_0) )}{\sqrt{\operatorname{var}(\sqrt{I}\overline{D}_{\Gamma, \mathrm{SW}}(\lambda_0) )}}\xrightarrow{d}\mathcal{N}(0,1), \quad \text{as} ~I \rightarrow \infty.
 \end{equation*}

\noindent Along with the moment bounds from Step 2, if the treatment assignment satisfies (5) at $\Gamma$ and $H_0^{SW}$ holds, then
\begin{equation*}\label{CLT result}
    \lim\limits_{I\to \infty} \Pr\{\sqrt{I}\overline{D}_{\Gamma, \mathrm{SW}}(\lambda_0) \geq \Phi^{-1}(1-\alpha)\sqrt{I\mathrm{E}(S^2_{\Gamma, \mathrm{SW}}(\lambda_0))}\} \leq \alpha.
\end{equation*}

\medskip
\noindent This, in combination with \eqref{SW: S2consistency}, yields the  conclusion of the theorem.

\end{proof}

\subsection*{A.7 Proof of Theorem 4}
\begin{proof}

Let $\hat\tau^{(t_1, t_2)}_{i,\mathrm{ACO}}$ denote the value of $\hat\tau_{i, \mathrm{ACO}}(\kappa_0)$ given the realization $(T_{i1}, T_{i2})=(t_1, t_2)$, where $(t_1, t_2)\in \{(+1,+1), (+1,-1), (-1,+1), (-1,-1)\}$. For each realization, we decompose:
$$\hat\tau^{(t_1, t_2)}_{i,\mathrm{ACO}} = \bar\tau_{i,\mathrm{ACO}}^{(t_1, t_2)} + (2A_i - 1)\eta_{i,\mathrm{ACO}}^{(t_1, t_2)}$$
where $\eta_{i,\mathrm{ACO}}^{(t_1, t_2)}$ are defined in Section A.1, and the quantities $\hat\tau^{(t_1, t_2)}_{i,\mathrm{ACO}}$ and $\bar\tau_{i,\mathrm{ACO}}^{(t_1, t_2)}$ are specified for each case as follows:

---

\textit{Case 1:} $(t_1, t_2) = (1,-1)$
\begin{align*}
\hat\tau_{i,\mathrm{ACO}}^{(1,-1)}
&= A_i (y_{T^ai11}-y_{C^ai12})  
+ (1-A_i) (y_{T^ai22} - y_{C^ai21}), \\[6pt]
\bar\tau_{i,\mathrm{ACO}}^{(1,-1)} 
&= \tfrac{1}{2}\!\left[(y_{T^ai11} - y_{C^ai12})
+  (y_{T^ai22} - y_{C^ai21})\right]. 
\end{align*}

---

\textit{Case 2:} $(t_1, t_2) = (-1,-1)$
\begin{align*}
\hat\tau_{i,\mathrm{ACO}}^{(-1,-1)}
&= A_i(y_{T^ai12}-y_{C^ai11})
+ (1-A_i) (y_{T^ai22} - y_{C^ai21}), \\[6pt]
\bar\tau_{i,\mathrm{ACO}}^{(-1,-1)} 
&= \tfrac{1}{2}\!\left[ (y_{T^ai12} - y_{C^ai11})
+  (y_{T^ai22} - y_{C^ai21})\right]. 
\end{align*}

---

\textit{Case 3:} $(t_1, t_2) = (-1,1)$
\begin{align*}
\hat\tau_{i,\mathrm{ACO}}^{(-1,1)}
&= A_i\{ (y_{T^ai12}-y_{C^ai11})\} + (1-A_i)\{ (y_{T^ai21} - y_{C^ai22})\}, \\[6pt]
\bar\tau_{i,\mathrm{ACO}}^{(-1,1)} 
&= \tfrac{1}{2}\!\left[(y_{T^ai12} - y_{C^ai11})
+  (y_{T^ai21} - y_{C^ai22})\right]. 
\end{align*}

---

\textit{Case 4:} $(t_1, t_2) = (1,1)$
\begin{align*}
\hat\tau_{i,\mathrm{ACO}}^{(1,1)}
&= A_i (y_{T^ai11} - y_{C^ai12}) + (1-A_i) (y_{T^ai21} - y_{C^ai22}), \\[6pt]
\bar\tau_{i,\mathrm{ACO}}^{(1,1)} 
&= \tfrac{1}{2}\!\left[ (y_{T^ai11} - y_{C^ai12})
+  (y_{T^ai21} - y_{C^ai22})\right]. 
\end{align*}

\noindent Define
$$U_{i,\Gamma,\mathrm{ACO}}^{(t_1, t_2)} \equiv \bar\tau_{i,\mathrm{ACO}}^{(t_1, t_2)} + V_{i,\Gamma}|\eta_{i,\mathrm{ACO}}^{(t_1, t_2)}| - \left(\frac{\Gamma-1}{\Gamma+1}\right)|\bar\tau_{i,\mathrm{ACO}}^{(t_1, t_2)} + V_{i,\Gamma}|\eta_{i,\mathrm{ACO}}^{(t_1, t_2)}||$$
Let $U^{\mathrm{ACO}}_{i,\Gamma}$ be a random variable taking
\begin{align*}
\Pr(U^{\mathrm{ACO}}_{i,\Gamma} = U_{i,\Gamma,\mathrm{ACO}}^{(t_1, t_2)}) = \frac{1}{4}, \quad \text{for } (t_1, t_2) \in \{(+1,+1), (+1,-1), (-1,+1), (-1,-1)\}
\end{align*}

\noindent Following the same logic as in the proof of Theorem~5, $U^{\mathrm{ACO}}_{i,\Gamma}$ stochastically dominates $ D^{\mathrm{ACO}}_{i,\Gamma}(\kappa_0)$ for each $i$, i.e., for any scalar $c$:
$$\Pr(U^{\mathrm{ACO}}_{i,\Gamma} \geq c) \geq \Pr( D^{\mathrm{ACO}}_{i,\Gamma}(\kappa_0)\geq c)$$

\noindent To bound the expectation of $\overline{D}_{\Gamma, \mathrm{ACO}}(\kappa_0)$, it therefore suffices to bound the expectation of its stochastic upper bound $\overline U^{\mathrm{ACO}}_{i,\Gamma}$. 

\noindent For each $i$,
\begin{align*}
    \mathrm{E}(U_{i,\Gamma,\mathrm{ACO}}^{(t_1, t_2)})&= \bar\tau_{i,\mathrm{ACO}}^{(t_1, t_2)}+\frac{\Gamma-1}{1+\Gamma}|\eta_{i,\mathrm{ACO}}^{(t_1, t_2)}|\\&- \left(\frac{\Gamma-1}{\Gamma+1}\right) 
    \{ \frac{\Gamma}{\Gamma+1}|\bar\tau_{i,\mathrm{ACO}}^{(t_1, t_2)} + |\eta_{i,\mathrm{ACO}}^{(t_1, t_2)}||+ \frac{1}{1+\Gamma}|\bar\tau_{i,\mathrm{ACO}}^{(t_1, t_2)} -|\eta_{i,\mathrm{ACO}}^{(t_1, t_2)}||\}
\end{align*}
\noindent By Lemma~2 of \citet{fogarty2020studentized}, for $\Gamma \ge 1$, the inequality $$\frac{\Gamma}{\Gamma+1}|\bar\tau_{i,\mathrm{ACO}}^{(t_1, t_2)} + |\eta_{i,\mathrm{ACO}}^{(t_1, t_2)}||+ \frac{1}{1+\Gamma}|\bar\tau_{i,\mathrm{ACO}}^{(t_1, t_2)} -|\eta_{i,\mathrm{ACO}}^{(t_1, t_2)}|| \geq |\eta_{i,\mathrm{ACO}}^{(t_1, t_2)}| + \frac{\Gamma-1}{\Gamma+1}\bar\tau_{i,\mathrm{ACO}}^{(t_1, t_2)}$$ always holds. Hence, 
\begin{align*}
    \mathrm{E}(U_{i,\Gamma,\mathrm{ACO}}^{(t_1, t_2)})
    &\leq \bar\tau_{i,\mathrm{ACO}}^{(t_1, t_2)}+\frac{\Gamma-1}{1+\Gamma}|\eta_{i,\mathrm{ACO}}^{(t_1, t_2)}|- \left(\frac{\Gamma-1}{\Gamma+1}\right)\{|\eta_{i,\mathrm{ACO}}^{(t_1, t_2)}| + \frac{\Gamma-1}{\Gamma+1}\bar\tau_{i,\mathrm{ACO}}^{(t_1, t_2)}\}\\
    &=\left[1- \left(\frac{\Gamma-1}{\Gamma+1}\right)^2\right]\bar\tau_{i,\mathrm{ACO}}^{(t_1, t_2)}.
\end{align*}
Averaging over all four treatment assignment cases of $(T_{i1},T_{i2})$,
\begin{align*}
\mathrm{E}[U^{\mathrm{ACO}}_{i,\Gamma}] &= \frac{1}{4}\sum_{t_{1} \in \{1,-1\}} \sum_{t_{2} \in \{1,-1\}} \mathrm{E}[U_{i,\Gamma,\mathrm{ACO}}^{(t_1, t_2)}]\\
&\leq \frac{1}{4}\left[1- \left(\frac{\Gamma-1}{\Gamma+1}\right)^2\right]\sum_{t_{1} \in \{1,-1\}} \sum_{t_{2} \in \{1,-1\}} \bar\tau_{i,\mathrm{ACO}}^{(t_1, t_2)}
\end{align*}

\noindent Summing over all strata:
\begin{align*}
    \mathrm{E}\left[\sum_{i=1}^I U^{\mathrm{ACO}}_{i,\Gamma}\right] &\leq \frac{1}{4}\left[1- \left(\frac{\Gamma-1}{\Gamma+1}\right)^2\right] \sum_{i=1}^I \sum_{t_{1} \in \{1,-1\}} \sum_{t_{2} \in \{1,-1\}} \bar\tau_{i,\mathrm{ACO}}^{(t_1, t_2)} 
    \\&= \frac{1}{4}\left[1- \left(\frac{\Gamma-1}{\Gamma+1}\right)^2\right]\sum_{i = 1}^I\sum_{j = 1}^2 \sum_{k = 1}^2 (y_{T^aijk} - y_{C^aijk}),
\end{align*}
where the last equality follows from the definition of $\bar\tau_{i,\mathrm{ACO}}^{(t_1, t_2)}$.  

\noindent Under the null hypothesis $H^{\mathrm{ACO}}_0: \kappa = \kappa_0$, we have $$\sum_{i = 1}^I \sum_{j = 1}^2 \sum_{k = 1}^2 (y_{T^aijk} - y_{C^aijk})=0,$$ which implies
$\mathrm{E}\left[\sum_{i=1}^I U^{\mathrm{ACO}}_{i,\Gamma}\right] \leq 0$.
Since $D^{\mathrm{ACO}}_{i, \Gamma}(\kappa_0)$ is stochastically dominated by $U^{\mathrm{ACO}}_{i,\Gamma}$ for all $i$,  it follows that $$\mathrm{E}[\overline{D}_{\Gamma, \mathrm{ACO}}(\kappa_0) ] \leq 0 \quad \text{under} ~H^{\mathrm{ACO}}_0.$$  

\noindent The remaining steps follow directly from the argument used in the proof of Theorem~5.
\end{proof}

\newpage

\section*{Supplemental Material B: Additional simulation studies}

\subsection*{B.1 Additional simulation results for switchers}
\begin{table}[H]
\centering
\caption{Simulation results over 1000 replicates for varying strata sizes ($I$), proportions of switchers ($p$), and mean treatment effects among switchers ($\mu$), under a uniform distribution of treatment effects among switchers: $\tau^{SW}_{ijk} \sim \text{Unif}[\mu - \sqrt{3}, \mu + \sqrt{3}]$. 
}
\resizebox{\textwidth}{!}{
\begin{tabular}{ccccccccccccccc}
\midrule
 \multirow{2}{*}{I}  & \multirow{2}{*}{p} & \multirow{2}{*}{$\mu$} & \multirow{2}{*}{level} & \multirow{2}{*}{power} & \multirow{2}{*}{$\bar{\iota}_a$} & \multirow{2}{*}{$\bar{\iota}_b$} & \multirow{2}{*}{$\bar{\lambda}$} & \multicolumn{2}{c}{CI}  & 100$\times$  & 100$\times$  & 100$\times$  &  100$\times$  \\ 
  \cline{9-10} 
  &  &  &  &  &  & & & length & coverage ($\%$)&  SD[T($\lambda$)] & $S(\lambda)$ & SD[T(0)] &  $S(0)$ \\ 
 \midrule
  \multicolumn{14}{c}{\textbf{$\tau^{SW}_{ijk} \sim \text{Unif}[\mu - \sqrt{3}, \mu + \sqrt{3}]$}} \\
\midrule
  100 & 0.30 & 0.00 & 0.06 & 0.06 & 0.14 & 0.44 & 0.00 & 4.88 & 94.37 & 6.11 & 6.15 & 6.11 & 6.15 \\ 
  100 & 0.30 & 0.25 & 0.06 & 0.07 & 0.14 & 0.44 & 0.25 & 4.88 & 94.49 & 6.09 & 6.14 & 6.10 & 6.15 \\ 
  100 & 0.30 & 0.50 & 0.06 & 0.09 & 0.14 & 0.44 & 0.50 & 4.81 & 94.08 & 6.15 & 6.18 & 6.14 & 6.18 \\ 
  100 & 0.30 & 1.00 & 0.06 & 0.23 & 0.14 & 0.44 & 1.00 & 5.01 & 94.37 & 6.40 & 6.42 & 6.26 & 6.31 \\ 
  100 & 0.50 & 0.00 & 0.05 & 0.05 & 0.10 & 0.60 & -0.00 & 2.20 & 94.60 & 6.08 & 6.13 & 6.08 & 6.13 \\ 
  100 & 0.50 & 0.25 & 0.05 & 0.08 & 0.10 & 0.60 & 0.25 & 2.21 & 94.20 & 6.08 & 6.13 & 6.09 & 6.15 \\ 
  100 & 0.50 & 0.50 & 0.05 & 0.16 & 0.10 & 0.60 & 0.50 & 2.22 & 94.90 & 6.10 & 6.15 & 6.12 & 6.16 \\ 
  100 & 0.50 & 1.00 & 0.04 & 0.49 & 0.10 & 0.60 & 1.00 & 2.32 & 94.90 & 6.31 & 6.37 & 6.30 & 6.40 \\ 
  100 & 0.70 & 0.00 & 0.06 & 0.06 & 0.06 & 0.76 & -0.00 & 1.44 & 93.90 & 6.04 & 6.12 & 6.05 & 6.13 \\ 
  100 & 0.70 & 0.25 & 0.06 & 0.12 & 0.06 & 0.76 & 0.25 & 1.43 & 93.90 & 6.02 & 6.10 & 6.03 & 6.12 \\ 
  100 & 0.70 & 0.50 & 0.05 & 0.31 & 0.06 & 0.76 & 0.50 & 1.44 & 94.60 & 6.06 & 6.13 & 6.09 & 6.16 \\ 
  100 & 0.70 & 1.00 & 0.05 & 0.79 & 0.06 & 0.76 & 1.00 & 1.49 & 93.80 & 6.22 & 6.32 & 6.31 & 6.44 \\ 
  500 & 0.30 & 0.00 & 0.04 & 0.04 & 0.14 & 0.44 & 0.00 & 1.55 & 96.20 & 2.73 & 2.75 & 2.73 & 2.75 \\ 
  500 & 0.30 & 0.25 & 0.07 & 0.11 & 0.14 & 0.44 & 0.25 & 1.53 & 93.40 & 2.73 & 2.75 & 2.73 & 2.75 \\ 
  500 & 0.30 & 0.50 & 0.05 & 0.28 & 0.14 & 0.44 & 0.50 & 1.56 & 94.30 & 2.75 & 2.76 & 2.75 & 2.76 \\ 
  500 & 0.30 & 1.00 & 0.05 & 0.74 & 0.14 & 0.44 & 1.00 & 1.63 & 95.30 & 2.86 & 2.88 & 2.80 & 2.84 \\ 
  500 & 0.50 & 0.00 & 0.04 & 0.04 & 0.10 & 0.60 & -0.00 & 0.87 & 95.80 & 2.72 & 2.75 & 2.72 & 2.75 \\ 
  500 & 0.50 & 0.25 & 0.04 & 0.22 & 0.10 & 0.60 & 0.25 & 0.87 & 95.40 & 2.72 & 2.75 & 2.72 & 2.75 \\ 
  500 & 0.50 & 0.50 & 0.06 & 0.61 & 0.10 & 0.60 & 0.50 & 0.87 & 94.30 & 2.73 & 2.76 & 2.74 & 2.76 \\ 
  500 & 0.50 & 1.00 & 0.04 & 0.99 & 0.10 & 0.60 & 1.00 & 0.90 & 95.70 & 2.83 & 2.85 & 2.82 & 2.86 \\ 
  500 & 0.70 & 0.00 & 0.05 & 0.05 & 0.06 & 0.76 & -0.00 & 0.61 & 94.50 & 2.70 & 2.74 & 2.70 & 2.74 \\ 
  500 & 0.70 & 0.25 & 0.04 & 0.34 & 0.06 & 0.76 & 0.25 & 0.61 & 95.00 & 2.70 & 2.74 & 2.71 & 2.75 \\ 
  500 & 0.70 & 0.50 & 0.04 & 0.87 & 0.06 & 0.76 & 0.50 & 0.61 & 95.30 & 2.71 & 2.75 & 2.73 & 2.77 \\ 
  500 & 0.70 & 1.00 & 0.04 & 1.00 & 0.06 & 0.76 & 1.00 & 0.63 & 94.90 & 2.78 & 2.82 & 2.82 & 2.87 \\ 
  1000 & 0.30 & 0.00 & 0.05 & 0.05 & 0.14 & 0.44 & -0.00 & 1.04 & 94.40 & 1.93 & 1.94 & 1.93 & 1.94 \\ 
  1000 & 0.30 & 0.25 & 0.05 & 0.15 & 0.14 & 0.44 & 0.25 & 1.04 & 94.40 & 1.93 & 1.95 & 1.94 & 1.95 \\ 
  1000 & 0.30 & 0.50 & 0.04 & 0.49 & 0.14 & 0.44 & 0.50 & 1.05 & 96.30 & 1.95 & 1.96 & 1.94 & 1.96 \\ 
  1000 & 0.30 & 1.00 & 0.05 & 0.97 & 0.14 & 0.44 & 1.00 & 1.10 & 94.80 & 2.03 & 2.04 & 1.98 & 2.00 \\ 
  1000 & 0.50 & 0.00 & 0.05 & 0.06 & 0.10 & 0.60 & -0.00 & 0.61 & 94.10 & 1.92 & 1.94 & 1.92 & 1.94 \\ 
  1000 & 0.50 & 0.25 & 0.06 & 0.36 & 0.10 & 0.60 & 0.25 & 0.61 & 94.10 & 1.92 & 1.94 & 1.92 & 1.94 \\ 
  1000 & 0.50 & 0.50 & 0.04 & 0.89 & 0.10 & 0.60 & 0.50 & 0.61 & 94.40 & 1.93 & 1.95 & 1.94 & 1.96 \\ 
  1000 & 0.50 & 1.00 & 0.03 & 1.00 & 0.10 & 0.60 & 1.00 & 0.63 & 95.90 & 2.00 & 2.02 & 2.00 & 2.03 \\ 
  1000 & 0.70 & 0.00 & 0.04 & 0.04 & 0.06 & 0.76 & 0.00 & 0.43 & 95.30 & 1.91 & 1.94 & 1.91 & 1.94 \\ 
  1000 & 0.70 & 0.25 & 0.04 & 0.60 & 0.06 & 0.76 & 0.25 & 0.43 & 95.20 & 1.91 & 1.94 & 1.92 & 1.94 \\ 
  1000 & 0.70 & 0.50 & 0.04 & 0.99 & 0.06 & 0.76 & 0.50 & 0.43 & 95.30 & 1.92 & 1.95 & 1.93 & 1.96 \\ 
  1000 & 0.70 & 1.00 & 0.05 & 1.00 & 0.06 & 0.76 & 1.00 & 0.44 & 94.40 & 1.97 & 1.99 & 2.00 & 2.03 \\ 
   \hline
\end{tabular}
}
\end{table}

\begin{table}[H]
\centering
\caption{Simulation results over 1000 replicates for varying strata sizes ($I$), proportions of switchers ($p$), and mean treatment effects among switchers ($\mu$), under an exponential distribution of treatment effects among switchers: $\tau^{SW}_{ijk} \sim \text{Exp}(1/\mu)$. 
}
\resizebox{\textwidth}{!}{
\begin{tabular}{ccccccccccccccc}
\midrule
 \multirow{2}{*}{I}  & \multirow{2}{*}{p} & \multirow{2}{*}{$\mu$} & \multirow{2}{*}{level} & \multirow{2}{*}{power} & \multirow{2}{*}{$\bar{\iota}_a$} & \multirow{2}{*}{$\bar{\iota}_b$} & \multirow{2}{*}{$\bar{\lambda}$} & \multicolumn{2}{c}{CI}  & 100$\times$  & 100$\times$  & 100$\times$  &  100$\times$  \\ 
  \cline{9-10} 
  &  &  &  &  &  & & & length & coverage ($\%$)&  SD[T($\lambda$)] & $S(\lambda)$ & SD[T(0)] &  $S(0)$ \\ 
 \midrule
\multicolumn{14}{c}{\textbf{$\tau^{SW}_{ijk} \sim \text{Exp}(1/\mu)$}} \\
\midrule
100 & 0.30 & 0.00 & 0.06 & 0.06 & 0.14 & 0.44 & 0.00 & 4.68 & 94.68 & 5.83 & 5.84 & 5.83 & 5.84 \\ 
  100 & 0.30 & 0.25 & 0.05 & 0.08 & 0.14 & 0.44 & 0.25 & 4.78 & 94.58 & 5.85 & 5.84 & 5.86 & 5.85 \\ 
  100 & 0.30 & 0.50 & 0.05 & 0.10 & 0.14 & 0.44 & 0.50 & 4.73 & 94.88 & 5.95 & 5.93 & 5.94 & 5.93 \\ 
  100 & 0.30 & 1.00 & 0.06 & 0.22 & 0.14 & 0.44 & 1.00 & 4.95 & 94.87 & 6.43 & 6.44 & 6.29 & 6.34 \\ 
  100 & 0.50 & 0.00 & 0.05 & 0.05 & 0.10 & 0.60 & 0.00 & 2.03 & 94.60 & 5.61 & 5.60 & 5.61 & 5.60 \\ 
  100 & 0.50 & 0.25 & 0.07 & 0.10 & 0.10 & 0.60 & 0.25 & 2.04 & 92.70 & 5.64 & 5.64 & 5.65 & 5.65 \\ 
  100 & 0.50 & 0.50 & 0.06 & 0.20 & 0.10 & 0.60 & 0.50 & 2.11 & 93.60 & 5.76 & 5.76 & 5.77 & 5.78 \\ 
  100 & 0.50 & 1.00 & 0.05 & 0.52 & 0.10 & 0.60 & 1.00 & 2.30 & 95.10 & 6.30 & 6.35 & 6.28 & 6.37 \\ 
  100 & 0.70 & 0.00 & 0.07 & 0.07 & 0.06 & 0.76 & 0.00 & 1.27 & 93.50 & 5.38 & 5.37 & 5.38 & 5.37 \\ 
  100 & 0.70 & 0.25 & 0.06 & 0.12 & 0.06 & 0.76 & 0.25 & 1.28 & 93.80 & 5.41 & 5.40 & 5.43 & 5.42 \\ 
  100 & 0.70 & 0.50 & 0.06 & 0.35 & 0.06 & 0.76 & 0.50 & 1.31 & 94.10 & 5.57 & 5.58 & 5.61 & 5.62 \\ 
  100 & 0.70 & 1.00 & 0.05 & 0.78 & 0.06 & 0.76 & 1.00 & 1.48 & 95.10 & 6.22 & 6.29 & 6.31 & 6.40 \\ 
  500 & 0.30 & 0.00 & 0.05 & 0.05 & 0.14 & 0.44 & 0.00 & 1.48 & 95.40 & 2.62 & 2.62 & 2.62 & 2.62 \\ 
  500 & 0.30 & 0.25 & 0.05 & 0.12 & 0.14 & 0.44 & 0.25 & 1.47 & 94.90 & 2.61 & 2.61 & 2.62 & 2.62 \\ 
  500 & 0.30 & 0.50 & 0.03 & 0.28 & 0.14 & 0.44 & 0.50 & 1.48 & 96.50 & 2.66 & 2.66 & 2.65 & 2.66 \\ 
  500 & 0.30 & 1.00 & 0.05 & 0.76 & 0.14 & 0.44 & 1.00 & 1.62 & 93.80 & 2.86 & 2.88 & 2.81 & 2.83 \\ 
  500 & 0.50 & 0.00 & 0.06 & 0.06 & 0.10 & 0.60 & 0.00 & 0.79 & 93.90 & 2.51 & 2.51 & 2.51 & 2.51 \\ 
  500 & 0.50 & 0.25 & 0.04 & 0.23 & 0.10 & 0.60 & 0.25 & 0.80 & 95.80 & 2.52 & 2.52 & 2.53 & 2.53 \\ 
  500 & 0.50 & 0.50 & 0.05 & 0.69 & 0.10 & 0.60 & 0.50 & 0.82 & 95.10 & 2.58 & 2.58 & 2.59 & 2.59 \\ 
  500 & 0.50 & 1.00 & 0.04 & 0.99 & 0.10 & 0.60 & 1.00 & 0.91 & 95.90 & 2.83 & 2.85 & 2.82 & 2.86 \\ 
  500 & 0.70 & 0.00 & 0.04 & 0.04 & 0.06 & 0.76 & 0.00 & 0.54 & 95.90 & 2.41 & 2.40 & 2.41 & 2.40 \\ 
  500 & 0.70 & 0.25 & 0.04 & 0.45 & 0.06 & 0.76 & 0.25 & 0.54 & 94.90 & 2.42 & 2.42 & 2.43 & 2.43 \\ 
  500 & 0.70 & 0.50 & 0.05 & 0.94 & 0.06 & 0.76 & 0.50 & 0.56 & 94.90 & 2.49 & 2.50 & 2.51 & 2.52 \\ 
  500 & 0.70 & 1.00 & 0.05 & 1.00 & 0.06 & 0.76 & 1.00 & 0.63 & 94.40 & 2.78 & 2.82 & 2.82 & 2.87 \\ 
  1000 & 0.30 & 0.00 & 0.06 & 0.06 & 0.14 & 0.44 & 0.00 & 0.99 & 94.40 & 1.85 & 1.85 & 1.85 & 1.85 \\ 
  1000 & 0.30 & 0.25 & 0.06 & 0.19 & 0.14 & 0.44 & 0.25 & 0.99 & 94.20 & 1.85 & 1.85 & 1.85 & 1.85 \\ 
  1000 & 0.30 & 0.50 & 0.06 & 0.50 & 0.14 & 0.44 & 0.50 & 1.01 & 93.80 & 1.88 & 1.88 & 1.88 & 1.88 \\ 
  1000 & 0.30 & 1.00 & 0.05 & 0.97 & 0.14 & 0.44 & 1.00 & 1.09 & 94.90 & 2.03 & 2.04 & 1.98 & 2.00 \\ 
  1000 & 0.50 & 0.00 & 0.05 & 0.05 & 0.10 & 0.60 & 0.00 & 0.55 & 95.20 & 1.78 & 1.77 & 1.78 & 1.77 \\ 
  1000 & 0.50 & 0.25 & 0.06 & 0.43 & 0.10 & 0.60 & 0.25 & 0.55 & 93.90 & 1.78 & 1.78 & 1.79 & 1.79 \\ 
  1000 & 0.50 & 0.50 & 0.06 & 0.93 & 0.10 & 0.60 & 0.50 & 0.57 & 93.30 & 1.82 & 1.83 & 1.83 & 1.84 \\ 
  1000 & 0.50 & 1.00 & 0.05 & 1.00 & 0.10 & 0.60 & 1.00 & 0.63 & 94.70 & 2.00 & 2.02 & 1.99 & 2.03 \\ 
  1000 & 0.70 & 0.00 & 0.05 & 0.05 & 0.06 & 0.76 & 0.00 & 0.37 & 95.20 & 1.70 & 1.70 & 1.70 & 1.70 \\ 
  1000 & 0.70 & 0.25 & 0.05 & 0.71 & 0.06 & 0.76 & 0.25 & 0.38 & 94.60 & 1.71 & 1.72 & 1.72 & 1.72 \\ 
  1000 & 0.70 & 0.50 & 0.03 & 1.00 & 0.06 & 0.76 & 0.50 & 0.39 & 96.00 & 1.76 & 1.77 & 1.78 & 1.79 \\ 
  1000 & 0.70 & 1.00 & 0.06 & 1.00 & 0.06 & 0.76 & 1.00 & 0.44 & 93.20 & 1.96 & 1.99 & 1.99 & 2.03 \\ 
   \hline
\end{tabular}}
\end{table}

\begin{table}[H]
\centering
\caption{
Simulation results over 1,000 replicates varying stratum size ($I$) and constant treatment effect among switchers ($\mu$), with exactly two switchers per stratum.
}
\resizebox{\textwidth}{!}{
\begin{tabular}{cccccccccccccc}
\midrule
 \multirow{2}{*}{I}  & \multirow{2}{*}{$\mu$} & \multirow{2}{*}{level} & \multirow{2}{*}{power} & \multirow{2}{*}{$\bar{\iota}_a$} & \multirow{2}{*}{$\bar{\iota}_b$} & \multirow{2}{*}{$\bar{\lambda}$} & \multicolumn{2}{c}{CI}  & 100$\times$  & 100$\times$  & 100$\times$  &  100$\times$  \\ 
  \cline{8-9} 
  &  &  &  &  & & & length & coverage ($\%$)&  SD[T($\lambda$)] & $S(\lambda)$ & SD[T(0)] &  $S(0)$ \\ 
\midrule
100 & 0.00 & 0.05 & 0.05 & 0.10 & 0.60 & 0.00 & 2.07 & 94.90 & 5.61 & 5.61 & 5.61 & 5.61 \\ 
  100 & 0.25 & 0.06 & 0.10 & 0.10 & 0.60 & 0.25 & 2.05 & 94.00 & 5.59 & 5.58 & 5.61 & 5.60 \\ 
  100 & 0.50 & 0.06 & 0.20 & 0.10 & 0.60 & 0.50 & 2.09 & 93.70 & 5.66 & 5.66 & 5.69 & 5.69 \\ 
  100 & 1.00 & 0.07 & 0.57 & 0.10 & 0.60 & 1.00 & 2.13 & 93.10 & 5.90 & 5.89 & 5.93 & 5.91 \\ 
  500 & 0.00 & 0.05 & 0.05 & 0.10 & 0.60 & 0.00 & 0.80 & 94.80 & 2.51 & 2.52 & 2.51 & 2.52 \\ 
  500 & 0.25 & 0.05 & 0.26 & 0.10 & 0.60 & 0.25 & 0.80 & 94.90 & 2.51 & 2.51 & 2.52 & 2.51 \\ 
  500 & 0.50 & 0.04 & 0.68 & 0.10 & 0.60 & 0.50 & 0.81 & 95.60 & 2.53 & 2.53 & 2.54 & 2.54 \\ 
  500 & 1.00 & 0.04 & 1.00 & 0.10 & 0.60 & 1.00 & 0.84 & 95.80 & 2.64 & 2.64 & 2.66 & 2.66 \\ 
  1000 & 0.00 & 0.04 & 0.04 & 0.10 & 0.60 & 0.00 & 0.56 & 96.00 & 1.78 & 1.78 & 1.78 & 1.78 \\ 
  1000 & 0.25 & 0.05 & 0.42 & 0.10 & 0.60 & 0.25 & 0.55 & 94.80 & 1.77 & 1.77 & 1.78 & 1.78 \\ 
  1000 & 0.50 & 0.05 & 0.94 & 0.10 & 0.60 & 0.50 & 0.56 & 94.80 & 1.79 & 1.79 & 1.80 & 1.80 \\ 
  1000 & 1.00 & 0.05 & 1.00 & 0.10 & 0.60 & 1.00 & 0.58 & 94.60 & 1.87 & 1.87 & 1.88 & 1.88 \\
  \midrule
\end{tabular}
}
\end{table}

\subsection*{B.2  Simulation details for always-compliers}

\subsubsection*{B.2.1 Data-generating process}
\label{subsubsec: simu rand DGP ACO}
We generate data for a PoP-NIV design consisting of $I$ strata. As in the switcher-focused DGP, we vary the following factors:

\begin{description}
    \item[Factor 1:] Number of matched strata $I$: $100$, $500$, and $1000$.
\end{description}

For each unit $k$ in matched pair $j$ within stratum $i \in [I]$, its principal stratum membership $S_{ijk}$ is sampled from a multinomial distribution with the following probabilities:
\begin{equation*}
\begin{split}
    &P(S_{ijk} = \text{ACO}) = p; \\
    &P(S_{ijk} \in \{\text{SW, AT-NT, AAT, NT-AT, ANT}\}) = (1-p)/5.
\end{split}
\end{equation*}

\begin{description}
    \item[Factor 2:] Proportion of always-compliers $p:$ $0.3$, $0.5$, and $0.7$.
\end{description}

We generate potential outcomes $(r_{d=0, ijk}, r_{d=1, ijk})$ for each unit as:
$$r_{d=0, ijk} \sim \mathcal{N}(0,1), \quad r_{d=1, ijk} = r_{d=0, ijk} + \tau_{ijk},$$
where the treatment effect $\tau_{ijk}$ depends on the unit’s principal stratum. Specifically, 
\begin{itemize}
\item $\tau_{ijk} \sim \mathcal{N}(0.5,1)$ for switcher strata,
    \item $\tau_{ijk} \sim \mathcal{N}(0.1,1)$ for othrer non-ACO strata (AT-NT, AAT, NT-AT, ANT),
    \item $\tau_{ijk}$ is drawn from one of the following distributions for always-compliers:
\end{itemize}

\begin{description}
    \item[Factor 3:] Treatment effect among always-compliers, $\tau^{ACO}_{ijk}$ follows one of the three distributions:
    \begin{itemize}
        \item[(i)] $\tau^{ACO}_{ijk} \sim \text{Unif}~[\mu - \sqrt{3}, \mu + \sqrt{3}]$,
        \item[(ii)] $\tau^{ACO}_{ijk} \sim \mathcal{N}(\mu, 1)$,
        \item[(iii)] $\tau^{ACO}_{ijk} \sim \text{Exp}(1/\mu)$.
    \end{itemize}
\end{description}

Here, $\mu$ controls the magnitude of the treatment effect among always-compliers, varied as:
\begin{description}
    \item[Factor 4:] Effect size among always-compliers $\mu$: $0$, $0.25$, $0.5$, and $1$.
\end{description}

After generating all potential outcomes, including treatment assignments $(d_{T^a ijk}, d_{C^a ijk}, d_{T^b ijk}, d_{C^b ijk})$ based on $S_{ijk}$, we generate the IV assignment vector $\boldsymbol{Z_i}$ for each stratum independently, according to the randomization scheme described in Assumption~1. Finally, we determine observed treatment and outcomes:
\begin{equation*}
\begin{split}
    &D_{ijk} = d_{C^aijk}\mathbbm{I}\{{Z}_{ijk}=0_a\} + d_{T^aijk}\mathbbm{I}\{{Z}_{ijk}=1_a\}+d_{C^bijk}\mathbbm{I}\{{Z}_{ijk}=0_b\} + d_{T^bijk}\mathbbm{I}\{{Z}_{ijk}=1_b\}, \\
    &R_{ijk} = r_{d=1, ijk}\mathbbm{I}\{D_{ijk}=1\} + r_{d=0, ijk}\mathbbm{I}\{D_{ijk}=0\}.
\end{split}
\end{equation*}

In total, Factors 1--4 define $108$ data-generating processes. For each DGP, we generated $1000$ datasets and recorded the true sample average treatment effect $\kappa_{\text{true}}$ among always-compliers in each dataset. We then evaluated the level and power of the randomization-based test from Theorem~1 by testing: (i) $H_0: \kappa = \kappa_{\text{true}}$ and (ii) $H_0: \kappa = 0$. Confidence intervals for $\kappa_{\text{true}}$ were obtained by inverting these tests.

\subsubsection*{B.2.2 Simulation results}
\begin{table}[H]
\centering
\caption{Simulation results over 1000 replicates for varying strata sizes ($I$), proportions of always-compliers ($p$), and mean treatment effects among always-compliers ($\mu$), under a normal distribution of treatment effects among always-compliers: $\tau^{ACO}_{ijk} \sim \mathcal{N}(\mu, 1)$. 
}
\begin{tabular}{ccccccccccccccc}
\midrule
\multirow{2}{*}{I}  & \multirow{2}{*}{p} & \multirow{2}{*}{$\mu$} & \multirow{2}{*}{level} & \multirow{2}{*}{power} & \multirow{2}{*}{$\bar{\iota}_a$} & \multirow{2}{*}{$\bar{\iota}_b$} & \multirow{2}{*}{$\bar{\kappa}$} & \multicolumn{2}{c}{CI}   \\ 
  &  &  &  &  &  & & & length & coverage ($\%$)\\ 
 \midrule
  \multicolumn{10}{c}{\textbf{$\tau^{ACO}_{ijk} \sim \mathcal{N}(\mu, 1)$}} \\
\midrule
100 & 0.30 & 0.00 & 0.04 & 0.04 & 0.30 & 0.44 & -0.00 & 2.94 & 95.40 \\ 
  100 & 0.30 & 0.25 & 0.05 & 0.08 & 0.30 & 0.44 & 0.25 & 2.84 & 94.60 \\ 
  100 & 0.30 & 0.50 & 0.05 & 0.15 & 0.30 & 0.44 & 0.50 & 2.90 & 94.20 \\ 
  100 & 0.30 & 1.00 & 0.05 & 0.38 & 0.30 & 0.44 & 1.00 & 3.06 & 94.80 \\ 
  100 & 0.50 & 0.00 & 0.05 & 0.05 & 0.50 & 0.60 & 0.00 & 1.43 & 94.60 \\ 
  100 & 0.50 & 0.25 & 0.04 & 0.12 & 0.50 & 0.60 & 0.25 & 1.43 & 96.20 \\ 
  100 & 0.50 & 0.50 & 0.06 & 0.29 & 0.50 & 0.60 & 0.50 & 1.45 & 94.00 \\ 
  100 & 0.50 & 1.00 & 0.05 & 0.79 & 0.50 & 0.60 & 1.00 & 1.49 & 94.90 \\ 
  100 & 0.70 & 0.00 & 0.04 & 0.04 & 0.70 & 0.76 & 0.00 & 0.98 & 95.70 \\ 
  100 & 0.70 & 0.25 & 0.04 & 0.18 & 0.70 & 0.76 & 0.25 & 0.98 & 95.00 \\ 
  100 & 0.70 & 0.50 & 0.05 & 0.53 & 0.70 & 0.76 & 0.50 & 0.98 & 94.60 \\ 
  100 & 0.70 & 1.00 & 0.05 & 0.98 & 0.70 & 0.76 & 1.00 & 1.01 & 95.40 \\ 
  500 & 0.30 & 0.00 & 0.04 & 0.05 & 0.30 & 0.44 & 0.00 & 1.05 & 95.20 \\ 
  500 & 0.30 & 0.25 & 0.05 & 0.18 & 0.30 & 0.44 & 0.25 & 1.04 & 93.70 \\ 
  500 & 0.30 & 0.50 & 0.05 & 0.48 & 0.30 & 0.44 & 0.50 & 1.05 & 94.50 \\ 
  500 & 0.30 & 1.00 & 0.05 & 0.96 & 0.30 & 0.44 & 1.00 & 1.10 & 94.60 \\ 
  500 & 0.50 & 0.00 & 0.04 & 0.04 & 0.50 & 0.60 & -0.00 & 0.61 & 95.40 \\ 
  500 & 0.50 & 0.25 & 0.03 & 0.40 & 0.50 & 0.60 & 0.25 & 0.60 & 95.60 \\ 
  500 & 0.50 & 0.50 & 0.04 & 0.88 & 0.50 & 0.60 & 0.50 & 0.61 & 94.90 \\ 
  500 & 0.50 & 1.00 & 0.05 & 1.00 & 0.50 & 0.60 & 1.00 & 0.63 & 94.00 \\ 
  500 & 0.70 & 0.00 & 0.04 & 0.04 & 0.70 & 0.76 & 0.00 & 0.43 & 95.30 \\ 
  500 & 0.70 & 0.25 & 0.05 & 0.62 & 0.70 & 0.76 & 0.25 & 0.43 & 94.10 \\ 
  500 & 0.70 & 0.50 & 0.05 & 1.00 & 0.70 & 0.76 & 0.50 & 0.43 & 93.60 \\ 
  500 & 0.70 & 1.00 & 0.05 & 1.00 & 0.70 & 0.76 & 1.00 & 0.44 & 93.40 \\ 
  1000 & 0.30 & 0.00 & 0.05 & 0.05 & 0.30 & 0.44 & 0.00 & 0.72 & 94.30 \\ 
  1000 & 0.30 & 0.25 & 0.05 & 0.28 & 0.30 & 0.44 & 0.25 & 0.72 & 94.00 \\ 
  1000 & 0.30 & 0.50 & 0.05 & 0.77 & 0.30 & 0.44 & 0.50 & 0.73 & 94.90 \\ 
  1000 & 0.30 & 1.00 & 0.05 & 1.00 & 0.30 & 0.44 & 1.00 & 0.76 & 94.30 \\ 
  1000 & 0.50 & 0.00 & 0.05 & 0.06 & 0.50 & 0.60 & -0.00 & 0.42 & 93.50 \\ 
  1000 & 0.50 & 0.25 & 0.04 & 0.63 & 0.50 & 0.60 & 0.25 & 0.42 & 94.00 \\ 
  1000 & 0.50 & 0.50 & 0.04 & 0.99 & 0.50 & 0.60 & 0.50 & 0.42 & 94.60 \\ 
  1000 & 0.50 & 1.00 & 0.05 & 1.00 & 0.50 & 0.60 & 1.00 & 0.44 & 94.90 \\ 
  1000 & 0.70 & 0.00 & 0.03 & 0.04 & 0.70 & 0.76 & 0.00 & 0.30 & 95.60 \\ 
  1000 & 0.70 & 0.25 & 0.05 & 0.87 & 0.70 & 0.76 & 0.25 & 0.30 & 94.30 \\ 
  1000 & 0.70 & 0.50 & 0.04 & 1.00 & 0.70 & 0.76 & 0.50 & 0.30 & 95.40 \\ 
  1000 & 0.70 & 1.00 & 0.04 & 1.00 & 0.70 & 0.76 & 1.00 & 0.31 & 94.90 \\ 
 \hline
\end{tabular}
\end{table}

\begin{table}[H]
\centering
\caption{Simulation results over 1000 replicates for varying strata sizes ($I$), proportions of always-compliers ($p$), and mean treatment effects among always-compliers ($\mu$), under a uniform distribution of treatment effects among always-compliers: $\tau^{ACO}_{ijk} \sim \text{Unif}[\mu - \sqrt{3}, \mu + \sqrt{3}]$. 
}
\begin{tabular}{ccccccccccccccc}
\midrule
\multirow{2}{*}{I}  & \multirow{2}{*}{p} & \multirow{2}{*}{$\mu$} & \multirow{2}{*}{level} & \multirow{2}{*}{power} & \multirow{2}{*}{$\bar{\iota}_a$} & \multirow{2}{*}{$\bar{\iota}_b$} & \multirow{2}{*}{$\bar{\kappa}$} & \multicolumn{2}{c}{CI}   \\ 
  &  &  &  &  &  & & & length & coverage ($\%$)\\ 
 \midrule
  \multicolumn{10}{c}{\textbf{$\tau^{ACO}_{ijk} \sim \text{Unif}[\mu - \sqrt{3}, \mu + \sqrt{3}]$}} \\
\midrule
100 & 0.30 & 0.00 & 0.06 & 0.06 & 0.30 & 0.44 & -0.00 & 3.00 & 93.50 \\ 
  100 & 0.30 & 0.25 & 0.05 & 0.07 & 0.30 & 0.44 & 0.25 & 2.98 & 94.10 \\ 
  100 & 0.30 & 0.50 & 0.05 & 0.13 & 0.30 & 0.44 & 0.50 & 2.89 & 94.70 \\ 
  100 & 0.30 & 1.00 & 0.05 & 0.40 & 0.30 & 0.44 & 1.00 & 3.07 & 94.50 \\ 
  100 & 0.50 & 0.00 & 0.05 & 0.05 & 0.50 & 0.60 & 0.00 & 1.43 & 94.70 \\ 
  100 & 0.50 & 0.25 & 0.05 & 0.12 & 0.50 & 0.60 & 0.25 & 1.43 & 94.50 \\ 
  100 & 0.50 & 0.50 & 0.06 & 0.33 & 0.50 & 0.60 & 0.50 & 1.44 & 94.80 \\ 
  100 & 0.50 & 1.00 & 0.04 & 0.79 & 0.50 & 0.60 & 1.00 & 1.50 & 93.90 \\ 
  100 & 0.70 & 0.00 & 0.04 & 0.05 & 0.70 & 0.76 & -0.00 & 0.98 & 95.10 \\ 
  100 & 0.70 & 0.25 & 0.05 & 0.15 & 0.70 & 0.76 & 0.25 & 0.98 & 94.30 \\ 
  100 & 0.70 & 0.50 & 0.04 & 0.55 & 0.70 & 0.76 & 0.50 & 0.98 & 93.70 \\ 
  100 & 0.70 & 1.00 & 0.05 & 0.97 & 0.70 & 0.76 & 1.00 & 1.01 & 95.00 \\ 
  500 & 0.30 & 0.00 & 0.06 & 0.05 & 0.30 & 0.44 & -0.00 & 1.04 & 94.50 \\ 
  500 & 0.30 & 0.25 & 0.05 & 0.15 & 0.30 & 0.44 & 0.25 & 1.04 & 94.10 \\ 
  500 & 0.30 & 0.50 & 0.05 & 0.48 & 0.30 & 0.44 & 0.50 & 1.06 & 94.80 \\ 
  500 & 0.30 & 1.00 & 0.06 & 0.97 & 0.30 & 0.44 & 1.00 & 1.10 & 94.40 \\ 
  500 & 0.50 & 0.00 & 0.04 & 0.04 & 0.50 & 0.60 & -0.00 & 0.60 & 95.20 \\ 
  500 & 0.50 & 0.25 & 0.04 & 0.37 & 0.50 & 0.60 & 0.25 & 0.60 & 94.60 \\ 
  500 & 0.50 & 0.50 & 0.03 & 0.91 & 0.50 & 0.60 & 0.50 & 0.61 & 95.60 \\ 
  500 & 0.50 & 1.00 & 0.05 & 1.00 & 0.50 & 0.60 & 1.00 & 0.63 & 93.20 \\ 
  500 & 0.70 & 0.00 & 0.04 & 0.06 & 0.70 & 0.76 & 0.00 & 0.43 & 94.00 \\ 
  500 & 0.70 & 0.25 & 0.04 & 0.60 & 0.70 & 0.76 & 0.25 & 0.43 & 95.20 \\ 
  500 & 0.70 & 0.50 & 0.05 & 0.99 & 0.70 & 0.76 & 0.50 & 0.43 & 94.00 \\ 
  500 & 0.70 & 1.00 & 0.05 & 1.00 & 0.70 & 0.76 & 1.00 & 0.44 & 93.60 \\ 
  1000 & 0.30 & 0.00 & 0.06 & 0.06 & 0.30 & 0.44 & -0.00 & 0.72 & 93.90 \\ 
  1000 & 0.30 & 0.25 & 0.04 & 0.26 & 0.30 & 0.44 & 0.25 & 0.72 & 95.10 \\ 
  1000 & 0.30 & 0.50 & 0.05 & 0.79 & 0.30 & 0.44 & 0.50 & 0.73 & 94.30 \\ 
  1000 & 0.30 & 1.00 & 0.06 & 1.00 & 0.30 & 0.44 & 1.00 & 0.76 & 94.10 \\ 
  1000 & 0.50 & 0.00 & 0.05 & 0.06 & 0.50 & 0.60 & -0.00 & 0.42 & 94.10 \\ 
  1000 & 0.50 & 0.25 & 0.04 & 0.62 & 0.50 & 0.60 & 0.25 & 0.42 & 94.50 \\ 
  1000 & 0.50 & 0.50 & 0.04 & 0.99 & 0.50 & 0.60 & 0.50 & 0.42 & 93.90 \\ 
  1000 & 0.50 & 1.00 & 0.05 & 1.00 & 0.50 & 0.60 & 1.00 & 0.44 & 93.80 \\ 
  1000 & 0.70 & 0.00 & 0.04 & 0.05 & 0.70 & 0.76 & 0.00 & 0.30 & 94.30 \\ 
  1000 & 0.70 & 0.25 & 0.04 & 0.89 & 0.70 & 0.76 & 0.25 & 0.30 & 94.30 \\ 
  1000 & 0.70 & 0.50 & 0.04 & 1.00 & 0.70 & 0.76 & 0.50 & 0.30 & 94.50 \\ 
  1000 & 0.70 & 1.00 & 0.05 & 1.00 & 0.70 & 0.76 & 1.00 & 0.31 & 94.10 \\ 
 \hline
\end{tabular}
\end{table}

\begin{table}[H]
\centering
\caption{Simulation results over 1000 replicates for varying strata sizes ($I$), proportions of always-compliers ($p$), and mean treatment effects among always-compliers ($\mu$), under an exponential distribution of treatment effects among always-compliers: $\tau^{ACO}_{ijk} \sim \text{Exp}(1/\mu)$.
}
\begin{tabular}{ccccccccccccccc}
\midrule
\multirow{2}{*}{I}  & \multirow{2}{*}{p} & \multirow{2}{*}{$\mu$} & \multirow{2}{*}{level} & \multirow{2}{*}{power} & \multirow{2}{*}{$\bar{\iota}_a$} & \multirow{2}{*}{$\bar{\iota}_b$} & \multirow{2}{*}{$\bar{\kappa}$} & \multicolumn{2}{c}{CI}   \\ 
  &  &  &  &  &  & & & length & coverage ($\%$)\\ 
 \midrule
  \multicolumn{10}{c}{\textbf{$\tau^{ACO}_{ijk} \sim \text{Exp}(1/\mu)$}} \\
\midrule
100 & 0.30 & 0.00 & 0.06 & 0.06 & 0.30 & 0.44 & 0.00 & 2.78 & 94.10 \\ 
  100 & 0.30 & 0.25 & 0.05 & 0.08 & 0.30 & 0.44 & 0.25 & 2.84 & 94.69 \\ 
  100 & 0.30 & 0.50 & 0.06 & 0.16 & 0.30 & 0.44 & 0.50 & 2.97 & 94.10 \\ 
  100 & 0.30 & 1.00 & 0.05 & 0.40 & 0.30 & 0.44 & 1.00 & 3.08 & 94.80 \\ 
  100 & 0.50 & 0.00 & 0.05 & 0.05 & 0.50 & 0.60 & 0.00 & 1.31 & 95.20 \\ 
  100 & 0.50 & 0.25 & 0.05 & 0.12 & 0.50 & 0.60 & 0.25 & 1.32 & 94.80 \\ 
  100 & 0.50 & 0.50 & 0.06 & 0.34 & 0.50 & 0.60 & 0.50 & 1.34 & 94.10 \\ 
  100 & 0.50 & 1.00 & 0.05 & 0.81 & 0.50 & 0.60 & 1.00 & 1.48 & 93.80 \\ 
  100 & 0.70 & 0.00 & 0.04 & 0.04 & 0.70 & 0.76 & 0.00 & 0.86 & 96.00 \\ 
  100 & 0.70 & 0.25 & 0.05 & 0.20 & 0.70 & 0.76 & 0.25 & 0.86 & 95.20 \\ 
  100 & 0.70 & 0.50 & 0.05 & 0.62 & 0.70 & 0.76 & 0.50 & 0.89 & 94.10 \\ 
  100 & 0.70 & 1.00 & 0.05 & 0.98 & 0.70 & 0.76 & 1.00 & 1.00 & 94.20 \\ 
  500 & 0.30 & 0.00 & 0.04 & 0.04 & 0.30 & 0.44 & 0.00 & 0.99 & 96.50 \\ 
  500 & 0.30 & 0.25 & 0.04 & 0.17 & 0.30 & 0.44 & 0.25 & 0.99 & 95.40 \\ 
  500 & 0.30 & 0.50 & 0.05 & 0.51 & 0.30 & 0.44 & 0.50 & 1.01 & 94.80 \\ 
  500 & 0.30 & 1.00 & 0.04 & 0.97 & 0.30 & 0.44 & 1.00 & 1.09 & 94.70 \\ 
  500 & 0.50 & 0.00 & 0.05 & 0.05 & 0.50 & 0.60 & 0.00 & 0.55 & 94.80 \\ 
  500 & 0.50 & 0.25 & 0.05 & 0.43 & 0.50 & 0.60 & 0.25 & 0.56 & 94.30 \\ 
  500 & 0.50 & 0.50 & 0.05 & 0.93 & 0.50 & 0.60 & 0.50 & 0.57 & 95.00 \\ 
  500 & 0.50 & 1.00 & 0.05 & 1.00 & 0.50 & 0.60 & 1.00 & 0.63 & 95.00 \\ 
  500 & 0.70 & 0.00 & 0.06 & 0.06 & 0.70 & 0.76 & 0.00 & 0.37 & 94.20 \\ 
  500 & 0.70 & 0.25 & 0.05 & 0.74 & 0.70 & 0.76 & 0.25 & 0.37 & 94.90 \\ 
  500 & 0.70 & 0.50 & 0.05 & 1.00 & 0.70 & 0.76 & 0.50 & 0.39 & 94.80 \\ 
  500 & 0.70 & 1.00 & 0.05 & 1.00 & 0.70 & 0.76 & 1.00 & 0.44 & 94.40 \\ 
  1000 & 0.30 & 0.00 & 0.04 & 0.04 & 0.30 & 0.44 & 0.00 & 0.69 & 95.80 \\ 
  1000 & 0.30 & 0.25 & 0.05 & 0.28 & 0.30 & 0.44 & 0.25 & 0.69 & 94.00 \\ 
  1000 & 0.30 & 0.50 & 0.05 & 0.80 & 0.30 & 0.44 & 0.50 & 0.70 & 95.20 \\ 
  1000 & 0.30 & 1.00 & 0.04 & 1.00 & 0.30 & 0.44 & 1.00 & 0.76 & 95.80 \\ 
  1000 & 0.50 & 0.00 & 0.04 & 0.04 & 0.50 & 0.60 & 0.00 & 0.39 & 95.50 \\ 
  1000 & 0.50 & 0.25 & 0.06 & 0.69 & 0.50 & 0.60 & 0.25 & 0.39 & 93.80 \\ 
  1000 & 0.50 & 0.50 & 0.05 & 1.00 & 0.50 & 0.60 & 0.50 & 0.40 & 94.10 \\ 
  1000 & 0.50 & 1.00 & 0.04 & 1.00 & 0.50 & 0.60 & 1.00 & 0.44 & 95.30 \\ 
  1000 & 0.70 & 0.00 & 0.04 & 0.04 & 0.70 & 0.76 & 0.00 & 0.26 & 95.70 \\ 
  1000 & 0.70 & 0.25 & 0.06 & 0.94 & 0.70 & 0.76 & 0.25 & 0.26 & 93.40 \\ 
  1000 & 0.70 & 0.50 & 0.04 & 1.00 & 0.70 & 0.76 & 0.50 & 0.27 & 95.50 \\ 
  1000 & 0.70 & 1.00 & 0.05 & 1.00 & 0.70 & 0.76 & 1.00 & 0.31 & 93.50 \\ 
 \hline
\end{tabular}
\end{table}

\clearpage
\subsection*{B.3 Simulation results for biased randomization inference}

\begin{table}[H]
\centering
\begin{tabular}{cccccccc}
  \toprule
& &   \multicolumn{4}{c}{Panel A} & \multicolumn{2}{c}{Panel B}\\ \cmidrule(lr){3-6}
$I$ & $\Gamma$ & \multicolumn{2}{c}{$\pi_{i} 
\sim \text{Unif} \left[\frac{1}{2}, \frac{\Gamma}{\Gamma + 1} \right]$} & \multicolumn{2}{c}{$\pi_i = \frac{\Gamma}{\Gamma + 1}$} & \multicolumn{2}{c}{$\pi_i = \frac{\Gamma}{\Gamma + 1}$} \\  \cmidrule(lr){3-4} \cmidrule(lr){5-6} \cmidrule(lr){7-8} 
    &          & Theorem 5 & Theorem 2 & Theorem 5 & Theorem 2 & Theorem 5 & Theorem 2\\ 
  \hline
100 & 1.1 & 0.016 & 0.048 & 0.020 & 0.072 & 0.021 & 0.080 \\ 
  100 & 1.2 & 0.011 & 0.043 & 0.024 & 0.118 & 0.027 & 0.142 \\ 
  100 & 1.3 & 0.005 & 0.104 & 0.012 & 0.181 & 0.033 & 0.247 \\ 
  100 & 1.4 & 0.002 & 0.106 & 0.007 & 0.261 & 0.024 & 0.348 \\ 
  100 & 1.5 & 0.003 & 0.107 & 0.004 & 0.356 & 0.028 & 0.479 \\ 
  100 & 1.6 & 0.001 & 0.148 & 0.004 & 0.437 & 0.038 & 0.584 \\ 
  100 & 1.7 & 0.000 & 0.159 & 0.002 & 0.529 & 0.032 & 0.672 \\ 
  100 & 1.8 & 0.000 & 0.185 & 0.003 & 0.605 & 0.023 & 0.748 \\ 
  100 & 1.9 & 0.000 & 0.206 & 0.000 & 0.686 & 0.030 & 0.860 \\ 
  100 & 2.0 & 0.000 & 0.257 & 0.001 & 0.753 & 0.030 & 0.903 \\ 
  500 & 1.1 & 0.002 & 0.065 & 0.004 & 0.142 & 0.025 & 0.148 \\ 
  500 & 1.2 & 0.000 & 0.129 & 0.000 & 0.340 & 0.022 & 0.465 \\ 
  500 & 1.3 & 0.001 & 0.214 & 0.001 & 0.620 & 0.020 & 0.794 \\ 
  500 & 1.4 & 0.000 & 0.316 & 0.001 & 0.818 & 0.029 & 0.934 \\ 
  500 & 1.5 & 0.000 & 0.419 & 0.000 & 0.934 & 0.026 & 0.988 \\ 
  500 & 1.6 & 0.000 & 0.523 & 0.000 & 0.975 & 0.022 & 0.999 \\ 
  500 & 1.7 & 0.000 & 0.647 & 0.000 & 0.994 & 0.032 & 1.000 \\ 
  500 & 1.8 & 0.000 & 0.677 & 0.000 & 0.997 & 0.020 & 1.000 \\ 
  500 & 1.9 & 0.000 & 0.761 & 0.000 & 0.999 & 0.025 & 1.000 \\ 
  500 & 2.0 & 0.000 & 0.827 & 0.000 & 1.000 & 0.030 & 1.000 \\ 
  1000 & 1.1 & 0.001 & 0.088 & 0.001 & 0.224 & 0.020 & 0.290 \\ 
  1000 & 1.2 & 0.000 & 0.219 & 0.000 & 0.598 & 0.021 & 0.751 \\ 
  1000 & 1.3 & 0.000 & 0.343 & 0.000 & 0.890 & 0.021 & 0.975 \\ 
  1000 & 1.4 & 0.000 & 0.527 & 0.000 & 0.971 & 0.033 & 0.998 \\ 
  1000 & 1.5 & 0.000 & 0.705 & 0.000 & 0.998 & 0.023 & 1.000 \\ 
  1000 & 1.6 & 0.000 & 0.808 & 0.000 & 1.000 & 0.020 & 1.000 \\ 
  1000 & 1.7 & 0.000 & 0.891 & 0.000 & 1.000 & 0.018 & 1.000 \\ 
  1000 & 1.8 & 0.000 & 0.939 & 0.000 & 1.000 & 0.025 & 1.000 \\ 
  1000 & 1.9 & 0.000 & 0.977 & 0.000 & 1.000 & 0.019 & 1.000 \\ 
  1000 & 2.0 & 0.000 & 0.974 & 0.000 & 1.000 & 0.022 & 1.000 \\ 
   \toprule
\end{tabular}
\caption{Simulation results for the case where IV dose assignments are generated under a partly biased randomization scheme $\mathcal{M}_{\Gamma}$: (I) $\pi_{i} 
\sim \text{Unif} \left[\frac{1}{2}, \frac{\Gamma}{\Gamma + 1} \right]$; and (II)  $\pi_i = \frac{\Gamma}{\Gamma + 1}$.
We report the test level under Theorem 2 and 5 across $1,000$ simulations.} 
\end{table}

\clearpage
\section*{Supplemental Material C: Additional details on the case study}

\begin{table}[ht]
\centering
\caption{Covariate balance before matching. Mean (SD) are reported for continuous variables. Count (percentage) are reported for categorical variables.}
\label{tb: PLCO before match}
\resizebox{\textwidth}{!}{
\begin{tabular}{lllllll}
\toprule
& \multicolumn{2}{c}{Prior to 1997 ($G= a$)} & \multicolumn{2}{c}{After 1997 ($G= b$)}& &\\
 \toprule
 & Control arm & Treatment arm & Control arm & Treatment arm & $p$-value &  \\ 
 \toprule
Sample size &  4210 &  4204 &  4970 &  4978 &  &  \\ 
  Treatment uptake (=Screening) (\%) &     0 ( 0.0)  &  2141 (50.9)  &     0 ( 0.0)  &  3989 (80.1)  & $<$0.001 &  \\ 
  Age (\%) &   &   &   &   & $<$0.001 &  \\ 
  \quad $\leq 60$ &  1031 (24.5)  &  1049 (25.0)  &  2930 (59.0)  &  2915 (58.6)  &  &  \\ 
   \quad $(60,65]$ &  1308 (31.1)  &  1291 (30.7)  &  1092 (22.0)  &  1096 (22.0)  &  &  \\ 
   \quad $(65, 70]$ &  1190 (28.3)  &  1188 (28.3)  &   650 (13.1)  &   663 (13.3)  &  &  \\ 
 \quad  $>70$ &   681 (16.2)  &   676 (16.1)  &   298 ( 6.0)  &   304 ( 6.1)  &  &  \\ 
  Age & 64.75 (5.09) & 64.69 (5.08) & 60.40 (5.26) & 60.42 (5.31) & $<$0.001 &  \\ 
  Sex (=Male) (\%) &  1568 (37.2)  &  1566 (37.3)  &  2160 (43.5)  &  2180 (43.8)  & $<$0.001 &  \\ 
  Race (=Minority) (\%) &   864 (20.5)  &   915 (21.8)  &   704 (14.2)  &   700 (14.1)  & $<$0.001 &  \\ 
  Education (\%)&   &   &   &   & $<$0.001 &  \\ 
 \quad No high school &   585 (13.9)  &   616 (14.7)  &   351 ( 7.1)  &   369 ( 7.4)  &  &  \\ 
 \quad  High school &  1775 (42.2)  &  1639 (39.0)  &  1586 (31.9)  &  1589 (31.9)  &  &  \\ 
 \quad  College or above  &  1850 (43.9)  &  1949 (46.4)  &  3033 (61.0)  &  3020 (60.7)  &  &  \\ 
 Smoking status (\%)&   &   &   &   &  0.663 &  \\ 
 \quad  Non-smoker &  1803 (42.8)  &  1798 (42.8)  &  2101 (42.3)  &  2076 (41.7)  &  &  \\ 
 \quad  Current smoker &   562 (13.3)  &   573 (13.6)  &   700 (14.1)  &   729 (14.6)  &  &  \\ 
 \quad  Former smoker &  1845 (43.8)  &  1833 (43.6)  &  2169 (43.6)  &  2173 (43.7)  &  &  \\ 
  BMI  ($>25$) (\%) &  2771 (65.8)  &  2758 (65.6)  &  3532 (71.1)  &  3492 (70.1)  & $<$0.001 &  \\ 
  Colorectal cancer (=Confirmed cancer) (\%)&    82 ( 1.9)  &    65 ( 1.5)  &    65 ( 1.3)  &    58 ( 1.2)  &  0.012 &  \\ 
 \toprule
\end{tabular}}
\end{table}